\documentclass[copyright,creativecommons,noncommercial,noderivs]{eptcs}
\providecommand{\publicationstatus}{To appear in \emph{Information and Computation}, 2019.\\
  doi:\href{https://doi.org/10.1016/j.ic.2019.104435}{10.1016/j.ic.2019.104435}}

\newfont{\fsc}{eusm10 scaled 1100}
\usepackage{amsmath} 
\usepackage{amssymb}
\usepackage{epsfig}
\usepackage{color}
\usepackage[normalem]{ulem}
\usepackage{graphicx}
\usepackage{tikz}
\usetikzlibrary{positioning}
\usetikzlibrary{decorations.pathreplacing}
\usetikzlibrary{arrows}
\newtheorem{theo}{Theorem}
\newtheorem{defi}[theo]{Definition}
\newtheorem{prop}[theo]{Proposition}
\newtheorem{lemm}[theo]{Lemma}
\newtheorem{coro}[theo]{Corollary}
\newtheorem{obs}[theo]{Observation}
\newtheorem{exam}[theo]{Example}
\newtheorem{rema}[theo]{Remark}

\newenvironment{definition}{\begin{defi} \rm }{\end{defi}}
\newenvironment{theorem}{\begin{theo} \rm }{\end{theo}}
\newenvironment{proposition}{\begin{prop} \rm }{\end{prop}}
\newenvironment{lemma}{\begin{lemm} \rm }{\end{lemm}}
\newenvironment{corollary}{\begin{coro} \rm }{\end{coro}}
\newenvironment{observation}{\begin{obs} \rm }{\end{obs}}
\newenvironment{example}{\begin{exam} \rm }{\end{exam}}

\newenvironment{proof}{\begin{trivlist} \item[\hspace{\labelsep}\bf Proof:]}{\hfill $\Box$\end{trivlist}}

\newcommand{\brel}{\mathrel{\mathcal{B}}}
\newcommand{\rrel}{\mathrel{\mathcal{R}}}
\newcommand{\AF}{\mbox{$\cal A\!\!F$}}
\newcommand{\AFO}{\mathcal{A\!\!F\!\!O}}
\DeclareMathAlphabet{\mathbbm}{U}{bbm}{m}{n}
\newcommand{\N}{\mathbbm{N}}
\newcommand{\G}{{\rm G}}
\renewcommand{\H}{{\rm H}}
\newcommand{\K}{{\rm K}}

\def\titlerunning{Divide and Congruence III: Stability and Divergence}
\def\authorrunning{W.J. Fokkink, R.J. van Glabbeek \& B. Luttik}
\title{Divide and Congruence III: From Decomposition of
  Modal Formulas to Preservation of Stability and Divergence}

\author{Wan Fokkink%
\institute{Department of Computer Science\\
Vrije University Amsterdam, The Netherlands}
\email{w.j.fokkink@vu.nl}
\and
Rob van Glabbeek
\institute{Data61, CSIRO, Sydney, Australia}
\institute{School of Computer Science and Engineering\\
University of New South Wales, Sydney, Australia}
\email{rvg@cs.stanford.edu}
   \and Bas Luttik
\institute{Department of Mathematics and Computer Science,
            Technische Universiteit Eindhoven,
            The Netherlands}
\email{s.p.luttik@tue.nl}}
\begin{document}
\maketitle

\newcommand{\plat}[1]{\raisebox{0pt}[0pt][0pt]{#1}} 
\newcommand{\diam}[1]{\langle#1\rangle}
\newcommand{\eps}{\langle\epsilon\rangle}
\newcommand{\trans}[1]{\stackrel{#1}{\longrightarrow}}
\newcommand{\ntrans}[1]{\mathord{\hspace{.4em}\not\hspace{-.4em}\trans{#1\;}}}
\newcommand{\epsarrow}{\stackrel{\epsilon}{\Longrightarrow}}
\newcommand{\bis}[2][]{\mathrel{\,\raisebox{.3ex}{$\underline{\makebox[.7em]{$\leftrightarrow$}}$}\,_{#2}^{\,#1}}}
\newcommand{\notbis}[2][]{\,\not\mathrel{\raisebox{.3ex}{$\underline{\makebox[.7em]{$\leftrightarrow$}}$}\,_{#2}^{#1}\,}}
\newcommand{\var}{{\it var}} \newcommand{\ar}{{\it ar}}
\newcommand{\IO}[1]{\mathbb{O}_{#1}} \renewcommand{\phi}{\varphi}
\newcommand{\qed}{\hfill$\Box$} \newcommand{\Red}[1]{\underline{#1}}
\newcounter{saveenumi} \newcommand{\aL}{\aleph\mathord\cap\Lambda}

\begin{abstract}
  In two earlier papers we derived congruence formats with regard to
  transition system specifications for weak semantics on the basis of
  a decomposition method for modal formulas. The idea is that a
  congruence format for a semantics must ensure that the formulas in
  the modal characterisation of this semantics are always decomposed
  into formulas that are again in this modal characterisation. The
  stability and divergence requirements that are imposed on many of
  the known weak semantics have so far been outside the realm of this
  method. Stability refers to the absence of a $\tau$-transition. We
  show, using the decomposition method, how congruence formats can be
  relaxed for weak semantics that are stability-respecting. This
  relaxation for instance brings the priority operator within the
  range of the stability-respecting branching bisimulation
  format. Divergence, which refers to the presence of an infinite
  sequence of $\tau$-transitions, escapes the inductive decomposition
  method. We circumvent this problem by proving that a congruence
  format for a stability-respecting weak semantics is also a
  congruence format for its divergence-preserving counterpart.
\end{abstract}

\vspace{4mm}

\noindent {\bf 1998 ACM Subject Classification:} F.3.2 Operational
semantics, F.4.1 Modal logic\vspace{2mm}


\section{Introduction}\label{sec:introduction}

Structural operational semantics \cite{Plo04} provides specification
languages with an interpretation in terms of a mathematical notion
of behaviour. It generates a labelled transition
system, in which states are the closed terms over a (single-sorted,
first-order) signature, and transitions between states carry
labels. The transitions between states are obtained from a transition
system specification (TSS), which consists of a set of proof rules
called transition rules. States in labelled transition systems can be
identified by a wide range of behavioural equivalences, based on e.g.\
branching structure or decorated versions of execution sequences. {\sc
  Van Glabbeek} \cite{vGl93} classified so-called weak semantics,
which take into account the internal action $\tau$. A significant
number of the weak semantics based on a bisimulation relation carry a
stability or divergence requirement. Stability refers to the absence
of a $\tau$-transition and divergence to the presence of an infinite
sequence of $\tau$-transitions.

In general a behavioural equivalence induced by a TSS is not
guaranteed to be a congruence, i.e.\ the equivalence class of a term
$f(p_1,\ldots,p_n)$ need not be determined by $f$ and the equivalence
classes of its arguments $p_1,\ldots,p_n$. Being a congruence is an
important property, for instance in order to fit the equivalence into
an axiomatic framework. Respecting stability or preserving divergence
sometimes needs to be imposed in order to obtain a congruence
relation, for example in case of the priority operator \cite{BBK86}.

Behavioural equivalences can be characterised in terms of the
observations that an experimenter could make during a session with a
process. Modal logic captures such observations. A modal
characterisation of an equivalence on processes consists of a class
$C$ of modal formulas such that two processes are equivalent if and
only if they satisfy the same formulas in $C$. For instance,
Hennessy-Milner logic \cite{HM85} constitutes a modal characterisation
of (strong) bisimilarity. A cornerstone for the current paper is the
work in \cite{BFvG04} to decompose formulas from Hennessy-Milner logic
with respect to a structural operational semantics in the ntyft format
\cite{Gro93} without lookahead.  Here the \emph{decomposition} of a
modal formula $\phi$ w.r.t.\ a term $t=f(p_1,\ldots,p_n)$ is a
selection of $n$-tuples of modal formulas, one of which needs to be
satisfied by the processes $p_i$ in order for $t$ to satisfy $\phi$.
Based on this method, congruence formats can be generated for process
semantics from their modal characterisation, to ensure that the
equivalence is a congruence. Such formats help to avoid repetitive
congruence proofs and obtain insight into the congruence property. Key
idea is that congruence is ensured if formulas from the modal
characterisation of a semantics under consideration are always
decomposed into formulas that are again in this modal
characterisation. This approach was extended to weak semantics in
\cite{FvGdW12,FvG17}. It crosses the borders between process algebra,
structural operational semantics, process semantics, and modal logic.

Here we expand the latter work to weak semantics that respect
stability or preserve divergence.  We focus on \emph{branching
  bisimilarity} and \emph{rooted branching bisimilarity} \cite{GlWe96}
and consider for each a stability-respecting and two
divergence-preserving variants.  Divergence-preserving branching
bisimilarity \cite{GlWe96} is the coarsest congruence relation for the
parallel composition operator that only equates processes satisfying
the same formulas from the well-known temporal logic CTL$^\ast$ minus
the next-time operator $X$ \cite{GLT09b}.  With regard to stability
the expansion is relatively straightforward: we extend the modal
characterisation of the semantics with one clause to capture that a
semantics is stability-respecting, and study the decomposition of this
additional clause. Next we show how the congruence formats for
branching bisimilarity and rooted branching bisimilarity from
\cite{FvGdW12} can be relaxed, owing to the extended modal
characterisation for stability-respecting branching
bisimilarity. Notably, the transition rules for the priority operator
are within the more relaxed formats.

The divergence preservation property escapes the inductive
decomposition method, as it concerns an infinite sequence of
$\tau$-transitions. We overcome this problem by presenting a general
framework for lifting congruence formats from some weak semantics
$\approx$ to a finer weak semantics $\sim$, with the aim to turn
congruence formats for stability-respecting weak semantics into
congruence formats for their divergence-preserving counterparts.
(For the curious reader, it
roughly works as follows. Consider a TSS $P$ in a congruence format
for $\approx$, and for example a unary function symbol $f$. A new TSS
is created out of $P$, in which many occurrences of $\tau$ in
transition rules are suppressed by renaming them into some fresh
action name. Furthermore, processes are provided with what we call an
oracle transition that reveals some pertinent information on the
behaviour of the process, such as whether it can diverge.  In this way
we ensure that $\approx$ and $\sim$ coincide on the transformed
TSS\@. For each closed term $p$ of $P$ we moreover introduce a
constant $\hat p$ in the transformed TSS such that $p \sim q$ implies
$\hat p \sim \hat q$, and so $\hat p \approx \hat q$. We take care
that this entire transformation preserves the congruence format for
$\approx$. So $\hat p \approx \hat q$ implies
$f(\hat p) \approx f(\hat q)$, and hence, since $\approx$ and $\sim$
coincide on the transformed TSS,
$f(\hat p) \sim f(\hat q)$.\linebreak[4] Finally we cast $f(\hat p)$
and $f(\hat q)$ back to processes strongly bisimilar to $f(p)$ and
$f(q)$ in the original TSS $P$, taking care that $\sim$ is
preserved. Thus we have gone full circle: $p\sim q$ implies
$f(p)\sim f(q)$. This implies that the congruence format for $\approx$
is also a congruence format for $\sim$.) We show four instances where
this method can be applied.
In two cases $\approx$ is
stability-respecting branching bisimilarity and in two cases rooted
stability-respecting branching bisimilarity, while the definition of
$\sim$ is obtained from $\approx$ by replacing respect for stability
by preservation of one of the two aforementioned forms of divergence. Thus it can be
concluded that the congruence format for stability-respecting
branching bisimilarity is also applicable to divergence-preserving as
well as weakly divergence-preserving branching bisimilarity; and
likewise for the rooted counterparts of these semantics.

In \cite{BLY17}, it was argued that an adaptation of the classical operational
semantics of sequential composition in a process theory with the empty
process leads to a smoother integration of classical automata theory.
A detailed proof is given that rooted divergence-preserving branching
bisimilarity is a congruence for the resulting operator---here called
sequencing---, but also it is observed that this property can be inferred
from our congruence format.

An extended abstract of the present paper appeared as \cite{FvGL17}. This version includes elaborate proofs of the results,
provides additional explanations, and discusses sequencing as an extra example application.

\section{Preliminaries}

This section recalls the basic notions of labelled transition systems
and weak semantics and defines stability-respecting and
divergence-preserving branching bisimilarity
(Sect.~\ref{sec:equivalences_terms}) as well as a modal
characterisation of stability-respecting branching bisimilarity
(Sect.~\ref{sec:modal}). It also presents a brief introduction to
structural operational semantics (Sect.~\ref{sec:sos}) and recalls
some syntactic restrictions on transition rules
(Sect.~\ref{sec:ntytt}).

\subsection{Stability-respecting and divergence-preserving branching
  bisimilarity}\label{sec:equivalences_terms}

A \emph{labelled transition system (LTS)} is a triple
$(\mathbb{P},{\it Act},\rightarrow)$, with $\mathbb{P}$ a set of
\emph{processes}, $Act$ a set of \emph{actions}, and
$\rightarrow\;\subseteq\mathbb{P}\times Act \times\mathbb{P}$.  We
normally let $Act = A \cup\{\tau\}$ where $\tau$ is an \emph{internal
  action} and $A$ some set of external or observable actions not
containing $\tau$.  We write $A_\tau$ for $A\cup\{\tau\}$.  We use
$p,q$ to denote processes, $\alpha,\beta,\gamma$ for elements of
$A_\tau$, and $a,b$ for elements of $A$. We write
\plat{$p\trans\alpha q$} for $(p,\alpha,q)\in{\rightarrow}$,
\plat{$p{\trans\alpha}$} for
\plat{$\exists q\in\mathbb{P}:p\trans\alpha q$}, and
\plat{$p\ntrans\alpha$} for $\neg(p{\trans\alpha})$. Furthermore,
\plat{$\epsarrow$} denotes the transitive-reflexive closure of
\plat{$\trans{\tau}$}. If \plat{$p\ntrans\tau$}, then we say that
  $p$ is \emph{stable}.

\begin{definition}\label{def:bb}
  Let ${\brel}\subseteq\mathbb{P}\times\mathbb{P}$ be a symmetric
  relation.
  \begin{itemize}
  \item $\brel$ is a \emph{branching bisimulation} if $p\brel q$ and
    \plat{$p\trans{\alpha}p'$} implies that either $\alpha = \tau$ and
    $p'\,\brel\,q$, or\\ \plat{$q\epsarrow q' \trans{\alpha} q''$} for
    some $q'$ and $q''$ with $p\brel q'$ and $p'\brel q''$.
  \item $\brel$ is \emph{stability-respecting} if $p\brel q$ and
    $p \ntrans\tau$ implies that $q \epsarrow q' \ntrans\tau$ for some
    $q'$ with $p \brel q'$.
  \item $\brel$ is \emph{divergence-preserving} if it satisfies the
    following condition:
    \begin{enumerate}\itemsep 0pt
      \renewcommand{\theenumi}{D}
      \renewcommand{\labelenumi}{(\theenumi)}
    \item \label{cnd:rvgdivsim} if $p\brel q$ and there is an infinite
      sequence of processes $(p_k)_{k\in\N}$ such that $p=p_0$,
      $p_k\trans{\tau}p_{k+1}$ and $p_k\brel q$ for all $k\in\N$, then
      there exists an infinite sequence of processes
      $(q_\ell)_{\ell\in\N}$ such that $q=q_0$,
      \plat{$q_\ell\trans{\tau}q_{\ell+1}$} for all $\ell\in\N$, and
      $p_k\brel q_\ell$ for all $k,\ell\in\N$.
    \end{enumerate}
    The definition of a \emph{weakly divergence-preserving} relation
    is obtained by omitting the condition ``and $p_k\brel q_\ell$ for
    all $k,\ell\in\N$''. (The condition ``and $p_k\brel q$ for all
    $k\in\N$'' is then redundant.)
  \end{itemize}
  Processes $p,q$ are {\em branching bisimilar}, denoted $p\bis{b}q$,
  if there exists a branching bisimulation $\brel$ with $p\brel q$.
  They are {\em stability-respecting}, \emph{divergence-preserving} or
  \emph{weakly divergence-preserving} branching bisimilar, denoted
  $p\bis[s]{b}q$, $p\bis[\Delta]{b}q$ or $p\bis[\Delta\top\!]{b}q$, if
  moreover $\brel$ is stability-respecting, divergence-preserving or
  weakly divergence-preserving, respectively.
\end{definition}

\noindent
We have
${\bis{b}}\supset{\bis[s]{b}}\supset{\bis[\Delta\top\!]{b}}\supset{\bis[\Delta]{b}}$.

The relations $\bis{b}$, $\bis[s]{b}$, $\bis[\Delta]{b}$ and
$\bis[\Delta\top\!]{b}$ are equivalences
\cite{Bas96,vGl93,GLT09a}. However, they are not {\em congruences}
with respect to most process algebras from the literature, meaning
that the equivalence class of a process $f(p_1,\ldots,p_n)$, with $f$
an $n$-ary function symbol, is not always determined by the
equivalence classes of its arguments, i.e.\ the processes
$p_1,\ldots,p_n$. Therefore an additional rootedness condition is
imposed.

\begin{definition}\label{def:rbb}
  \emph{Rooted} branching bisimilarity, $\bis{rb}$, is the largest
  symmetric relation on $\mathbb{P}$ such that $p \bis{rb} q$ and
  \plat{$p \trans{\alpha} p'$} implies that
  \plat{$q \trans{\alpha}q '$} for some $q'$ with $p'\bis{b}q'$.
  Likewise, \emph{rooted} stability-respecting, divergence-preserving
  or weakly divergence-preserving branching bisimilarity, denoted by
  $p\bis[s]{rb}q$,\linebreak $p\bis[\Delta]{rb}q$ or
  $p\bis[\Delta\top\!]{rb}q$, is the largest symmetric relation
  ${\cal R}$ on $\mathbb{P}$ such that $p \mathrel{\cal R} q$ and
  \plat{$p \trans{\alpha} p'$} implies that
  \plat{$q \trans{\alpha} q'$} for some $q'$ with $p'\bis[s]{b}q'$,
  $p'\bis[\Delta]{b}q'$ or $p'\bis[\Delta\top\!]{b}q'$, respectively.
\end{definition}

\noindent
Our main aim is to develop congruence formats for stability-respecting
and divergence-preserving semantics.  These congruence formats will
impose syntactic restrictions on the transition rules that are used to
generate the underlying LTS. The congruence formats will be determined
using the characterising modal logics for the semantics.

We will sometimes refer to strong bisimulation semantics, which
ignores the special status of the $\tau$.
\begin{definition}\label{def:strong-b}
  A symmetric relation ${\brel}\subseteq\mathbb{P}\times\mathbb{P}$ is
  a \emph{strong bisimulation} if $p\brel q$ and \plat{$p\trans{a}p'$}
  implies that \plat{$q\trans{a} q'$} for some $q'$ with $p'\brel q'$.
  Processes $p,q$ are {\em strongly bisimilar}, denoted $p\bis{} q$,
  if there exists a strong bisimulation $\brel$ with $p\brel q$.
\end{definition}

\noindent
The various notions of bisimilarity defined above are examples of
so-called behavioural equivalences. For a general formulation of the
results in Sect.~\ref{sec:dpbbcong}, it is convenient to formally
define a notion of behavioural equivalence that includes at least the
examples above. Note that a common feature of their definitions is
that they associate with every LTS
$\G=(\mathbb{P}_\G,{\it Act}_\G,\rightarrow_\G)$ a binary relation
$\sim_\G$. (For instance, in the case of strong bisimilarity, the
relation $\sim_\G$ associated with $\G$ is defined as the binary
relation ${\bis{}}\subseteq\mathbb{P}_\G\times\mathbb{P}_G$ such that
$p\bis{} q$ if (and only if) there exists a strong bisimulation
${\brel}\subseteq\mathbb{P}_\G\times\mathbb{P}_G$ such that
$p\brel{} q$.)  One may, thus, think of a behavioural equivalence as a
family of binary relations indexed by LTSs. It turns out that we need
to impose just one extra condition on such families to arrive at a
suitable formalisation of the notion of behavioural equivalence. The
condition states that the relation associated with the disjoint union
of two LTSs restricted to one of the components coincides with the
relation associated with that component.

\begin{definition}\label{def:disjoint}
  Two LTSs $\G=(\mathbb{P}_\G,{\it Act}_\G,\rightarrow_\G)$ and
  $\H=(\mathbb{P}_{\,\H},{\it Act}_{\,\H},\rightarrow_{\,\H})$ are
  called \emph{disjoint} if
  $\mathbb{P}_\G \cap \mathbb{P}_{\,\H} = \emptyset$. In that case
  $\G\uplus\H$ denotes their union
  $(\mathbb{P}_\G\cup \mathbb{P}_{\,\H},{\it Act}_\G\cup {\it
    Act}_{\,\H},\rightarrow_\G \cup \rightarrow_{\,\H})$.
\end{definition}
A \emph{behavioural equivalence} $\sim$ on LTSs is a family of
equivalence relations $\sim_\G$, one for every LTS $\G$, such that for
each pair of disjoint LTSs
$\G=(\mathbb{P}_\G,{\it Act}_\G,\rightarrow_\G)$ and $\H$ we have
$g \sim_\G g' \Leftrightarrow g \sim_{\G\uplus \H} g'$ for any
$g,g'\in \mathbb{P}_\G$.

The notions of bisimilarity defined above clearly qualify as
behavioural equivalences. Given two behavioural equivalences $\sim$
and $\approx$, we write ${\sim} \subseteq {\approx}$ iff
${\sim_\G} \subseteq {\approx_\G}$ for each LTS $\G$.

\subsection{Modal characterisation}\label{sec:modal}

Modal logic formulas express properties on the behaviour of processes
in an LTS\@. Following \cite{vGl93}, we extend Hennessy-Milner logic
\cite{HM85} with the modal connectives $\eps\phi$ and
$\diam{\hat\tau}\phi$, expressing that a process can perform zero or
more, respectively zero or one, $\tau$-transitions to a process where
$\phi$ holds.

\begin{definition}\label{def:formulas}
  The class $\mathbb{O}$ of \emph{modal formulas} is defined as
  follows, where $I$ ranges over all index sets and $\alpha$ over
  $A_\tau$:
  \[
    \mathbb{O} \hspace*{2cm} \phi ~~::=~~
    {\displaystyle\bigwedge_{i\in I}}\,\phi_i ~~|~~ \neg\phi ~~|~~
    \diam{\alpha}\phi ~~|~~ \eps\phi ~~|~~ \diam{\hat\tau}\phi
  \]
  We use abbreviations $\top$ for the empty conjunction,
  $\phi_1\land\phi_2$ for $\bigwedge_{i\in\{1,2\}}\phi_i$,
  $\phi\diam{\alpha}\phi'$ for $\phi\land\diam{\alpha}\phi'$, and
  $\phi\diam{\hat\tau}\phi'$ for $\phi\land\diam{\hat\tau}\phi'$.
\end{definition}

\noindent
$p\models\phi$ denotes that process $p$ satisfies formula $\phi$. The
first two operators represent the standard Boolean operators
conjunction and negation. We define that
$p\models\diam{\alpha}\phi$ if \plat{$p\trans{\alpha}p'$} for some
$p'$ with $p'\models\phi$, $p\models\eps\phi$ if
\plat{$p\epsarrow p'$} for some $p'$ with $p'\models\phi$, and
$p\models\diam{\hat\tau}\phi$ if either $p\models\phi$ or
\plat{$p\trans{\tau}p'$} for some $p'$ with $p'\models\phi$.

For each $L\subseteq\mathbb{O}$, we write $p\sim_L q$ if $p$ and $q$
satisfy the same formulas in $L$. We say that $L$ is a \emph{modal
  characterisation} of some behavioural equivalence $\sim$ if $\sim_L$
coincides with $\sim$.
We write $\phi\equiv\phi'$ if $p\models\phi\Leftrightarrow p\models\phi'$
for all processes $p$. The class $L^\equiv$ denotes the closure of
$L\subseteq\mathbb{O}$ under $\equiv$. Trivially,
$p\sim_L q \Leftrightarrow p\sim_{L^\equiv}q$.

\begin{definition} \cite{vGl93} The subclasses $\IO{b}$ and $\IO{rb}$
  of $\mathbb{O}$ are defined as follows, where $a$ ranges over $A$
  and $\alpha$ over $A_\tau$:
  \[
    \begin{array}{ll}
      \IO{b}       & \phi ~~::=~~ {\displaystyle\bigwedge_{i\in I}}\,\phi_i ~~|~~ \neg\phi ~~|~~ \eps(\phi\diam{\hat{\tau}}\phi) ~~|~~ \eps(\phi\diam{a}\phi) \vspace{2mm}\\
      \IO{rb} \hspace*{2cm}      & \overline\phi ~~::=~~ {\displaystyle\bigwedge_{i\in I}}\,\overline\phi_i ~~|~~ \neg\overline\phi ~~|~~ \diam{\alpha}\phi ~~|~~ \phi ~~~~~~~(\phi\in\IO{b})
    \end{array}
  \]
\end{definition}

\noindent
$\IO{b}$ and $\IO{rb}$ are modal characterisations of $\bis{b}$ and
$\bis{rb}$, respectively (see \cite{FvGdW12}).

The idea behind stability is: (I) if $p\bis[s]{b} q$ and
$p\epsarrow p'\ntrans{\tau}$ with $p\bis[s]{b}p'$, then
$q\epsarrow q'\ntrans{\tau}$ with $q\bis[s]{b}q'$. In the definition
of $\bis[s]{b}$ this was formulated more weakly: (II) if
$p\bis[s]{b} q$ and $p\ntrans{\tau}$, then
$q\epsarrow q'\ntrans{\tau}$ with $q\bis[s]{b}q'$. To argue that
formulations (I) and (II) are equivalent, suppose $p\bis[s]{b}q$ and
$p\epsarrow p'\ntrans{\tau}$ with $p\bis[s]{b}p'$. Clearly
$q\epsarrow q''$ with $p'\bis[s]{b}q''$. Now the weaker property (II)
yields $q''\epsarrow q'\ntrans{\tau}$ with $q''\bis[s]{b}q'$. So
$q\epsarrow q'\ntrans{\tau}$ with $q\bis[s]{b}q'$. That formulations
(I) and (II) coincide is important to grasp the following modal
characterisation of stability-respecting branching bisimilarity,
because the additional clause at the end of the definition of
$\IO{b}^s$ is based on formulation (I).
\begin{definition} \cite{vGl93} The subclasses $\IO{b}^s$ and
  $\IO{rb}^s$ of $\mathbb{O}$ are defined as follows:
  \[
    \begin{array}{ll}
      \IO{b}^s & \phi ~~::=~~ {\displaystyle\bigwedge_{i\in I}}\,\phi_i ~~|~~ \neg\phi ~~|~~ \eps(\phi\diam{\hat{\tau}}\phi) ~~|~~ \eps(\phi\diam{a}\phi) ~~|~~ \eps(\neg\diam{\tau}\top\land\,\overline\phi) ~~~~~~~~(\overline\phi\in\IO{rb}^s) \vspace{2mm}\\
      \IO{rb}^s \hspace*{1cm}      & \overline\phi ~~::=~~ {\displaystyle\bigwedge_{i\in I}}\,\overline\phi_i ~~|~~ \neg\overline\phi ~~|~~ \diam{\alpha}\phi ~~|~~ \phi ~~~~~~~(\phi\in\IO{b}^s)
    \end{array}
  \]
\end{definition}

The additional clause $\eps(\neg\diam{\tau}\top\land\,\overline\phi)$
in the definition of $\IO{b}^s$ expresses stability. The first part
$\eps(\neg\diam{\tau}\top\,\ldots)$ captures
$p\epsarrow p'\ntrans{\tau}$, while the second part
$\ldots\,\land\,\overline\phi$ captures the stability-respecting
branching bisimulation class of $p'$. Note that since $p'$ is stable
and $\bis[s]{b}$ and $\bis[s]{rb}$ coincide on stable processes, we
can take the second part from 
$\IO{rb}^s$.  The proof of the following theorem is presented in the
appendix.

\begin{theorem}\label{thm:characterisation}
  $p\bis[s]{b} q\Leftrightarrow p\sim_{\IO{b}^s}q$ and
  $p\bis[s]{rb} q\Leftrightarrow p\sim_{\IO{rb}^s}q$, for all
  $p,q\in\mathbb{P}$.
\end{theorem}
\noindent

\subsection{Structural operational semantics}\label{sec:sos}

A \emph{signature} is a set $\Sigma$ of function symbols $f$ with
arity $\ar(f)$.  A function symbol of arity 0 is called a
\emph{constant}.  Let $V$ be an infinite set of variables, with
typical elements $x,y,z$; we assume $|\Sigma|, |A| \leq |V|$.  A
syntactic object is {\em closed} if it does not contain any variables.
The sets $\mathbb{T}(\Sigma)$ and $\mbox{T}(\Sigma)$ of terms over
$\Sigma$ and $V$ and closed terms over $\Sigma$, respectively, are
defined as usual; $t,u,v,w$ denote terms, $p,q$ denote closed terms,
and $\var(t)$ is the set of variables that occur in term $t$. A
\emph{substitution} $\sigma$ is a partial
function from $V$ to $\mathbb{T}(\Sigma)$. The result $\sigma(t)$
  of applying a substitution $\sigma$ to a term
  $t\in\mathbb{T}(\Sigma)$ is the term obtained by replacing in $t$
  all occurrences of variables $x$ in the domain of $\sigma$ by
  $\sigma(x)$. A \emph{closed substitution} is a substitution that is
  defined on all variables in $V$ and maps every variable to a closed
  term.

Structural operational semantics \cite{Plo04} generates an LTS in
which the processes are the closed terms. The labelled transitions
between processes are obtained from a transition system specification,
which consists of a set of proof rules called transition rules.

\begin{definition}\label{def:TSS}
  A (\emph{positive} or \emph{negative}) \emph{literal} is an
  expression \plat{$t \trans\alpha u$} or \plat{$t\ntrans\alpha$}. A
  \emph{(transition) rule} is of the form $\frac{H}{\lambda}$ with $H$
  a set of literals called the \emph{premises}, and $\lambda$ a
  literal called the \emph{conclusion}; the terms at the left- and
  right-hand side of $\lambda$ are called the \emph{source} and
  \emph{target} of the rule, respectively.  A rule
  $\frac{\emptyset}{\lambda}$ is also written $\lambda$.  A rule is
  {\em standard} if it has a positive conclusion.  A {\em transition
    system specification (TSS)}, written $(\Sigma,{\it Act},R)$,
  consists of a signature $\Sigma$, a set of actions $Act$, and a
  collection $R$ of transition rules over $\Sigma$.  A TSS is {\em
    standard} if all its rules are.
\end{definition}

\noindent
A TSS is meant to specify an LTS in which the transitions are the
closed positive literals that can be proved using the rules of the
TSS\@. It is straightforward to associate an appropriate notion of
provability in the special case of standard TSSs with only positive
premises. In the general case, in which the rules of the TSS may have
negative premises, consistency is a concern. Literals
\plat{$t \trans{\alpha} u$} and \plat{$t\ntrans{\alpha}$} are said to
\emph{deny} each other; a notion of provability associated with TSSs
is \emph{consistent} if it is not possible to prove two literals that
deny each other. To arrive at a consistent notion of provability in
the general case, we proceed in two steps: first we define the notion
of \emph{irredundant proof}, which on a standard TSS does not allow
the derivation of negative literals at all, and then arrive at a
notion of \emph{well-supported proof} that allows the derivation of
negative literals whose denials are manifestly underivable by
irredundant proofs.  In \cite{vGl04} it was shown that the notion of
well-supported provability is consistent.

\begin{definition}\label{def:proof} \cite{BFvG04}
  Let $P=(\Sigma,{\it Act},R)$ be a TSS. An {\em irredundant proof}
  from $P$ of a rule $\frac{H}{\lambda}$ is a well-founded tree with
  the nodes labelled by literals and some of the leaves marked
  ``hypothesis''. The root of this tree
    has label $\lambda$ and $H$ is the set of labels of the
    hypotheses. Moreover, the tree must satisfy the property that if
    $\mu$ is the label of a node that is not a hypothesis and $K$ is
    the set of labels of the children of this node, then
    $\frac{K}{\mu}$ is a substitution instance of a rule in $R$.
    If there exists such a tree for the rule
    \plat{$\frac{H}{\lambda}$}, then we say that it is irredundantly
    provable from $P$ (notation: \plat{$P \vdash_{\it irr}\frac{H}{\lambda}$}).
\end{definition}

\noindent
We note that if a leaf in a proof from $P$ is not marked as
hypothesis, then it is a substitution instance of a rule without
premises in $R$.

\begin{definition}\label{def:wsp} \cite{vGl04}
  Let $P=(\Sigma,{\it Act},R)$ be a standard TSS. A
  \emph{well-supported proof} from $P$ of a closed literal ${\lambda}$
  is a well-founded tree with the nodes labelled by closed literals,
  such that the root is labelled by $\lambda$, and if $\mu$ is the
  label of a node and $K$ is the set of labels of the children of this
  node, then:
  \begin{enumerate}
  \item either $\mu$ is positive and \plat{$\frac{K}{\mu}$} is a
    closed substitution instance of a rule in $R$;
  \item or $\mu$ is negative and for each set $N$ of closed negative
    literals with \plat{$\frac{N}{\nu}$} irredundantly provable from
    $P$ and $\nu$ a closed positive literal denying $\mu$, a literal
    in $K$ denies one in $N$.
  \end{enumerate}
  $P\vdash_{\it ws}\lambda$ denotes that a well-supported proof from
  $P$ of $\lambda$ exists.  A standard TSS $P$ is \emph{complete} if
  for each $p$ and $\alpha$, either
  \plat{$P\vdash_{\it ws}p\ntrans\alpha$} or there exists a closed
  term $q$ such that \plat{$P\vdash_{\it ws}p\trans\alpha q $}.
\end{definition}

\noindent
If $P=(\Sigma,{\it Act},R)$ is a complete TSS, then the LTS
\emph{associated with} $P$ is the LTS
$(\mbox{T}(\Sigma),{\it Act},\rightarrow)$ with
${\rightarrow}=\{(p,\alpha,q)\mid P\vdash_{\it ws}p{\trans\alpha}q\}$.
We do not associate an LTS with an incomplete TSS\@.

\subsection{Congruence formats}\label{sec:ntytt}

\noindent
Let $P=(\Sigma,{\it Act},R)$ be a transition system specification, and
let $\sim_P$ be an equivalence relation defined on the set of closed
terms $\mbox{T}(\Sigma)$. Then $\sim_{P}$ is a \emph{congruence} for
$P$ if, for each $f\in\Sigma$, we have that $p_i\sim_P q_i$ implies
$f(p_1,\dots,p_{\ar(f)})\sim_P f(q_1,\dots,q_{\ar(f)})$.  Note that
this is the case if for each open term $t \in \mathbb{T}(\Sigma)$ and
each pair of closed substitutions
$\rho,\rho':V \rightarrow \mbox{T}(\Sigma)$ we have
$(\forall x\in\var(t).~ \rho(x) \sim_P \rho'(x)) \Rightarrow \rho(t)
\sim_P \rho'(t)$.

Recall that we have associated with every complete TSS an LTS of which
the states are the closed terms of the TSS, and that a behavioural
equivalence $\sim$ associates with every such transition system an
equivalence on its set of states. Thus, $\sim$ associates with every
TSS $P$ an equivalence $\sim_P$ on its set of closed terms.  By a
\emph{congruence format} for a behavioural equivalence $\sim$ we mean
a class of TSSs such that for every TSS $P$ in the class the
equivalence $\sim_P$ is a congruence. Usually, a congruence format is
defined by means of a list of syntactic restrictions on the rules of
TSSs. We proceed to recap some terminology for syntactic restrictions
on rules \cite{BFvG04,Gro93,GV92}.

\begin{definition}\label{def:ntytt}
  An \emph{ntytt rule} is a rule in which the right-hand sides of
  positive premises are variables that are all distinct, and that do
  not occur in the source. An ntytt rule is an \emph{ntyxt rule} if
  its source is a variable, an \emph{ntyft rule} if its source
  contains exactly one function symbol and no multiple occurrences of
  variables, and an \emph{nxytt rule} if the left-hand sides of its
  premises are variables.
\end{definition}

\noindent
A well-known congruence format for strong
bisimulation semantics is the class of TSSs that consists of ntyft
  and ntyxt rules only \cite{Gro93,GV92}. Congruence formats for other semantics are
generally obtained by imposing additional restrictions on this
ntyft/ntyxt format.

\begin{definition}\label{def:decent}
  A variable in a rule is \emph{free} if it occurs neither in the
  source nor in right-hand sides of premises. A rule has
  \emph{lookahead} if some variable occurs in the right-hand side of a
  premise and in the left-hand side of a premise. A rule is
  \emph{decent} if it has no lookahead and does not contain free
  variables.
\end{definition}

\noindent
Each combination of syntactic restrictions on rules induces a
corresponding syntactic format for TSSs of the same name.  For
instance, a TSS is in decent ntyft format if it contains decent ntyft
rules only.

\begin{definition}\label{def:ready_sim}
  A TSS is in \emph{ready simulation format} if it consists of ntyft
  and ntyxt rules that have no lookahead.
\end{definition}

\noindent
In congruence formats for weak semantics, lookahead of two consecutive
actions from $A$ must be forbidden. To see this, consider the
extension of CCS \cite{Mil89} with a unary operator $f$ defined by the
rule
\[
  \frac{x\trans a y\quad y\trans b z}{f(x)\trans c z} \enskip.
\]
Then $ab{\bf 0}\bis{rb} a\tau b{\bf 0}$, whereas
$f(ab{\bf 0})\not\bis{rb}f(a\tau b{\bf 0})$.  Therefore congruence
formats for weak semantics are generally obtained by imposing
additional restrictions on the ready simulation format.

\section[Modal decomposition]{Modal decomposition}\label{sec:decompmethod}

  \begin{figure}[htb]
    \begin{center}
      \begin{tikzpicture}[baseline,double equal sign distance]
        \node (A) at (2,4)
        {$\forall {x\in\var(t)}.\ \rho(x)\models \psi(x)$}; \node (B)
        at (10,4) {$\forall {x\in\var(t)}.\ \rho'(x)\models\psi(x)$};
        \node (C) at (2,2) {$\rho(t)\models\varphi$}; \node (D) at
        (10,2) {$\rho'(t)\models\varphi$}; \node (MD) at (6,3)
        {$\begin{array}{c}\text{decomposition mapping}\\
            \psi: V\rightarrow L
          \end{array}$};
        \node (E) at (6,0) {$\rho(t)\sim_L\rho'(t)$};
        \node (F) at (6,6) {$\forall {x\in\var(t)}.\ \rho(x)\sim_L \rho'(x)$};
        \draw[decorate,decoration={brace,amplitude=10pt}] (0,4.5) -- (12,4.5);
        \draw[decorate,decoration={brace,amplitude=10pt}] (12,1.5) -- (0,1.5);
        \draw[double,-implies] (A) -- (B);
        \draw[double,-implies] (B) -- (D);
        \draw[double,-implies] (C) -- (A);
        \draw[double,dashed,-implies] (C)--(D);
        \draw[double,-implies] (F) -- (6,5);
        \draw[double,-implies] (6,1)--(E);
        \draw[->] (MD)--(2.5,3) node[above, midway] {$\exists$};
        \draw[->] (MD)-- (9.5,3) node[above, midway] {$\forall$};

      \end{tikzpicture}
    \end{center}
    \caption{Modal decomposition.}\label{fig:moddecomp}
  \end{figure}
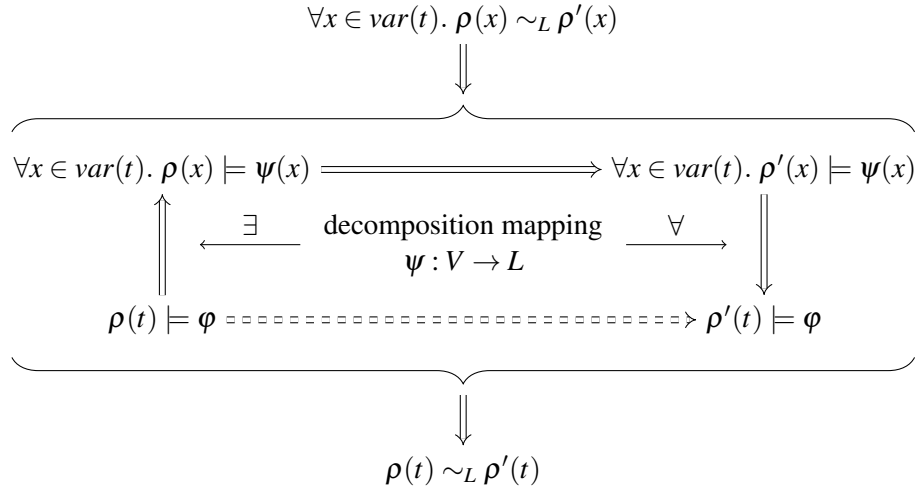

  Our goal in this paper is to present and discuss congruence formats
  for (rooted) stability-respecting and divergence-preserving
  branching bisimilarity. To derive these congruence formats and prove
  their correctness, we shall use the modal decomposition approach
  first proposed in \cite{BFvG04}, which conveniently employs a modal
  characterisation of the behavioural equivalence and a modal
  decomposition result. In this section, we shall explain the main
  ideas underlying the approach.

  Let $\sim$ be the behavioural equivalence under consideration. To
  prove that some class of TSSs is a congruence format for $\sim$, we
  should establish that $\sim_P$ is a congruence for every TSS $P$ in
  the class under consideration. That is, we should establish that for
  every $t\in\mathbb{T}(\Sigma)$ and for every pair of closed
  substitutions $\rho,\rho':V\rightarrow \mbox{T}(\Sigma)$ such that
  $\rho(x)\sim_P\rho'(x)$ for all $x\in\var(t)$ we have that
  $\rho(t)\sim_P\rho'(t)$.  Now, if $L$ is a modal characterisation of
  $\sim$, then our proof obligation can be reformulated as: for every
  $t\in\mathbb{T}(\Sigma)$ and for every pair of closed substitutions
  $\rho,\rho':V\rightarrow \mbox{T}(\Sigma)$, if, for all
  $x\in\var(t)$, $\rho(x)$ and $\rho'(x)$ satisfy the same formulas in
  $L$, then $\rho(t)$ and $\rho'(t)$ satisfy the same formulas in $L$.
  Clearly, by symmetry and using standard logical reasoning, it
  suffices to assume that $\rho(x)$ and $\rho'(x)$ satisfy the same
  formulas in $L$ for all $x\in\var(t)$, consider an arbitrary formula
  $\varphi\in L$ such that $\rho(t)\models\varphi$, and prove that
  $\rho'(t)\models\varphi$.
 
  Figure~\ref{fig:moddecomp} schematically illustrates how a modal
  decomposition result can support the correctness argument.  To
  effectively use the assumption that $\rho(x)$ and $\rho'(x)$ satisfy
  the same formulas in $L$, it is convenient to express a general
  correspondence between the satisfaction of formulas in $L$ by
  $\rho(t)$ and the satisfaction of formulas in $L$ by $\rho(x)$ for
  all $x\in\var(t)$.  With every $t\in\mathbb{T}(\Sigma)$ and every
  formula $\varphi\in L$ we wish to associate an assignment $\psi$ of
  formulas to the variables occurring in $t$ such that, for all closed
  substitutions $\rho$, we have that
  \begin{equation} \label{eq:decompprop} \rho(t) \models \varphi\
    \Leftrightarrow\ \forall x\in\var(t).\ \rho(x)\models
    \psi(x)\enskip.
  \end{equation}
  There does not, in general, exist an assignment $\psi$ that
  satisfies \eqref{eq:decompprop} for all substitutions $\rho$. We
  shall see, however, that we can associate with every
  $t\in\mathbb{T}(\Sigma)$ and $\varphi\in L$ a set of assignments
  $\psi$ satisfying the implication from right to left for all
  substitutions $\rho$; we call such assignments \emph{decomposition
    mappings} for $t$ and $\varphi$. Moreover, we shall see that for
  every substitution $\rho$ there exists a decomposition mapping that
  also satisfies the implication from left to right.

  The approach described above is applied in \cite{FvGdW12} to derive
  congruence formats for (rooted) branching bisimilarity and (rooted)
  $\eta$-bisimilarity, and establish their correctness.
  Here we do the same for (rooted) \emph{stability-respecting}
  branching bisimilarity. To this end, we have to show that if the TSS
  is in the congruence format for (rooted) stability-respecting
  branching bisimilarity, to be presented in
  Sect.~\ref{sec:congruence}, then the decomposition mappings
  associated with formulas in the modal characterisation for (rooted)
  stability-respecting branching bisimilarity yield again formulas in
  the modal characterisation for (rooted) stability-respecting
  branching bisimilarity.
  Since the modal
  characterisation of (rooted) stability-respecting branching
  bisimilarity is a small extension of the modal characterisation of (rooted)
  branching bisimilarity, we can build upon the preservation result
  for the latter semantics. 
 
  In the remainder of this section, we shall further explain and
  provide some intuitions for the technical ingredients of the modal
  decomposition approach, referring to the result from
  \cite{FvGdW12}. First, we shall introduce in
  Sect.~\ref{sec:ruloids} the auxiliary notion of \emph{ruloid}
  associated with a TSS, which allows us, for every term $t$ and every
  closed substitution $\rho$, to characterise the derivability of a
  transition from $\rho(t)$ in terms of the (non-)derivability of
  transitions from $\rho(x)$ (with $x$ ranging over $\var(t)$). Then
  we associate, in Sect.~\ref{sec:decomposition}, with every term a
  set of decomposition mappings. In Sect.~\ref{sec:rbbf} we present
  and briefly discuss the congruence format for (rooted) branching
  bisimilarity that was proposed in \cite{FvGdW12}.

\subsection{Ruloids}\label{sec:ruloids}

  Our modal decomposition result should establish a correspondence
  between the satisfaction of formulas in $L$ by $\rho(t)$ and the
  satisfaction of formulas in $\mathbb{O}$ by $\rho(x)$ for all
  $x\in\var(t)$. Since the satisfaction of formulas of the shape
  $\diam{\alpha}\varphi$ refers to the transition relation, it is
  convenient to characterise the provability in a presupposed TSS $P$
  of a transition from $\rho(t)$ in terms of the provability or
  refutability in $P$ of transitions from the $\rho(x)$. The
  characterisation is provided by the notion of $P$-ruloid. A
  $P$-ruloid is a decent nxytt rule \plat{$\frac{H}{\lambda}$} that is
  irredundantly provable in a TSS that is the result of a
  transformation of $P$. In \cite{BFvG04}, the transformation is
  presented at length and proved correct. Here we shall explain it
  superficially, providing just enough detail to be able to
  convincingly argue later that our congruence formats are preserved
  under his transformation.

First $P$ is converted to a standard TSS in decent ntyft format.  In
this conversion from \cite{GV92}, free variables in a rule are
replaced by arbitrary closed terms, and if the source is of the form
$x$, then this variable is replaced by a term
$f(x_1,\ldots,x_{\ar(f)})$ for each $n$-ary function symbol $f$ in the
signature of $P$, where the variables $x_1,\ldots,x_{\ar(f)}$ are
fresh.  Next, using a construction from \cite{FvG96}, left-hand sides
of positive premises in rules of $P$ are reduced to variables.  In the
final transformation step, non-standard rules with a negative
conclusion \plat{$t\ntrans\alpha$} are introduced.  The motivation is
that instead of the notion of well-founded provability of
Def.~\ref{def:wsp}, we want a more constructive notion like
Def.~\ref{def:proof}, by making it possible that a negative premise is
matched with a negative conclusion. A non-standard rule
$\frac{H}{f(x_1,\ldots,x_{\ar(f)})\ntrans{\alpha}}$ is obtained by
picking one premise from each standard rule with a conclusion of the
form \plat{$f(x_1,\ldots,x_{\ar(f)})\trans{\alpha}t$}, and including
the denial of each of the selected premises as a premise in $H$.

The resulting TSS, which is in decent ntyft format, is denoted by
$P^+$.  In \cite{BFvG04} it was established, for all closed
  literals $\mu$, that $P\vdash_{\it ws}\mu$ if and only if $\mu$ is
  irredundantly provable from $P^+$.  The $P$-ruloids are those
decent nxytt rules that are irredundantly provable from $P^+$.
  In \cite[Lem.~13---aliased 8.2]{BFvG04}
  it was established that $P\vdash_{\it ws}\rho(t)\trans{a}q$, with
  $\rho$ a closed substitution, if and only if there is a ruloid
  \plat{$\frac{H}{t\trans{a}t'}$} and a closed substitution $\rho'$
  that agrees with $\rho$ on $\var(t)$ such that $\rho(t')=q$ and
  $P\vdash_{\it ws}\rho'(H)$ .
  
\subsection{Decomposition of modal formulas}\label{sec:decomposition}

The decomposition method proposed in \cite{FvGdW12} gives a special
treatment to arguments of function symbols that are deemed
\emph{patient}; a predicate marks the arguments that get this special
treatment.

\begin{definition}\label{def:patience_rule} \cite{Blo95,Fok00}
  Let $\Gamma$ be a predicate on arguments of function symbols.  A
  standard ntyft rule is a {\em $\Gamma$-patience rule} if it is of
  the form
  \[
    \frac{x_i\trans{\tau}y}{f(x_1,\ldots,x_{\ar(f)})\trans{\tau}f(x_1,\ldots,x_{i-1},y,x_{i+1},\ldots,x_{\ar(f)})}
  \]
  with $\Gamma(f,i)$.  A TSS is \emph{$\Gamma$-patient} if it contains
  all $\Gamma$-patience rules.  A standard ntytt rule is {\em
    $\Gamma$-patient} if it is irredundantly provable from the
  $\Gamma$-patience rules; else it is called {\em $\Gamma$-impatient}.
\end{definition}

\noindent
A patience rule for an argument $i$ of a function symbol $f$ expresses
that terms $f(p_1,\ldots,p_{\ar(f)})$ can mimic the $\tau$-transitions
of argument $p_i$ (cf.\ \cite{Blo95,Fok00}).  Typically, in process
algebra, there are patience rules for the arguments of parallel composition and
for the first argument of sequential composition, but not for the
arguments of the alternative composition $+$ or for the second
argument of sequential composition.

\begin{definition}\label{def:liquid/frozen} \cite{BFvG04}
  Let $\Gamma$ be a predicate on
  $\{(f,i)\mid 1 \leq i \leq ar(f),~f \in \Sigma \}$. If
  $\Gamma(f,i)$, then argument $i$ of $f$ is {\em $\Gamma$-liquid};
  otherwise it is {\em $\Gamma$-frozen}. An occurrence of $x$ in $t$
  is {\em $\Gamma$-liquid} if either $t=x$, or
  $t=f(t_1,\ldots,t_{\ar(f)})$ and the occurrence is $\Gamma$-liquid
  in $t_i$ for a liquid argument $i$ of $f$; otherwise the occurrence
  is {\em $\Gamma$-frozen}.
\end{definition}

We now show how to decompose formulas from $\mathbb{O}$.  To each term
$t$ and formula $\phi$ we assign a set $t^{-1}(\phi)$ of decomposition
mappings $\psi:V\rightarrow\mathbb{O}$. Each of these mappings
$\psi\in t^{-1}(\phi)$ has the property that, for all closed
substitutions $\rho$, $\rho(t)\models\phi$ if $\rho(x)\models\psi(x)$
for all $x\in\var(t)$. Vice versa, whenever $\rho(t)\models\phi$,
there is a decomposition mapping $\psi \in t^{-1}(\phi)$ with
$\rho(x)\models\psi(x)$ for all $x \in \var(t)$.

\begin{definition}\label{def:decomposition} \cite{FvGdW12}
  Let $P=(\Sigma,A_\tau,R)$ be a $\Gamma$-patient standard TSS in
  ready simulation format. We define
  $\cdot^{-1}:\mathbb{T}(\Sigma)\times\mathbb{O}\rightarrow\mbox{\fsc
    P}(V\mathbin{\rightarrow}\mathbb{O})$ as the function that for
  each $t\in\mathbb{T}(\Sigma)$ and $\varphi\in\mathbb{O}$ returns the
  smallest set
  $t^{-1}(\varphi)\in\mbox{\fsc P}(V\mathbin{\rightarrow}\mathbb{O})$
  of decomposition mappings $\psi:V\mathbin{\rightarrow}\mathbb{O}$
  satisfying the following six conditions. Let $t$ denote a univariate
  term, i.e.\ without multiple occurrences of the same variable.
  (Cases \ref{dec1}--\ref{dec5} associate with every univariate term
  $t$ a set $t^{-1}(\varphi)$. Then, in Case~\ref{dec6}, the
  definition is generalised to terms that are not univariate, using
  that every term can be obtained by applying a non-injective
  substitution to a univariate term.)

  \begin{enumerate}
  \item\label{dec1} $\psi\in t^{-1}(\bigwedge_{i\in I}\phi_i)$ iff
    there are $\psi_i\in t^{-1}(\phi_i)$ for each $i\in I$ such that
    \[
      \psi(x)= \displaystyle{\bigwedge_{i\in I}\psi_i(x)}\qquad\text{
        for all }x\in V
    \]

  \item\label{dec2} $\psi\in t^{-1}(\neg\phi)$ iff there is a function
    $h:t^{-1}(\phi)\rightarrow\var(t)$ such that
    \[
      \psi(x)=\left\{
        \begin{array}{ll}
          \displaystyle{\bigwedge_{\chi\in h^{-1}(x)}\!\!\!\neg\chi(x)} & \mbox{if $x\in\var(t)$}\vspace{2mm}\\
          \top & \mbox{if $x\notin\var(t)$}
        \end{array}
      \right.
    \]

  \item\label{dec3} $\psi\in t^{-1}(\diam\alpha\phi)$ iff there is a
    $P$-ruloid $\frac{H}{t\trans{\alpha}u}$ and a
    $\chi\in u^{-1}(\phi)$ such that
    \[
      \psi(x)=\left\{
        \begin{array}{ll}
          \displaystyle{\chi(x)\ \land\ \bigwedge_{x\trans{\beta}y\in H}\!\!\!\!\!\diam{\beta}\chi(y)
          \ \land\  \bigwedge_{x\ntrans{\gamma}\in H}\!\!\!\!\!\neg\diam{\gamma}\top\;} & \mbox{if $x\in\var(t)$}\\
          \top & \mbox{if $x\notin\var(t)$}
        \end{array}
      \right.
    \]

  \item\label{dec4} $\psi\in t^{-1}(\eps\phi)$ iff one of the
    following holds: \vspace{2mm}

    \begin{enumerate}
    \item
      \label{4a}
      either there is a $\chi\in t^{-1}(\phi)$ such that
      \[
        \psi(x)=\left\{
          \begin{array}{ll}
            \eps\chi(x) & \textrm{if $x$ occurs $\Gamma$-liquid in $t$}\vspace{1mm}\\
            \chi(x) & \textrm{otherwise}\\
          \end{array}
        \right.
      \]

    \item
      \label{4b}
      or there is a $\Gamma$-impatient $P$-ruloid
      $\frac{H}{t\trans{\tau}u}$ and a $\chi\in u^{-1}(\eps\phi)$ such
      that
      \[
        \psi(x)=\left\{
          \begin{array}{@{}ll@{}}
            \top & \mbox{if $x\notin\var(t)$}\vspace{2mm}\\
            \displaystyle{\eps\Big(\chi(x)\ \land\ \!\!\!\!\!\!\bigwedge_{x\trans{\beta}y\in H}\!\!\!\!\!\diam{\beta}\chi(y)
            \ \land\  \!\!\!\!\!\!\bigwedge_{x\ntrans{\gamma}\in H}\!\!\!\!\!\neg\diam{\gamma}\top\Big)} &
                                                                                                           \mbox{\begin{tabular}{@{}l@{}}if $x$ occurs\\ $\Gamma$-liquid in $t$\end{tabular}}\vspace{2mm}\\
            \displaystyle{\chi(x) \land \!\!\!\!\!\!\bigwedge_{x\trans{\beta}y\in H}\!\!\!\!\!\diam{\beta}\chi(y)
            \ \land \!\!\!\!\!\bigwedge_{x\ntrans{\gamma}\in
            H}\!\!\!\!\!\neg\diam{\gamma}\top\;} &
                                                   \mbox{otherwise}
          \end{array}
        \right.
      \]
    \end{enumerate}

  \item\label{dec5} $\psi\in t^{-1}(\diam{\hat\tau}\phi)$ iff one of
    the following holds: \vspace{2mm}

    \begin{enumerate}
    \item
      \label{5a}
      either $\psi\in t^{-1}(\phi)$; \vspace{2mm}

    \item
      \label{5b}
      or there is an $x_0$ that occurs $\Gamma$-liquid in $t$, and a
      $\chi\in t^{-1}(\phi)$ such that
      \[
        \psi(x)=\left\{
          \begin{array}{ll}
            \diam{\hat\tau}\chi(x) & \textrm{if $x=x_0$} \vspace{1mm}\\
            \chi(x) & \textrm{otherwise}\\
          \end{array}
        \right.
      \]

    \item
      \label{5c}
      or there is a $\Gamma$-impatient $P$-ruloid
      $\frac{H}{t\trans{\tau}u}$ and a $\chi\in u^{-1}(\phi)$ such
      that
      \[
        \psi(x)=\left\{
          \begin{array}{@{}ll@{}}
            \displaystyle{\chi(x) \land \!\!\!\!\!\!\bigwedge_{x\trans{\beta}y\in H}\!\!\!\!\!\diam{\beta}\chi(y)
            \ \land \!\!\!\!\!\bigwedge_{x\ntrans{\gamma}\in
            H}\!\!\!\!\!\neg\diam{\gamma}\top\;} & \mbox{if $x\in\var(t)$}\\
            \top &         \mbox{otherwise}
          \end{array}
        \right.
      \]
    \end{enumerate}

  \item\label{dec6} $\psi\in\sigma(t)^{-1}(\phi)$ for a non-injective
    substitution $\sigma:\var(t)\rightarrow V$ iff there is a
    $\chi\in t^{-1}(\phi)$ such that
    \[
      \psi(x)=\bigwedge_{z\in\sigma^{-1}(x)}\chi(z)\qquad\text{ for
        all }x\in V
    \]
  \end{enumerate}
\end{definition}

\noindent
The following theorem will be the key to the forthcoming congruence
results.

\begin{theorem}\label{thm:decomposition} \cite{FvGdW12}
  Let $P=(\Sigma,A_\tau,R)$ be a $\Gamma$-patient complete standard
  TSS in ready simulation format. For each term
  $t\in\mathbb{T}(\Sigma)$, closed substitution $\rho$, and
  $\phi\in\mathbb{O}$:
  \[
    \rho(t)\models\phi ~~\Leftrightarrow~~ \exists\psi\in
    t^{-1}(\phi)~\,\forall x\in \var(t):~\rho(x)\models\psi(x)
  \]
\end{theorem}

\subsection{Deriving the (rooted) branching bisimulation
    format} \label{sec:rbbf}

  Def.~\ref{def:decomposition} yields for every term $t$ and every
  formula in the modal language $\mathbb{O}$ a set of decomposition
  mappings. Note that the modal language $\mathbb{O}_b$
  ($\mathbb{O}_{rb}$), which characterises (rooted) branching
  bisimilarity, is a sublanguage of $\mathbb{O}$. In order to prove
  that (rooted) branching bisimilarity is a congruence for a TSS $P$, in
  view of Thm.~\ref{thm:decomposition} it therefore suffices to
  provide an argument that the decomposition mappings associated with
  formulas in $\mathbb{O}_b$ ($\mathbb{O}_{rb}$) assign formulas in
  $\mathbb{O}_b$ ($\mathbb{O}_{rb}$) to all variables. The congruence
  format should facilitate this argument. By carefully studying the
  syntactic shapes of the formulas assigned to variables by the
  decomposition mappings associated with a term $t$ and a formula
  $\varphi\in\mathbb{O}_b$ ($\varphi\in\mathbb{O}_{rb}$), we can
  derive sufficient syntactic conditions on ruloids associated with
  $P$, which play a role in the definition of the decomposition
  mappings. These syntactic conditions on the ruloids, in turn, give
  rise to sufficient conditions on the rules of $P$.

  Let us illustrate this approach by considering a $\Gamma$-patient
  standard TSS $P$ in ready simulation format, a formula
  $\varphi\in\mathbb{O}_b$ of the shape
  $\eps(\varphi_1\diam{a}\varphi_2)$, with
  $\varphi_1,\varphi_2\in\mathbb{O}_b$, a term $t$, and a variable
  $x\in\var(t)$. We first determine the general shape of the formula
  $\psi(x)$ assigned to $x$ by a decomposition mapping
  $\psi\in t^{-1}(\varphi)$, and then formulate sufficient conditions
  on the rules of $P$ that restrict the general shape in such a way
  that we can be certain that $\psi(x)\in\mathbb{O}_b$.

  Since $\varphi=\eps(\varphi_1\diam{a}\varphi_2)$, the decomposition
  mapping $\psi$ satisfies clause~\ref{dec4} of
  Def.~\ref{def:decomposition} and hence satisfies one of its two
  subclauses \ref{4a} or \ref{4b}. For the purpose of the explanation
  of the branching bisimulation format, it is enough to consider the
  case that $\varphi$ satisfies subclause \ref{4a}.

  If, on the one hand, $x$ occurs $\Gamma$-liquid in $t$, then from
  clauses \ref{dec1} and \ref{dec3} it follows that there exist a
  $P$-ruloid $\frac{H}{t\trans{\alpha}u}$ and decomposition mappings
  $\psi_1\in t^{-1}(\varphi_1)$ and $\psi_2\in u^{-1}(\varphi_2)$ such
  that
  \begin{equation*}
    \psi(x)=
    \eps\left(\psi_1(x)\
      \land\
      \psi_2(x)\ \land\
      \bigwedge_{x\trans{\beta}y\in H}\!\!\!\!\!\diam{\beta}\psi_2(y)
      \ \land\  \bigwedge_{x\ntrans{\gamma}\in
        H}\!\!\!\!\!\neg\diam{\gamma}\top\right) \enskip.
  \end{equation*}
  If we assume, inductively, that $\psi_1(x)$ and $\psi_2(y)$ are
  formulas in $\mathbb{O}_b$, then $\psi(x)\in\mathbb{O}_b$ if we can
  write the argument of $\eps$ in $\psi(x)$ as
  $\varphi'\diam{a}\varphi''$ or as
  $\varphi'\diam{\hat\tau}\varphi''$.  Clearly, this means that $H$
  can contain at most one positive premise $x\trans{\beta}y$, with
  $\beta\neq\tau$,  and
  cannot contain a negative premise \plat{$x\ntrans{\gamma}$}.

  If, on the other hand, $x$ does not occur $\Gamma$-liquid in $t$,
  then from clauses \ref{dec1} and \ref{dec3} it follows that there
  exist a $P$-ruloid $\frac{H}{t\trans{\alpha}u}$ and decomposition
  mappings $\psi_1\in t^{-1}(\varphi_1)$ and
  $\psi_2\in u^{-1}(\varphi_2)$ such that
  \begin{equation*}
    \psi(x)=
    \psi_1(x)\
    \land
    \psi_2(x)\
    \land\
    \bigwedge_{x\trans{\beta}y\in H}\!\!\!\!\!\diam{\beta}\psi_2(y)
    \ \land\  \bigwedge_{x\ntrans{\gamma}\in
      H}\!\!\!\!\!\neg\diam{\gamma}\top \enskip,
  \end{equation*}
  and it is clear that $x$ cannot occur at all in $H$.

  The analysis above yields a rudimentary syntactic requirement for
  the $P$-ruloids \plat{$\frac{H}{t\trans{\alpha}u}$}: \textit{$H$
    may not contain negative premises and, for all variables
    $x\in\var(t)$, if $x$ has no $\Gamma$-liquid occurrence in $t$,
    then it does not have a $\Gamma$-liquid occurrence in $H$ either,
    and if $x$ does have a $\Gamma$-liquid occurrence in $t$, then it
    can have at most one $\Gamma$-liquid occurrence of $x$ in
    $H$, which must be in a positive premise $x\trans{\beta}y$ with
    $\beta\neq\tau$}.
  
  As is shown in \cite{FvGdW12}, the requirement derived above can be
  relaxed if $\Gamma$ is defined as $\Lambda\cap\aleph$, where
  $\Lambda$ and $\aleph$ are two auxiliary predicates. Intuitively,
  the predicate $\Lambda$ marks arguments that contain processes that
  have started executing (but may currently be unable to execute),
  while $\aleph$ marks arguments that contain processes that
  can execute immediately.  For example, in process algebra, $\Lambda$
  and $\aleph$ hold for the arguments of the parallel composition $t_1\|t_2$ and for
  the first argument of sequential composition $t_1{\cdot}t_2$; they
  can contain processes that started to execute in the past, and these
  processes can continue their execution immediately. In absence of
  the empty process, $\Lambda$ and $\aleph$ do not hold for the second argument of sequential
  composition; it contains a process that did not yet start to
  execute, and cannot execute immediately. If immediate
    termination is possible, i.e., in the presence of the empty
    process,
  then this argument becomes $\aleph$-liquid; cf.\ the sequencing operator
  example in Sect.~\ref{sec:sequencing}.
  $\Lambda$ does not hold and $\aleph$ holds for the
  arguments of alternative composition $t_1+t_2$; they contain
  processes that did not yet start to execute, but that can start
  executing immediately.

  Below, we recall the (rooted) branching bisimulation format proposed
  in \cite{FvGdW12}. Due to the use of the two predicates $\Lambda$
  and $\aleph$, the syntactic restrictions are technically more
  complicated than sketched above. We refer the reader to
  \cite{FvGdW12} for further explanations and examples.

\begin{definition}\label{def:rbbsafe}
  Let $\aleph$ and $\Lambda$ be predicates on
  $\{(f,i)\mid 1 \leq i \leq ar(f),~f \in \Sigma \}$.  A standard
  ntytt rule $r=\frac{H}{t\trans\alpha u}$ is {\em rooted branching
    bisimulation safe} w.r.t.\ $\aleph$ and $\Lambda$ if it satisfies
  the following conditions.
  \begin{enumerate}
  \item Right-hand sides of positive premises occur only
    $\Lambda$-liquid in $u$.
  \item If $x\in\var(t)$ occurs only $\Lambda$-liquid in $t$, then $x$
    occurs only $\Lambda$-liquid in $r$.
  \item If $x\in\var(t)$ occurs only $\aleph$-frozen in $t$, then $x$
    occurs only $\aleph$-frozen in $H$.
  \item \label{mainrbb} If $x\in\var(t)$ has exactly one
    $\aleph$-liquid occurrence in $t$ and this occurrence is also
    $\Lambda$-liquid, then $x$ has at most one $\aleph$-liquid
    occurrence in $H$ and this occurrence must be in a positive
    premise. If, moreover, this premise is labelled $\tau$, then $r$
    must be $\aL$-patient.
  \end{enumerate}
  A standard TSS is in {\em rooted branching bisimulation format} if
  it is in ready simulation format and, for some $\aleph$ and
  $\Lambda$, it is $\aL$-patient and only contains rules that are
  rooted stability-respecting branching bisimulation safe w.r.t.\
  $\aleph$ and $\Lambda$. It is in {\em branching bisimulation format}
  if moreover $\Lambda$ is universal, i.e., $\Lambda(f,i)$ for all
  $f\in\Sigma$ and $i=1,\ldots,\ar(f)$.
\end{definition}

\section{Stability-respecting branching bisimilarity as a
  congruence}\label{sec:congruence}

  We present, in Sect.~\ref{sec:formats}, congruence formats for
  stability-respecting branching bisimilarity and rooted
  stability-respecting branching bisimilarity, as relaxations of the
  congruence formats for branching bisimilarity and rooted branching
  bisimilarity, respectively. We illustrate the usefulness of the
  formats with applications to the priority operator and an operator
  for sequencing. In Sect.~\ref{sec:correctness} we prove the
  correctness of the formats. 

\subsection{The (rooted) stability-respecting branching bisimulation
  format}\label{sec:formats}

We define when a standard ntytt rule is
rooted stability-respecting branching bisimulation safe, and base the
rooted stability-respecting branching bisimulation format on that
notion. As with the branching bisimulation format, to define the
stability-respecting branching bisimulation format we add one additional restriction to its
rooted counterpart: $\Lambda$ is universal.  Our
aim for the rest of this section will be to prove that the (rooted)
stability-respecting branching bisimulation format guarantees that
(rooted) stability-respecting branching bisimilarity is a congruence.

\begin{definition}\label{def:rooted_bra_bisimulation_safe}
  Let $\aleph$ and $\Lambda$ be predicates on
  $\{(f,i)\mid 1 \leq i \leq ar(f),~f \in \Sigma \}$.  A standard
  ntytt rule $r=\frac{H}{t\trans\alpha u}$ is {\em rooted
    stability-respecting branching bisimulation safe} w.r.t.\ $\aleph$
  and $\Lambda$ if it satisfies the following conditions.
  \begin{enumerate}
  \item \label{rhs} Right-hand sides of positive premises occur only
    $\Lambda$-liquid in $u$.\vspace{-1pt}
  \item \label{Lambda} If $x\in\var(t)$ occurs only $\Lambda$-liquid
    in $t$, then $x$ occurs only $\Lambda$-liquid in $r$.\vspace{-1pt}
  \item \label{aleph} If $x\in\var(t)$ occurs only $\aleph$-frozen in
    $t$, then $x$ occurs only $\aleph$-frozen in $H$.\vspace{-1pt}
  \item \label{main} Suppose that $x$ has exactly one $\aleph$-liquid
    occurrence in $t$, and that this occurrence is also
    $\Lambda$-liquid.\vspace{-1pt}
    \begin{enumerate}
    \item \label{main1} If $x$ has an $\aleph$-liquid occurrence in a
      negative premise in $H$ or more than one $\aleph$-liquid
      occurrence in the positive premises in $H$, then there is a
      premise \plat{$v\ntrans\tau$} in $H$ such that $x$ occurs
      $\aleph$-liquid in $v$.\vspace{-1pt}
    \item \label{main2} If there is a premise $w\trans{\tau}y$ in $H$
      and $x$ occurs $\aleph$-liquid in $w$, then $r$ is
      $\aL$-patient.
    \end{enumerate}
  \end{enumerate}
\end{definition}
Conditions \ref{rhs}--\ref{aleph} have been copied from Def.~\ref{def:rbbsafe}, and condition
\ref{main2} is part of condition 4 in that definition.  Condition
\ref{main1}, however, establishes a relaxation of condition 4 in the
definition of rooted branching bisimulation safeness, where it was
required that $x$ has at most one $\aleph$-liquid occurrence in $H$,
which must be in a positive premise. Here, owing to stability, we can
be more tolerant, as long as $x\ntrans\tau$ can be derived.  As a
consequence of this relaxation we will see that the rule for the
priority operator is rooted stability-respecting branching
bisimulation safe, while it is not rooted branching bisimulation safe.

\begin{definition}\label{def:bra_bisimulation_format}
  A standard TSS is in {\em rooted stability-respecting branching
    bisimulation format} if it is in ready simulation format and, for
  some $\aleph$ and $\Lambda$, it is $\aL$-patient and only contains
  rules that are rooted stability-respecting branching bisimulation
  safe w.r.t.\ $\aleph$ and $\Lambda$.

  This TSS is in {\em stability-respecting branching bisimulation
    format} if moreover $\Lambda$ is universal.
\end{definition}

\paragraph{Application to the priority operator}
\label{sec:priority}

\renewcommand{\epsilon}{\varepsilon} 

The {\em priority} operator \cite{BBK86} is a unary function the
definition of which is based on an ordering $<$ on atomic actions.
The term $\Theta(p)$ executes the transitions of the term $p$, with
the restriction that a transition \plat{$p\trans {\alpha} q$} only gives
rise to a transition \plat{$\Theta(p)\trans {\alpha} \Theta(q)$} if there
does not exist a transition \plat{$p\trans{\beta} q'$} with $\beta>\alpha$.
This intuition is captured by the rule for the priority operator below.\vspace{-1.7ex}
\[
  \frac{x\trans {\alpha} y~~~~~~~~x\ntrans{\beta}\mbox{ for all }
    \beta>\alpha}{\Theta(x)\trans {\alpha} \Theta(y)}
\]

\noindent
The priority operator does not preserve [rooted] branching bisimilarity
(cf.\ \cite[pp.\ 130--132]{Vaa90}), as shown by the following example.

\begin{example}\label{ex:priority}
  Consider the following two LTSs:

  \vspace{2mm}

  \centerline{\begin{picture}(0,0)%
\includegraphics{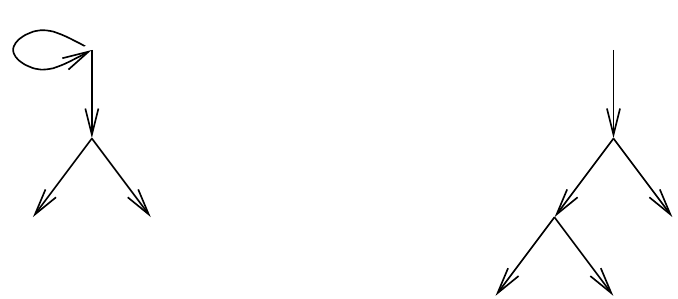}%
\end{picture}%
\setlength{\unitlength}{4144sp}%
\begingroup\makeatletter\ifx\SetFigFont\undefined%
\gdef\SetFigFont#1#2#3#4#5{%
  \reset@font\fontsize{#1}{#2pt}%
  \fontfamily{#3}\fontseries{#4}\fontshape{#5}%
  \selectfont}%
\fi\endgroup%
\begin{picture}(3087,1362)(1606,-1693)
\put(4456,-736){\makebox(0,0)[lb]{\smash{{\SetFigFont{10}{12.0}{\rmdefault}{\mddefault}{\updefault}{\color[rgb]{0,0,0}$a$}%
}}}}
\put(4006,-1456){\makebox(0,0)[rb]{\smash{{\SetFigFont{10}{12.0}{\rmdefault}{\mddefault}{\updefault}{\color[rgb]{0,0,0}$a$}%
}}}}
\put(4276,-1456){\makebox(0,0)[lb]{\smash{{\SetFigFont{10}{12.0}{\rmdefault}{\mddefault}{\updefault}{\color[rgb]{0,0,0}$b$}%
}}}}
\put(4546,-1096){\makebox(0,0)[lb]{\smash{{\SetFigFont{10}{12.0}{\rmdefault}{\mddefault}{\updefault}{\color[rgb]{0,0,0}$a$}%
}}}}
\put(4276,-1096){\makebox(0,0)[rb]{\smash{{\SetFigFont{10}{12.0}{\rmdefault}{\mddefault}{\updefault}{\color[rgb]{0,0,0}$\tau$}%
}}}}
\put(4411,-466){\makebox(0,0)[b]{\smash{{\SetFigFont{10}{12.0}{\rmdefault}{\mddefault}{\updefault}{\color[rgb]{0,0,0}$q$}%
}}}}
\put(2071,-736){\makebox(0,0)[lb]{\smash{{\SetFigFont{10}{12.0}{\rmdefault}{\mddefault}{\updefault}{\color[rgb]{0,0,0}$a$}%
}}}}
\put(2161,-1096){\makebox(0,0)[lb]{\smash{{\SetFigFont{10}{12.0}{\rmdefault}{\mddefault}{\updefault}{\color[rgb]{0,0,0}$b$}%
}}}}
\put(1891,-1096){\makebox(0,0)[rb]{\smash{{\SetFigFont{10}{12.0}{\rmdefault}{\mddefault}{\updefault}{\color[rgb]{0,0,0}$a$}%
}}}}
\put(1621,-601){\makebox(0,0)[rb]{\smash{{\SetFigFont{10}{12.0}{\rmdefault}{\mddefault}{\updefault}{\color[rgb]{0,0,0}$\tau$}%
}}}}
\put(2049,-466){\makebox(0,0)[b]{\smash{{\SetFigFont{10}{12.0}{\rmdefault}{\mddefault}{\updefault}{\color[rgb]{0,0,0}$p$}%
}}}}
\end{picture}%
}

  \vspace{2mm}

\noindent
Clearly $p\bis{b}q$. Note that on the other hand $p\notbis[s]{b}q$,
because $q$ is stable while $p$ cannot perform a sequence of
$\tau$-transitions to a stable state.

Suppose that $a<b$. Let us try to extend the ordering $<$ such that
$\Theta(p)\bis{b}\Theta(q)$. Since $\Theta(p)$ cannot execute the
trace $aa$, we must declare $a<\tau$, to block this trace in
$\Theta(q)$. But then $\Theta(p)$ can only execute an infinite
$\tau$-sequence while $\Theta(q)$ can execute $a$. Hence
$\Theta(p)\notbis{b}\Theta(q)$ for every ordering $<$ (with $a<b$).

Moreover, if $p_0$ and $q_0$ are processes with as only transitions
$p_0 \trans{\tau} p$ and $q_0 \trans{\tau} q$, then
$p_0 \bis{rb} q_0$, but $\Theta(p_0)\notbis{rb}\Theta(q_0)$ for every
ordering $<$ (with $a<b$).
\end{example}
So inevitably, as observed in \cite{FvGdW12}, the rule for the
priority operator is not in the rooted branching bisimulation format.
Namely, the $\aL$-liquid argument $x$ in the source occurs
$\aleph$-liquid in the negative premises, which violates the more
restrictive condition \ref{main} of the rooted branching bisimulation
format.

We proceed to show that the rule for the priority operator does
  satisfy the relaxed condition \ref{main} of
  Def.~\ref{def:rooted_bra_bisimulation_safe}. To this end, first note
  that in view of the target $\Theta(y)$, by condition \ref{rhs} of
Def.~\ref{def:rooted_bra_bisimulation_safe}, the argument of $\Theta$
must be chosen $\Lambda$-liquid. And in view of condition \ref{aleph}
of Def.~\ref{def:rooted_bra_bisimulation_safe}, the argument of
$\Theta$ must be chosen $\aleph$-liquid.  Then the rule above is
rooted stability-respecting branching bisimulation safe, if the
following condition on the ordering on atomic actions is satisfied: if
there is a $\beta$ such that $\beta>\alpha$, then $\tau>\alpha$. Namely,
this guarantees that condition \ref{main1} of
Def.~\ref{def:rooted_bra_bisimulation_safe} is satisfied: if there is
a negative premise \plat{$x\ntrans{\beta}$}, then there is also a negative
premise $x\ntrans\tau$.
Moreover, note that then the rule for the priority operator with $\alpha=\tau$
constitutes a patience rule because there can be no $\beta$ with $\beta>\tau$. Furthermore,
since the argument of $\Theta$ is $\Lambda$-liquid, this operator is within
the stability-respecting branching bisimulation format.
Thm.~\ref{thm:congruence1} and Thm.~\ref{thm:congruence2}, which
are presented at the end of Sect.~\ref{sec:congruence}, will therefore
imply the following congruence results.

\begin{corollary}
  $\bis[s]{b}$ and $\bis[s]{rb}$ are congruences for the priority operator.
\end{corollary}

\paragraph{Application to sequencing}\label{sec:sequencing}
\newcommand{\seqc}{\ensuremath{\mathop{;}}}

The binary \emph{sequencing operator} $\seqc$ is a variant of
sequential composition that does not rely on a notion of successful
termination. Intuitively, the process $p \seqc q$ behaves as its
left-hand side argument until that can no longer do any transitions;
then it proceeds with its right-hand side argument. The rules below,
which appeared e.g.\ in \cite{Blo94}, formalise this behaviour:\vspace{-4pt}
\[
  \frac{x\trans{\alpha} x'}{x\seqc{} y \trans{\alpha} x'\seqc{}y}
  \qquad\qquad \frac{x\ntrans{\alpha}\mbox{ for all } \alpha\in
    A_\tau~~~~~~~~y\trans{\beta} y'}{x\seqc{}y\trans{\beta} y'}
\]
Sequencing does not preserve [rooted] branching bisimilarity, as shown
by the following example.
\begin{example}
  Consider
  three processes $p$, $q$, and $r$ with $p\trans{\tau}p$ and $r\trans{a}\_$
  while $q$ cannot perform any transitions.
Then $p \bis{b} q$, but $p\seqc r \bis{} p \not\bis{b} r \bis{} q\seqc r$. Note
that $p\not\bis[s]{b} q$ since $q$ is stable while $p$ cannot perform
a sequence of $\tau$-transitions to a stable state.

Using processes $p_0$ and $q_0$ as in Ex.~\ref{ex:priority} shows that
also $\bis{rb}$ fails to be a congruence for sequencing.
\end{example}

We proceed to show that if both arguments of $\seqc$ are chosen to be
$\aleph$-liquid, and only the first argument of $\seqc$ is chosen to
be $\Lambda$-liquid, then both rules are rooted stability-respecting
branching bisimulation safe w.r.t. $\aleph$ and $\Lambda$ (see
Def.~\ref{def:rooted_bra_bisimulation_safe}):
\begin{enumerate}
\item The right-hand side $x'$ of the positive premise in the first
  rule occurs $\Lambda$-liquid in $x'\seqc y$.
\item In both rules, the variable $x$ has only $\Lambda$-liquid
  occurrences.
\item In both rules, both variables $x$ and $y$ have $\aleph$-liquid
  occurrences in the source.
\item In both rules, only the variable $x$ has exactly one
  $\aleph$-liquid occurrence in the source that is also
  $\Lambda$-liquid.
  \begin{enumerate}
  \item The variable $x$ does not have more than one $\aleph$-liquid
    occurrence in the positive premises of the rules. It does have
    $\aleph$-liquid occurrences in negative premises of the second
    rule, but, since $\alpha$ ranges over all actions in $A_\tau$,
    there is also a premise $x\ntrans{\tau}$.
  \item Clearly the first rule with $\alpha=\tau$, which has a premise $x\trans\tau x'$
    with an $\aleph$-liquid occurrence of $x$, is
    $\aleph{\cap}\Lambda$-patient.
  \end{enumerate}
\end{enumerate}
Thm.~\ref{thm:congruence2}, which is presented at the end of
Sect.~\ref{sec:congruence}, will therefore imply the following congruence result.

\begin{corollary}
  $\bis[s]{rb}$ is a congruence for the sequencing operator.
\end{corollary}

\noindent
Note that we cannot take $\Lambda$ to be universal, for then by
condition \ref{main2} we would need the second rule for $\beta=\tau$
to be $\aleph{\cap}\Lambda$-patient.  Yet, as is easy to check,
$\bis[s]{b}$ is a congruence for sequencing.  This shows that our
(unrooted) stability-respecting branching bisimulation format does not
cover all relevant operators from the literature.

We argue that it is not straightforward to formulate a congruence format for
stability-preserving branching bisimilarity that includes the rules
for the sequencing operator. The second rule for the
sequencing operator is not a patience rule for two reasons: it has
negative premises and the target is not of the form $x \seqc{}
y'$. Putting $x\seqc{} y'$ as the target would not change the semantics in an
essential way, so only the first reason is relevant. The following
example shows that a mild generalisation of the notion of patience
rule allowing some negative premises would not work.

\begin{example}
Consider a binary function symbol $f$ defined by the following three rules:
\[
  \frac{x \trans{\tau} x'~~~~~~~~y \ntrans{\tau}}{f(x,y) \trans{\tau} f(x',y)}
  \qquad\qquad
  \frac{x \ntrans{\tau}~~~~~~~~y \trans{\tau} y'}{f(x,y) \trans{\tau} f(x,y')}
  \qquad\qquad
  \frac{x \trans{a} x'}{f(x,y) \trans{a} f(x',y)}
\]
Similar to the second argument of the sequencing operator, the `patience rules'
for the two arguments of $f$ both carry an additional negative
premise. Making both arguments $\Lambda$- and $\aleph$-liquid, all
requirements of the stability-respecting branching bisimulation format
are met, including \ref{main2} if the first two rules are considered
patience rules.

However, $\bis[s]{b}$ is not a congruence for $f$. Suppose $p \trans{\tau} p' \trans{a} p''$ and $q \trans{\tau} q'$.
Then clearly $p \bis[s]{b} p'$, but $f(p,q) \not\bis[s]{b} f(p',q)$ because $f(p,q)$ cannot
perform any transition while $f(p',q)$ can perform an $a$-transition.
\end{example}

\subsection{Correctness} \label{sec:correctness}

  To prove that (rooted) stability-respecting branching bisimilarity
  is indeed a congruence for every complete standard TSS $P$ in the
  (rooted) stability-respecting branching bisimulation format, we use
  the modal decomposition method discussed in
  Sect.~\ref{sec:decompmethod}. It suffices to establish that the
  transformation to ruloids preserves the syntactic restrictions of
  the format, and that the format, in turn, ensures that decomposition
  mappings associated with a formula in $\mathbb{O}^s_{b}$
  ($\mathbb{O}^s_{rb}$) assign to the variables again formulas in
  $\mathbb{O}^s_{b}$ ($\mathbb{O}^s_{rb}$).
  
\paragraph{Preservation of syntactic restrictions}

The definition of modal decomposition is based on
$P$-ruloids. Therefore we must verify that if $P$ is in rooted
stability-respecting branching bisimulation format, then the
$P$-ruloids are rooted stability-respecting branching bisimulation
safe (Prop.~\ref{prop:preservation_bra_bisimulation_safe}).  The key
part of the proof is to show that the syntactic restriction of decent
rooted stability-respecting branching bisimulation safety is preserved
under irredundant provability
(Lem.~\ref{lem:preservation_branching_bisimulation_safe}).

In the proof of the preservation lemma below, rules with a negative
conclusion will play an important role. Therefore the notion of rooted
stability-respecting branching bisimulation safety needs to be
extended to non-standard rules. The following definition coincides
with the definition of rooted branching bisimulation safe for
non-standard rules in \cite{FvGdW12}.

\begin{definition}\label{def:nonstandard-branching}
  An ntytt rule $r=\frac{H}{t\ntrans\alpha}$ is {\em rooted
    stability-respecting branching bisimulation safe} w.r.t.\ $\aleph$
  and $\Lambda$ if it satisfies conditions \ref{Lambda} and
  \ref{aleph} of Def.~\ref{def:rooted_bra_bisimulation_safe}.
\end{definition}

\begin{lemma}
  \label{lem:negative-tau}
  Let $Q$ be an $\aL$-patient TSS in decent ntyft format.  If an ntytt
  rule $\frac{H}{t\ntrans\tau}$ is provable from $Q^+$ and $x$ occurs
  $\aL$-liquid in $t$, then $H$ contains a premise
  \plat{$v\ntrans\tau$} where $x$ occurs $\aL$-liquid in $v$.
\end{lemma}

\begin{proof}
  Recall from Sect.~\ref{sec:ruloids} that also the TSS $Q^+$ is in
  decent ntyft format.  Let \plat{$\frac{H}{t\ntrans\tau}$} be
  provable from $Q^+$, by means of a proof $\pi$. We apply structural
  induction with respect to $\pi$.

  \vspace{2mm}

  \noindent {\em Induction basis}: The case where $\pi$ has only one
  node, marked ``hypothesis'', is trivial, as then $H$ contains
  $t\ntrans\tau$.

  \vspace{2mm}

  \noindent {\em Induction step}: Let $r$ be the decent ntyft rule and
  $\sigma$ the substitution used at the bottom of $\pi$. Let
  $f(x_1,\ldots,x_{\ar(f)})\ntrans{\tau}$ be the conclusion of
  $r$. Then $\sigma(f(x_1,\ldots,x_{\ar(f)}))=t$. Since $x$ occurs
  $\aL$-liquid in $t$, it occurs $\aL$-liquid in $\sigma(x_{i_0})$
  where $i_0$ is an $\aL$-liquid argument of $f$. In view of the
  patience rule for this argument of $f$ and the construction of rules
  with a negative premise it follows that $r$ has a premise
  $x_{i_0}\ntrans\tau$. So a rule
  $\frac{H'}{\sigma(x_{i_0})\ntrans{\tau}}$ is provable from $Q^+$ by
  means of a strict subproof of $\pi$, where $H'\subseteq H$. By
  induction $H'$ contains a premise \plat{$v\ntrans\tau$} where $x$
  occurs $\aL$-liquid in $v$.
\end{proof}

\begin{lemma}
  \label{lem:preservation_branching_bisimulation_safe}
  Let $P$ be an $\aL$-patient TSS in decent ntyft format, in which
  each rule is rooted stability-respecting branching bisimulation safe
  w.r.t.\ $\aleph$ and $\Lambda$.  Moreover, $P$ is either a standard
  TSS or a TSS of the form $Q^+$ with $Q$ an $\aL$-patient standard
  TSS in decent ntyft format.  Then each ntytt rule irredundantly
  provable from $P$ is rooted stability-respecting branching
  bisimulation safe w.r.t.\ $\aleph$ and $\Lambda$.
\end{lemma}

\begin{proof}
  For all conditions of Def.~\ref{def:rooted_bra_bisimulation_safe}
  and Def.~\ref{def:nonstandard-branching} except \ref{main1}, the
  preservation proof coincides with the corresponding proof for rooted
  branching bisimulation safeness in
  \cite[Lem.~5---aliased~4.5]{FvGdW12}. Therefore we here focus only
  on condition \ref{main1}.  Let an ntytt rule
  $\frac{H}{t\trans{\alpha} u}$ be irredundantly provable from $P$, by
  means of a proof $\pi$. We prove, using structural induction with
  respect to $\pi$, that this rule satisfies condition \ref{main1} of
  Def.~\ref{def:rooted_bra_bisimulation_safe}.

  \vspace{2mm}

  \noindent {\em Induction basis}: Suppose $\pi$ has only one node,
  marked ``hypothesis''. Then $\frac{H}{t\trans{\alpha} u}$ equals
  \plat{$\frac{t\trans{\alpha}u}{t\trans{\alpha}u}$} (so $u$ is a
  variable).  This rule trivially satisfies condition \ref{main1} of
  Def.~\ref{def:rooted_bra_bisimulation_safe}.

  \vspace{2mm}

  \noindent {\em Induction step}: Let $r$ be the decent ntyft rule and
  $\sigma$ the substitution used at the bottom of $\pi$. By
  assumption, $r$ is decent, ntyft, and rooted stability-respecting
  branching bisimulation safe w.r.t.\ $\aleph$ and $\Lambda$.  Let $r$
  be of the form
  \[
    \frac{\{v_k\trans{\beta_k}y_k\mid k\in
      K\}\cup\{w_\ell\ntrans{\gamma_\ell}\mid \ell\in L\}}
    {f(x_1,\ldots,x_{\ar(f)})\trans{\alpha}v}
  \]
  Then $\sigma(f(x_1,\ldots,x_{\ar(f)}))=t$ and
  $\sigma(v)=u$. Moreover, decent ntytt rules
  $r_k = \frac{H_k}{\sigma(v_k)\trans{\beta_k}\sigma(y_k)}$ for each
  $k\in K$ and
  $r_\ell = \frac{H_\ell}{\sigma(w_\ell)\ntrans{\gamma_\ell}}$ for
  each $\ell\in L$ are irredundantly provable from $P$ by means of
  strict subproofs of $\pi$, where
  $H=\bigcup_{k\in K}H_k\cup\bigcup_{\ell\in L}H_\ell$.  By induction,
  they are rooted stability-respecting branching bisimulation safe
  w.r.t.\ $\aleph$ and $\Lambda$.

  Suppose that $x$ has exactly one $\aleph$-liquid occurrence in $t$,
  which is also $\Lambda$-liquid. Then there is an
  $i_0\in \{1,\ldots,\ar(f)\}$ with $\aleph(f,i_0)$ and
  $\Lambda(f,i_0)$ such that $x$ has exactly one $\aleph$-liquid
  occurrence in $\sigma(x_{i_0})$, which is also
  $\Lambda$-liquid. Furthermore, for each
  $i\in \{1,\ldots,\ar(f)\}\backslash\{i_0\}$, $\neg\aleph(f,i)$ or
  $x$ occurs only $\aleph$-frozen in $\sigma(x_i)$. Since the ntyft
  rule $r$ is rooted stability-respecting branching bisimulation safe
  w.r.t.\ $\aleph$ and $\Lambda$, by condition \ref{aleph} of
  Def.~\ref{def:rooted_bra_bisimulation_safe}, if $\neg\aleph(f,i)$,
  then $x_i$ occurs only $\aleph$-frozen in $v_k$ for all $k\in K$, as
  well as in $w_\ell$ for all $\ell\in L$. We distinguish three
  possible cases, and argue each time that condition \ref{main1} of
  Def.~\ref{def:rooted_bra_bisimulation_safe} is satisfied.

  \vspace{2mm}

  \noindent {\sc Case 1}: $x_{i_0}$ has no $\aleph$-liquid occurrences
  in the premises of $r$. Then $x$ has no $\aleph$-liquid occurrences
  in $\sigma(v_k)$ for $k\in K$ and $\sigma(w_\ell)$ for $\ell\in
  L$. So by condition \ref{aleph} of
  Def.~\ref{def:rooted_bra_bisimulation_safe} and decency of the $r_k$
  and $r_\ell$, $x$ has no $\aleph$-liquid occurrences in $H$.

  \vspace{2mm}

  \noindent {\sc Case 2}: $x_{i_0}$ has exactly one $\aleph$-liquid
  occurrence in the premises of $r$, in $v_{k_0}$ for some $k_0\in
  K$. By condition \ref{Lambda} of
  Def.~\ref{def:rooted_bra_bisimulation_safe} this occurrence is also
  $\Lambda$-liquid. Then $x$ has exactly one $\aleph$-liquid
  occurrence in $\sigma(v_{k_0})$, and this occurrence is also
  $\Lambda$-liquid. So by condition \ref{main1} of
  Def.~\ref{def:rooted_bra_bisimulation_safe}, if $x$ has an
  $\aleph$-liquid occurrence in a negative premise in $H_{k_0}$ or
  more than one in the positive premises in $H_{k_0}$, then there is a
  premise \plat{$w\ntrans\tau$} in $H_{k_0}\subseteq H$ where $x$
  occurs $\aleph$-liquid in $w$.  Furthermore, $x$ has no
  $\aleph$-liquid occurrences in $\sigma(v_k)$ for
  $k\in K\setminus\{k_0\}$ and $\sigma(w_\ell)$ for $\ell\in L$. So by
  condition \ref{aleph} of
  Def.~\ref{def:rooted_bra_bisimulation_safe}, $x$ has no
  $\aleph$-liquid occurrences in $H_k$ for $k\in K\setminus\{k_0\}$
  and $H_\ell$ for $\ell\in L$.

  \vspace{2mm}

  \noindent {\sc Case 3}: $x_{i_0}$ occurs $\aleph$-liquid in
  $w_{\ell_0}$ for some $\ell_0\in L$ or has more than one
  $\aleph$-liquid occurrence in the $v_k$ for $k\in K$. Then, by
  condition \ref{main1} of
  Def.~\ref{def:rooted_bra_bisimulation_safe}, $x_{i_0}$ occurs
  $\aleph$-liquid in $w_{\ell_1}$ for some $\ell_1\in L$ with
  $\gamma_{\ell_1}=\tau$. By condition \ref{Lambda} of
  Def.~\ref{def:rooted_bra_bisimulation_safe} this occurrence is also
  $\Lambda$-liquid. It follows that $x$ occurs $\aL$-liquid in
  $\sigma(w_{\ell_1})$.  In case $P$ is a standard TSS the premise
  \plat{$w_\ell\ntrans{\gamma_\ell}$} occurs in $H$, and otherwise, by
  Lem.~\ref{lem:negative-tau}, $H_{\ell_1}\subseteq H$ contains a
  premise \plat{$w\ntrans\tau$} where $x$ occurs $\aL$-liquid in $w$.
\end{proof}

\noindent
The following proposition can now be proved in the same way as the
corresponding Prop.~2---aliased~4.6 ---for rooted branching
bisimulation safeness in \cite{FvGdW12}.

\begin{proposition}
  \label{prop:preservation_bra_bisimulation_safe}
  Let $P$ be an $\aL$-patient TSS in ready simulation format, in which
  each rule is rooted stability-respecting branching bisimulation safe
  w.r.t.\ $\aleph$ and $\Lambda$.  Then each $P$-ruloid is rooted
  stability-respecting branching bisimulation safe w.r.t.\ $\aleph$
  and $\Lambda$.
\end{proposition}

\paragraph{Preservation of modal characterisations}

Given a standard TSS in rooted stability-respecting branching
bisimulation format, w.r.t.\ some $\aleph$ and $\Lambda$,
Def.~\ref{def:decomposition} yields decomposition mappings
$\psi\mathbin\in t^{-\!1}\hspace{-1pt}(\phi)$, with $\Gamma:=\aL$. In
this section we prove that if $\phi\in\IO{b}^s$, then
$\psi(x)\in\IO{b}^{s\,\equiv}$ if $x$ occurs only $\Lambda$-liquid in
$t$. (That is why in the stability-respecting branching bisimulation
format, $\Lambda$ must be universal.) Furthermore, we prove that if
$\phi\in\IO{rb}^s$, then $\psi(x)\in\IO{rb}^{s\,\equiv}$ for all
variables $x$. From these preservation results we will deduce the promised congruence
results for unrooted and rooted stability-respecting branching
bisimilarity, respectively.

In \cite{FvGdW12} the following proposition was proved in two separate steps: Prop.~3 and
Prop.~4---aliased~4.7 and~4.8---in that paper. This was possible since $\IO{rb}$ incorporates $\IO{b}$ but not vice versa.
However, since there is a circular dependency between the definitions of $\IO{b}^s$ and
$\IO{rb}^s$, here we need to prove the corresponding results for these modal characterisations
simultaneously. The third part of the proposition below is merely a tool in proving the other parts.

\begin{proposition}\label{prop:rbrab-rbrab}
Let $P$ be an $\aL$-patient standard TSS in
ready simulation format, in which each rule is
rooted stability-respecting branching bisimulation safe w.r.t.\ $\aleph$ and $\Lambda$.
\begin{enumerate}
\item
For each term $t$ and variable $x$ that occurs only $\Lambda$-liquid in $t$:
\[
\phi\in\IO{b}^s\ \Rightarrow\ \forall\psi\in t^{-1}(\phi):\psi(x)\in\IO{b}^{s\,\equiv}
\]
\item
For each term $t$ and variable $x$:
\[
\overline{\phi}\in\IO{rb}^s\ \Rightarrow\ \forall\psi\in t^{-1}(\overline{\phi}):\psi(x)\in\IO{rb}^{s\,\equiv}
\]
\item
For each term $t$ and variable $x$ that occurs only $\Lambda$-liquid and $\aleph$-frozen in $t$:
\[
\overline{\phi}\in\IO{rb}^s\ \Rightarrow\ \forall\psi\in t^{-1}(\overline{\phi}):\psi(x)\in\IO{b}^{s\,\equiv}
\]
\end{enumerate}
\end{proposition}

\begin{trivlist} \item[\hspace{\labelsep}\bf Proof:]
We apply simultaneous induction on the structure of $\phi$, resp.\ $\overline\phi$, and the construction of $\psi$.
We only treat the case where $t$ is univariate. The case where $t$ is not univariate can be dealt with in
the same way as in the proofs of the corresponding Prop.~3 and Prop.~4---aliased~4.7 and~4.8---in \cite{FvGdW12}.
Let $\psi\in t^{-1}(\phi)$. If $x\mathbin{\notin}\var(t)$ then $\psi(x)\mathbin{\equiv}\top\mathbin\in\IO{b}^{s\,\equiv}$.
So suppose $x$ occurs exactly once in $t$.

We start with the first claim of the proposition.
The cases where $\phi$ is of the form $\bigwedge_{i\in I}\phi_i$ or $\neg\phi'$ can be dealt with as in the proof of
Prop.~3---aliased 4.7---in \cite{FvGdW12}. We therefore focus on the other cases.

\begin{itemize}
\item $\phi=\eps(\phi_1\diam{\hat\tau}\phi_2)$ with $\phi_1,\phi_2\in\IO{b}^s$. Let the occurrence of $x$ in $t$ be $\Lambda$-liquid. According to Def.~\ref{def:decomposition}.\ref{dec4} we can distinguish two cases.

\begin{description}
\item[{\sc Case 1:}] $\psi(x)$ is defined based on Def.~\ref{def:decomposition}.\ref{4a}. Then $\psi(x)=\eps\chi(x)$ if $x$ occurs $\aleph$-liquid in $t$, or $\psi(x)=\chi(x)$ if $x$ occurs $\aleph$-frozen in $t$, for some $\chi\in t^{-1}(\phi_1\diam{\hat\tau}\phi_2)$. By Def.~\ref{def:decomposition}.\ref{dec1}, $\chi(x)=\chi_1(x)\land\chi_2(x)$ with $\chi_1\in t^{-1}(\phi_1)$ and $\chi_2\in t^{-1}(\diam{\hat\tau}\phi_2)$. By induction on formula size, $\chi_1(x)\in\IO{b}^{s\,\equiv}$. For $\chi_2(x)$, according to Def.~\ref{def:decomposition}.\ref{dec5}, we can distinguish three cases. Cases 1.1 and 1.2 where $\chi_2(x)$ is defined based on Def.~\ref{def:decomposition}.\ref{5a} and Def.~\ref{def:decomposition}.\ref{5b}, respectively, proceed in the same way is in the proof of \cite[Prop.~3]{FvGdW12}. We focus on the third case.

\item[{\sc Case 1.3:}]  $\chi_2(x)$ is defined based on
  Def.~\ref{def:decomposition}.\ref{5c}, employing an $\aL$-impatient $P$-ruloid
  $\frac{H}{t\trans{\tau}u}$ and a $\xi\in u^{-1}(\phi_2)$.
  So $\chi_2(x) = \xi(x)\land\bigwedge_{x\trans{\beta}y\in H}\diam{\beta}\xi(y)
              \land\bigwedge_{x\ntrans{\gamma}\in H}\neg\diam{\gamma}\top$.
  By Prop.~\ref{prop:preservation_bra_bisimulation_safe},
  $\frac{H}{t\trans{\tau}u}$ is rooted stability-respecting branching bisimulation safe.
  Since the occurrence of $x$ in $t$ is $\Lambda$-liquid, by condition
  \ref{Lambda} of Def.~\ref{def:rooted_bra_bisimulation_safe}, $x$
  occurs only $\Lambda$-liquid in $u$. Therefore, by induction on formula size,
  $\xi(x)\in\IO{b}^{s\,\equiv}$. Case 1.3.1 where
  the occurrence of $x$ in $t$ is $\aleph$-frozen still proceeds in the same way is in the proof of \cite[Prop.~3]{FvGdW12}.
  We focus on the other case.
  
\item[{\sc Case 1.3.2:}]
  The occurrence of $x$ in $t$ is $\aleph$-liquid. If $H$ has at most one premise of the form \plat{$x\trans{\beta}y$}, for which
  $\beta\neq\tau$, and none of the form $x\ntrans{\gamma}$, then still the proof proceeds in the same way is in the proof of \cite[Prop.~3]{FvGdW12}.
  However, the more relaxed condition \ref{main1} of Def.~\ref{def:rooted_bra_bisimulation_safe} allows $H$ to have more than one
  premise of the form \plat{$x\trans{\beta}y$} with $\beta\neq\tau$, or premises $x\ntrans{\gamma}$. Only, then $H$ must also contain
  the premise $x\ntrans{\tau}$. Thus 
  $\psi(x) \equiv \eps \big(\neg\diam{\tau}\top\land\chi_1(x)\land \xi(x)\land
  \bigwedge_{x\trans{b}y\in H}\diam{b}\xi(y)\land\bigwedge_{x\ntrans{c}\in H}\neg\diam{c}\top\big)$.
  By condition \ref{rhs} of Def.~\ref{def:rooted_bra_bisimulation_safe} the right-hand sides $y$ of positive premises in $H$ occur only
  $\Lambda$-liquid in $u$, so by induction $\xi(y)\in\IO{b}^{s\,\equiv}$. Hence the conjuncts $\diam{b}\xi(y)$
  are in $\IO{rb}^{s\,\equiv}$. It follows that $\psi(x)\in\IO{b}^{s\,\equiv}$.

\item[{\sc Case 2:}] $\psi(x)$ is defined based on
  Def.~\ref{def:decomposition}.\ref{4b}, employing an $\aL$-impatient $P$-ruloid
  $\frac{H}{t\trans{\tau}u}$ and a $\chi\in u^{-1}(\eps(\phi_1\diam{\hat\tau}\phi_2))$.
  By Prop.~\ref{prop:preservation_bra_bisimulation_safe},
  $\frac{H}{t\trans{\tau}u}$ is rooted stability-respecting branching bisimulation safe.
  Since the occurrence of $x$ in $t$ is $\Lambda$-liquid, by condition
  \ref{Lambda} of Def.~\ref{def:rooted_bra_bisimulation_safe}, $x$
  occurs only $\Lambda$-liquid in $u$. Therefore, by induction on the
  construction of $\psi$, $\chi(x)\in\IO{b}^\equiv$. Case 2.1 where
  the occurrence of $x$ in $t$ is $\aleph$-frozen proceeds in the same way is in the proof of \cite[Prop.~3]{FvGdW12}.
  We focus on the other case.

\item[{\sc Case 2.2:}] \hypertarget{2.2}{The occurrence of $x$ in $t$ is $\aleph$-liquid.}
  Then
  $\psi(x)=\eps\big(\chi(x)\ \land\bigwedge_{x\trans{\beta}y\in H}\diam{\beta}\chi(y)
              \land \bigwedge_{x\ntrans{\gamma}\in H}\neg\diam{\gamma}\top\big)$.
  If $H$ has at most one premise of the form \plat{$x\trans{\beta}y$}, for which
  $\beta\neq\tau$, and none of the form $x\ntrans{\gamma}$, then still the proof proceeds in the same way is in the proof of \cite[Prop.~3]{FvGdW12}.
  However, the more relaxed condition \ref{main1} of Def.~\ref{def:rooted_bra_bisimulation_safe} allows $H$ to have more than one
  premise of the form \plat{$x\trans{\beta}y$} with $\beta\neq\tau$, or premises $x\ntrans{\gamma}$. Only, then $H$ must also contain
  the premise $x\ntrans{\tau}$. Thus
  $\psi(x) \equiv \eps\big(\neg\diam{\tau}\top\land\chi(x)\land\bigwedge_{x\trans{b}y\in H}\diam{b}\chi(y)
              \land\bigwedge_{x\ntrans{c}\in H}\neg\diam{c}\top\big)$.
  By condition \ref{rhs} of Def.~\ref{def:rooted_bra_bisimulation_safe} the right-hand sides $y$ of positive premises in $H$ occur only
  $\Lambda$-liquid in $u$, so by induction $\xi(y)\in\IO{b}^{s\,\equiv}$. Hence the conjuncts $\diam{b}\chi(y)$
  are in $\IO{rb}^{s\,\equiv}$. It follows that $\psi(x)\in\IO{b}^{s\,\equiv}$.
\end{description}
\end{itemize}

\noindent
The proof of the case $\phi=\eps(\phi_1\diam{a}\phi_2)$ in $\IO{b}$
from \cite[Prop.~3]{FvGdW12} needs to be adapted in a similar fashion as
the case $\phi=\eps(\phi_1\diam{\hat\tau}\phi_2)$. We take the liberty to omit this adaptation here, and continue with
the additional clause in the modal characterisation of $\IO{b}^s$, compared to $\IO{b}$.

\begin{itemize}

\item $\phi=\eps(\neg\diam{\tau}\top\land\,\overline{\phi})$ with $\overline{\phi}\in\IO{rb}^s$. Let the occurrence of $x$ in $t$ be $\Lambda$-liquid. According to Def.~\ref{def:decomposition}.\ref{dec4} we can distinguish two cases.

\begin{description}
\item[{\sc Case 1:}] $\psi(x)$ is defined based on Def.~\ref{def:decomposition}.\ref{4a}. Then $\psi(x)=\eps\chi(x)$ if $x$ occurs $\aleph$-liquid in $t$, or $\psi(x)=\chi(x)$ if $x$ occurs $\aleph$-frozen in $t$, for some $\chi\in t^{-1}(\neg\diam{\tau}\top\land\,\overline{\phi})$. By Def.~\ref{def:decomposition}.\ref{dec1}, $\chi(x)=\chi_1(x)\land\chi_2(x)$ with $\chi_1\in t^{-1}(\neg\diam{\tau}\top)$ and $\chi_2\in t^{-1}(\overline{\phi})$. By Def.~\ref{def:decomposition}.\ref{dec2}
there is a function $h:t^{-1}(\diam{\tau}\top)\rightarrow\var(t)$ such that $\chi_1(x)=\wedge_{\xi\in h^{-1}(x)}\neg\xi(x)$.

\item[{\sc Case 1.1:}] $x$ occurs $\aleph$-liquid in $t$. By Def.~\ref{def:decomposition}.\ref{dec3}, for
  each $\xi\in h^{-1}(x)$, $\xi(x)$ is of the form $\bigwedge_{x\trans{\beta}y\in H}\diam{\beta}\top\wedge\bigwedge_{x\ntrans{\gamma}\in H}\neg\diam{\gamma}\top$ for some
  $P$-ruloid $\frac{H}{t\trans{\tau}u}$. Note that such formulas are in $\IO{rb}^s$. Moreover, by
  induction on formula size, $\chi_2(x)\in\IO{rb}^{s\,\equiv}$. Since the occurrence of $x$ in $t$
  is $\aL$-liquid, there is an $\aL$-patient ruloid $\frac{x\trans{\tau}y}{t\trans{\tau} t'}$. This gives
  rise to a $\xi\in h^{-1}(x)$ such that $\xi(x)=\neg\diam{\tau}\top$. Concluding,
  $\psi(x)=\eps(\chi_1(x)\land\chi_2(x))$ is of the form
  $\eps(\neg\diam{\tau}\top\land\overline{\phi}')$ for some $\overline{\phi}'\in\IO{rb}^s$. So
  $\psi(x)\in\IO{b}^s$.

\item[{\sc Case 1.2:}] $x$ occurs $\aleph$-frozen in $t$. Then by condition \ref{aleph} of
  Def.~\ref{def:rooted_bra_bisimulation_safe}, $x$ does not occur in $H$. Since moreover
  $\omega(x)\equiv\top$ for each $\omega\in t^{-1}(\top)$, by
  Def.~\ref{def:decomposition}.\ref{dec3}, $\xi(x)\equiv\top$ for each $\xi\in h^{-1}(x)$. So either
  $\chi_1(x)\equiv\neg\top$ if $h^{-1}(x)$ is non-empty or $\chi_1(x)\equiv\top$ if $h^{-1}(x)$ is
  empty. In the first case $\psi(x)\equiv\neg\top$, so then we are done. In the second case,
  by induction on formula size, using the third claim of this proposition, $\psi(x)\equiv\chi_2(x)\in\IO{b}^{s\,\equiv}$.

\item[{\sc Case 2:}] $\psi(x)$ is defined based on Def.~\ref{def:decomposition}.\ref{4b}.
This case proceeds in the same way as case 2 of $\phi=\eps(\phi_1\diam{\hat\tau}\phi_2)$.
\end{description}
\end{itemize}
This completes the proof of the first claim of the proposition. We continue with the second claim of the proposition.
The cases where $\overline{\phi}$ is of the form $\bigwedge_{i\in I}\overline{\phi}_i$ or $\neg\overline{\phi}'$ or $\diam\alpha\phi$ can be dealt with as in the proof of Prop.~4 in \cite{FvGdW12}. We focus on the only other case.

\begin{itemize}
\item
$\overline{\phi}\in\IO{b}^s$. The cases where $\overline{\phi}$ is of the form $\bigwedge_{i\in I}\phi_i$ or $\neg\phi'$ or $\eps(\phi_1\diam{\hat\tau}\phi_2)$ or $\eps(\phi_1\diam{a}\phi_2)$ can be dealt with as in the proof of Prop.~4 in \cite{FvGdW12}. We focus on the only new case here.

\item[$\ast$]
$\overline{\phi}=\eps(\neg\diam{\tau}\top\land\overline{\phi}')$ with $\overline{\phi}'\in\IO{rb}^s$. If the occurrence of $x$ in $t$ is $\Lambda$-liquid, then we already proved in the corresponding case for the first claim of the proposition that $\psi(x) \in \IO{b}^{s\,\equiv}\subset\IO{rb}^{s\,\equiv}$. So we can assume that this occurrence is $\Lambda$-frozen. According to Def.~\ref{def:decomposition}.\ref{dec4} we can distinguish two cases.

\begin{description}
\item[{\sc Case 1:}] $\psi(x)$ is defined based on Def.~\ref{def:decomposition}.\ref{4a}. Then, since $x$ occurs $\Lambda$-frozen in $t$, $\psi(x)=\chi(x)$ for some $\chi\in t^{-1}(\neg\diam{\tau}\top\land\overline{\phi}')$. By Def.~\ref{def:decomposition}.\ref{dec1}, $\chi(x)=\chi_1(x)\land\chi_2(x)$ with $\chi_1\in t^{-1}(\neg\diam{\tau}\top)$ and $\chi_2\in t^{-1}(\overline{\phi}')$. By induction on formula size, $\chi_1(x)\in\IO{rb}^{s\,\equiv}$ and $\chi_2(x)\in\IO{rb}^{s\,\equiv}$. Hence, $\psi(x)=\chi_1(x)\land\chi_2(x)$ is in $\IO{rb}^{s\,\equiv}$.

\item[{\sc Case 2:}] $\psi(x)$ is defined based on Def.~\ref{def:decomposition}.\ref{4b}, using an $\aL$-impatient $P$-ruloid $\frac{H}{t\trans\tau u}$ and a $\chi\in u^{-1}(\eps(\neg\diam{\tau}\top\land\overline{\phi}'))$. As the occurrence of $x$ in $t$ is $\Lambda$-frozen,\\
$\psi(x)=\chi(x)\land \bigwedge_{x\trans\beta y\in H}\diam{\beta}\chi(y)\land \bigwedge_{x\ntrans\gamma\in H}\neg\diam{\gamma}\top$.
By induction on the construction of $\psi$, $\chi(x)\in\IO{rb}^\equiv$. Moreover, by condition \ref{rhs} of Def.~\ref{def:rooted_bra_bisimulation_safe} the $y$ occur only $\Lambda$-liquid in $u$, so we proved before that the $\chi(y)$ are in $\IO{b}^{s\,\equiv}$. Hence $\psi(x)\in\IO{rb}^{s\,\equiv}$.
\end{description}
\end{itemize}
This completes the proof of the second claim of the proposition. We finish with the last claim.
The cases where $\overline\phi$ is of the form $\bigwedge_{i\in I}\overline\phi_i$ or $\neg\overline\phi'$ can be dealt with as in the proof of
Prop.~3 in \cite{FvGdW12}, and the case $\overline\phi=\phi\in\IO{b}$ follows immediately from the first claim.
We therefore focus on the remaining case: $\overline\phi=\diam\alpha\phi$ with $\phi\in\IO{b}$.

$\psi(x)$ is defined based on Def.~\ref{def:decomposition}.\ref{dec3}, for some $P$-ruloid $\frac{H}{t\trans{\alpha}u}$
and $\chi\in u^{-1}(\phi)$. Since the occurrence of $x$ in $t$ is $\aleph$-frozen, by condition \ref{aleph} of
Def.~\ref{def:rooted_bra_bisimulation_safe}, $x$ does not occur in $H$. Hence,
$\psi(x)=\chi(x)$. By Prop.~\ref{prop:preservation_bra_bisimulation_safe}, $\frac{H}{t\trans{\alpha}u}$
is rooted stability-respecting branching bisimulation safe.
  Since the occurrence of $x$ in $t$ is $\Lambda$-liquid, by condition
  \ref{Lambda} of Def.~\ref{def:rooted_bra_bisimulation_safe}, $x$
  occurs only $\Lambda$-liquid in $u$. Therefore, by the first claim of this proposition, $\chi(x)\in\IO{b}^\equiv$.
\qed
\end{trivlist}

Now the promised congruence results for $\bis[s]{b}$ and $\bis[s]{rb}$ can be proved in the same way as their counterparts for $\bis{b}$ and $\bis{rb}$ in \cite{FvGdW12}.

\begin{theorem}\label{thm:congruence1}
Let $P$ be a complete standard TSS in stability-respecting branching bisimulation format. Then $\bis[s]{b}$ is a congruence for $P$.
\hfill $\Box$
\end{theorem}

\begin{theorem}\label{thm:congruence2}
Let $P$ be a complete standard TSS in rooted stability-respecting branching bisimulation format. Then $\bis[s]{rb}$ is a congruence for $P$.
\hfill $\Box$
\end{theorem}

\section{Divergence-preserving branching bisimilarity as a congruence} \label{sec:dpbbcong}

To get a modal characterisation of (weakly) divergence-preserving
  branching bisimilarity, a modality that captures divergence needs to be added to the modal
  logic. A modal characterisation of
  divergence-preserving branching bisimilarity would be obtained by
  adding a unary modality $\Delta$ to the 
  modal logic for branching bisimilarity, with the interpretation
    $p \models \Delta\varphi$ if there is an infinite trace
    $p=p_0\trans{\tau}p_1\trans{\tau}p_2\trans{\tau}\cdots$
  such that $p_i\models\varphi$ for all $i\in\N$.

Modal formulas $\Delta\phi$ to capture divergence elude the inductive decomposition method from
Def.~\ref{def:decomposition}, because they ask for the existence of an infinite sequence of
$\tau$-transitions.
The following example shows that the decomposition method does not readily extend to modalities $\Delta\phi$.

  \begin{example}\label{exa:parallelcomp}
Let $A=\{a_i,\overline{a_i}\mid i\in\N\}$ and let
$A_{\iota}=A\cup\{\iota\}$. The interleaving parallel composition operator $\|$ is usually defined by the following two rules, where $\alpha$ ranges over $A_{\iota}\cup\{\tau\}$:
\[
\frac{x_1\trans\alpha y}{x_1\|x_2\trans\alpha y\|x_2} \qquad\qquad \frac{x_2\trans\alpha y}{x_1\|x_2\trans\alpha x_1\|y}
\]
We extend the operational semantics of this operator with communication rules, for all $i\in\N$:
\[
\frac{x_1\trans{a_i}y_1~~~~~~x_2\trans{\overline{a_i}}y_2}{x_1\|x_2\trans
  \iota y_1\|y_2}
\qquad\qquad
\frac{x_1\trans{\overline{a_i}}y_1~~~~~~x_2\trans{a_i}y_2}{x_1\|x_2\trans \iota y_1\|y_2}
\]
Consider the following two collections of processes $p_i,q_i$ for $i\in\N$. We define $p_i\trans\tau p_{i+1}$ and $p_i\trans{a_i}{\bf 0}$ and $q_i\trans\tau q_{i+1}$ and \plat{$q_i\trans{\overline{a_i}}{\bf 0}$} for all $i\in\N$. We have $p_0\|q_0\models\Delta(\eps\diam{\iota}\top)$. There is no obvious way to decompose this modal property of $p_0\|q_0$ into modal properties of its arguments $p_0$ and $q_0$.
\end{example}%

\noindent
  Rather than modifying the decomposition method in order to derive
  new congruence formats for (weakly) divergence-preserving branching
  bisimilarity, we shall argue in this section that the
  stability-respecting branching bisimilarity format is also a
  congruence format for weakly divergence-preserving branching
  bisimilarity and divergence-preserving branching bisimilarity, and
  similarly for their rooted variants.

  Suppose that $P$ is a $\Gamma$-patient TSS (for some predicate $\Gamma$) on which $\bis[s]{b}$ is a congruence,
  and suppose that we wish to show that $\bis[\Delta\top]{b}$ is a
  congruence on $P$ too. Then, using that
  ${\bis[\Delta\top]{b}}\subseteq{\bis[s]{b}}$, it suffices to argue that,
  for all $p_1,\dots,p_n$ and $q_1,\dots,q_n$ such that
  $p_i\bis[\Delta\top]{b}q_i$ ($1\leq i \leq n$), we have that
  $f(p_1,\dots,p_n)\bis[s]{b}f(q_1,\dots,q_n)$ implies
  $f(p_1,\dots,p_n)\bis[\Delta\top]{b} f(q_1,\dots,q_n)$. The feature
  that distinguishes $\bis[\Delta\top]{b}$ from $\bis[s]{b}$ is that
  $\bis[\Delta\top]{b}$ preserves divergences whereas $\bis[s]{b}$
  does not. Now suppose that $f(p_1,\dots,p_n)$ admits a
  divergence. Then we can distinguish two cases:
  \begin{enumerate}
  \item\label{case:patience}
    the process $f(p_1,\dots,p_n)$ directly inherits this
    divergence from one of its arguments through a patience rule, say,
    $p_i$ admits a divergence and the $i$th argument of $f$ is
    $\Gamma$-liquid; or
  \item none of the $p_i$ for $i$ a $\Gamma$-liquid argument of $f$
    admits a divergence, but somehow there is an interplay between $f$
    and its arguments $p_1,\dots,p_n$ that facilitates the divergence.
  \end{enumerate}
  We want to argue that in both cases $f(q_1,\dots,q_n)$ necessarily admits a
  divergence as well. In the first case, using that $p_i
  \bis[\Delta\top]{b} q_i$, also $q_i$ admits a divergence. So by the
  patience rule for the $i$th argument of $f$, $f(q_1,\dots,q_n)$
  admits a divergence. It hence suffices to reduce the second case to
  the first. The following example roughly illustrates how such a
  reduction can be achieved.
  
  \begin{example}
    Let $A=\{a,\overline{a}\}$ and consider the parallel composition
    operator $\mid$ of CCS for this set of actions $A$
    defined by the following four rules (with $\alpha$ ranging over
    $A_{\tau}$): \
    \[
      \frac{x_1\trans\alpha y}{x_1\mathbin{\mid}x_2\trans\alpha
        y\mathbin{\mid}x_2} \qquad \frac{x_2\trans\alpha
        y}{x_1\mathbin{\mid}x_2\trans\alpha x_1\mathbin{\mid}y} \qquad
      \frac{x_1\trans{a}y_1~~~~~~x_2\trans{\overline{a}}y_2}{x_1\mathbin{\mid}x_2\trans\tau
        y_1\mathbin{\mid}y_2} \qquad
      \frac{x_1\trans{\overline{a}}y_1~~~~~~x_2\trans{a}y_2}{x_1\mathbin{\mid}x_2\trans\tau
        y_1\mathbin{\mid}y_2}
    \]
    Note that the operational semantics of $\mid$ includes patience
    rules for both arguments (specific instances of the two left-most
    rules), but also rules that may lead to divergences not directly
    inherited from arguments. For instance, if $p_1$ is the process with
    just one transition $p_1\trans{a}p_1$ and $p_2$ is the process with just
    one transition \plat{$p_2\trans{\overline{a}}p_2$}, then $p_1\mathbin{\mid}p_2$
    admits a divergence, while $p_1$ and $p_2$ do not admit divergences.

    Now also consider the parallel composition operator $\|$ defined in
    Ex.~\ref{exa:parallelcomp} in combination with a unary
    abstraction operator $\tau_{\iota}$ defined by the following two
    rules, where $\alpha$ ranges over $A_{\tau}$ (i.e., $\alpha\neq\iota$):
    \[
      \frac{x\trans\alpha y}{\tau_{\iota}(x)\trans\alpha
        \tau_{\iota}(y)}
    \qquad\qquad
      \frac{x\trans{\iota}y}{\tau_{\iota}(x)\trans\tau \tau_{\iota}(y)} 
    \]

    \noindent
    It is not hard to see that in a TSS $P$ that includes the CCS
    parallel composition $\mid$, the parallel composition operator 
    $\|$ as defined in Ex.~\ref{exa:parallelcomp} and the unary
    abstraction operator $\tau_{\iota}$, we have, for all processes
    $p_1$ and $p_2$, that ${p_1\mathbin{\mid}p_2} \bis{}
    {\tau_{\iota}(p_1\mathbin{\|}p_2)}$. Hence, since
    ${\bis{}}\subseteq{\bis[\Delta\top]{b}}$, in order to show that
    $\bis[\Delta\top]{b}$ is a congruence for $\mid$ it suffices to
    establish that it is a congruence for $\|$ and that it is a
    congruence for $\tau_{\iota}$. Namely, then $p_1\bis[\Delta\top]{b}\!q_1$ and
    $p_2\bis[\Delta\top]{b}\!q_2$ yield $p_1\mathbin{\mid}p_2\bis{}\tau_{\iota}(p_1\mathbin{\|}p_2)
    \bis[\Delta\top]{b}\!\tau_{\iota}(q_1\mathbin{\|}q_2)\bis{}q_1\mathbin{\mid}q_2$.
    Moreover, it is well-known and easy to see that
    $\tau_{\iota}$ is compositional for a wide range of
    behavioural equivalences, including the ones mentioned in
    Sect.~\ref{sec:equivalences_terms}.
    Since all operational rules for $\|$ with a $\tau$-labelled
    conclusion are patience rules, only case~\ref{case:patience} of
    $p_1\|p_2$ admitting a divergence can apply.
   \end{example}

  \noindent
  The idea discussed in the example above can be
  generalised: every operator $f$ can be expressed as a combination of $\tau_{\iota}$ and
  an \emph{abstraction-free variant} $f'$ of $f$, for which all rules with a $\tau$-labelled
  conclusion are patience rules. The main ingredient of our proof that
  on every TSS $P$ in the stability-respecting branching bisimulation format
  both $\bis[\Delta\top]{b}$ and $\bis[\Delta]{b}$ are congruences,
  and similarly for the rooted case, is a transformation that
  allows us to establish compositionality of an operator $f$ by
  establishing the compositionality of its abstraction-free variant.

  Although the idea is fairly simple, the formal technicalities are
  quite involved. It is therefore convenient to first formulate, in
  Sect.~\ref{sec:Framework}, sufficient conditions on a presupposed
  transformation on TSSs and associated encoding and decoding
  functions for them to be suitable for lifting a congruence format
  for some behavioural equivalence to a finer equivalence.
Instead of dealing with specific equivalences $\bis[\Delta]{b}$ and $\bis[s]{b}$,
we work with parametric equivalences $\sim$ and $\approx$ where $\sim$ is finer than $\approx$; they
will later be instantiated with
$\bis[\Delta]{b}$ and $\bis[s]{b}$. This allows a reuse of our work with $\bis[\Delta\top]{b}$ and
$\bis[s]{b}$ in the roles of $\sim$ and $\approx$, as well as with $\bis[\Delta\top]{rb}$ and
$\bis[\Delta]{rb}$ in the role of $\sim$ and $\bis[s]{rb}$ in the role of $\approx$.
In Sect.~\ref{sec:Abstraction-free} we introduce the machinery needed in Sect.~\ref{sec:Framework}
and show that it has the required properties, except for those that depend on the choice
of $\sim$ and $\approx$. Finally, in Sect.~\ref{sec:Apply} we apply our framework to derive congruence
formats for $\bis[\Delta\top]{b}$, $\bis[\Delta]{b}$, $\bis[\Delta\top]{rb}$ and $\bis[\Delta]{rb}$.

\subsection{A general framework for lifting congruence formats to finer equivalences}
\label{sec:Framework}

Our general proof idea is illustrated in Fig.~\ref{fig:Proof idea}.
Here $P$ (on the left) is a TSS in our congruence format for $\approx$.
We want to show that on $P$ also $\sim$ is a congruence.
So consider an operator $f$, for simplicity depicted as unary.
Given two processes $p$ and $q$ in $P$ (closed terms in its signature) with $p \sim q$,
we need to show that $f(p)\sim f(q)$. Fig.~\ref{fig:Proof idea} shows a roundabout trajectory from
$p$ to $f(p)$. Imagine a similar trajectory from $q$ to $f(q)$---not depicted in
Fig.~\ref{fig:Proof idea}, but hovering above the page.

\begin{figure}[t]
\centering
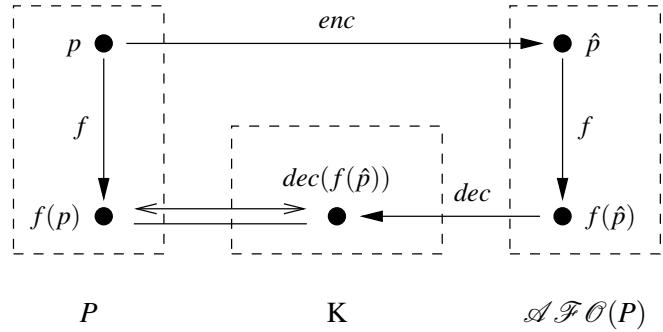
\caption{Proof idea \label{fig:Proof idea}}
\end{figure}

First we apply a transformation $\AFO$ on $P$. The TSS $\AFO(P)$ will be \emph{abstraction-free} in the sense
that it only allows patience rules and rules without premises to carry a conclusion with the label
$\tau$. Moreover, the transformation $\AFO$ needs to be such that
on $\AFO(P)$ the behavioural equivalences $\sim$ and $\approx$ coincide, and to realise this
property it will introduce \emph{oracle transitions}  that reveal some
pertinent information on the behaviour of a process, such as whether it can diverge.
We ensure that $\sim$-equivalence is preserved under the addition of these oracle transitions.

The transformation $\AFO$ will modify the rules of $P$, but we still
want to find for each closed term $p$ of $P$ a faithful representant
of $p$ inside $\AFO(P)$. To this end we introduce for each closed term $p$ of $P$
a constant $\hat p$ in $\AFO(P)$, in such a way that $p \sim q$
implies $\hat p \sim \hat q$.\footnote{In all our applications we have
$p \sim q$ iff $\hat p \sim \hat q$, but our general framework does not require this.}
  Note that by including these processes as
  constants, added oracle transition $\hat p \trans{\omega}\surd$ are decent ntyft rules without premises, which
  therefore do not compromise the $\approx$-congruence format that $\AFO(P)$
  needs to satisfy.
Each $n$-ary operator $f$ of $P$ remains an $n$-ary operator $f$ of $\AFO(P)$.
Since ${\sim}\subseteq{\approx}$ we have $\hat p \approx \hat q$.
We argue that if $P$ is within our congruence format for $\approx$, then
the TSS $\AFO(P)$ is also within this congruence format,
and conclude from $\hat p \approx \hat q$
that $f(\hat p) \approx f(\hat q)$. An important result, deferred to Sect.~\ref{sec:Apply}, is that on $\AFO(P)$
the equivalences $\sim$ and $\approx$ coincide. Hence $f(\hat p) \sim f(\hat q)$.

Finally, we decode the processes $f(\hat p)$ and $f(\hat q)$, aiming to return to
the LTS generated by $P$, but actually ending up in another LTS $\K$. Our decoding function
${\it dec}$ exactly undoes the effects of the encoding ${\it enc}$, which sent $p$ to $\hat p$,
so that ${\it dec}(f(\hat p))$ is strongly bisimilar with $f(p)$.
A crucial property of the function ${\it dec}$, also deferred to Sect.~\ref{sec:Apply},
is that it is compositional for $\sim$, meaning that from
$f(\hat p) \sim f(\hat q)$ we may conclude
${\it dec}(f(\hat p)) \sim {\it dec}(f(\hat q))$. By imposing the requirement
that $\sim$ contains strong bisimilarity, this implies that $f(p)\sim f(q)$.

We now formalise this proof idea. 
An LTS $G$ is called \emph{disjoint} from a complete TSS $P$ if it is disjoint from the LTS
$G_P$ associated with $P$. In that case $P\uplus G$ denotes the union of $G_P$ and $G$ (cf.\ Def.~\ref{def:disjoint}).

\begin{theorem}\label{thm:Main}
  Let $\sim$ and $\approx$ be behavioural equivalences on LTSs, with ${\bis{}} \subseteq {\sim} \subseteq {\approx}$.
  Let $\mathfrak{F}$ be a congruence format for $\approx$, included in the decent ntyft format,
  and let $\AFO$ be an operation on standard TSSs, where for each TSS $P=(\Sigma,{\it Act},R)$
  the signature $\hat \Sigma$ of $\AFO(P)$ contains $\Sigma$ enriched by a fresh constant $\hat p$ for each closed term $p$
  in $\mbox{T}(\Sigma)$, such that, for each complete standard TSS $P$ in decent ntyft format:
  \begin{enumerate}
  \item\label{Completeness} also $\AFO(P)$ is a complete standard TSS,
  \item\label{Preserved under AFO} if $P$ is in $\mathfrak{F}$-format then so is $\AFO(P)$,
  \item\label{Enc} $p \sim_P q ~~\Rightarrow~~ \hat p \sim_{\AFO(P)} \hat q$,
  \item\label{Coincide} $\sim_{\AFO(P)}$ and $\approx_{\AFO(P)}$ coincide, and
  \setcounter{saveenumi}{\theenumi}
  \end{enumerate}
  there is an LTS $\K=(\mathbb{P}_{\K},Act_\K,\rightarrow_{\K})$, disjoint from $P$,
  as well as a function ${\it dec}:\mbox{T}(\hat \Sigma) \rightarrow \mathbb{P}_{\K}$ such that:
  \begin{enumerate}
  \setcounter{enumi}{\thesaveenumi}
  \item\label{dec} $p \sim_{\AFO(P)} q  ~~\Rightarrow~~ {\it dec}(p) \sim_{\K} {\it dec}(q)$, ~and
  \item\label{Enc-dec} $f(p_1,\dots,p_n) \bis{P\uplus\K} {\it dec}(f(\hat p_1,\dots,\hat p_n))$
    for any $n$-ary $f\in\Sigma$ and $p_1,\dots,p_n\in \mbox{T}(\Sigma)$.
  \end{enumerate}
  Then $\mathfrak{F}$ is also a congruence format for $\sim$.
\end{theorem}

\begin{proof}
  Let $P=(\Sigma,{\it Act},R)$ be a complete standard TSS in $\mathfrak{F}$-format.
  We will show that $\sim_P$ is a congruence for $P$.
  So let $f\in\Sigma$ be an $n$-ary function symbol,
  and let $p_i,q_i\in \mbox{T}(\Sigma)$ with $p_i \sim_P q_i$ for $i=1,\dots,n$.
  We need to show that $f(p_1,\dots,p_n) \sim_P f(q_1,\dots,q_n)$.

  By requirements~\ref{Completeness} and~\ref{Preserved under AFO},
  $\AFO(P)$ is a complete standard TSS in $\mathfrak{F}$-format; hence
  $\approx_{\AFO(P)}$ is a congruence for $\AFO(P)$.
  By requirement~\ref{Enc} $\hat p_i\sim_{\AFO(P)}\hat q_i$ for $i=1,\dots,n$.
  By requirement~\ref{Coincide} also $\sim_{\AFO(P)}$ is a congruence for $\AFO(P)$.
  So $f(\hat p_1,\dots,\hat p_n) \sim_{\AFO(P)} f(\hat q_1,\dots,\hat q_n)$.
  Hence, by requirement~\ref{dec}, $${\it dec}(f(\hat p_1,\dots,\hat p_n))
  \sim_{\K} {\it dec}(f(\hat q_1,\dots,\hat q_n)).$$
  Therefore, by two applications of
  requirement~\ref{Enc-dec}, and the definition of a behavioural equivalence,\\
  $f(p_1,\dots,p_n) \sim_{P\uplus\K} f(q_1,\dots,q_n)$
  and consequently $f(p_1,\dots,p_n) \sim_{P} f(q_1,\dots,q_n)$.
\end{proof}

\subsection{Abstraction-freeness}
\label{sec:Abstraction-free}

In this section we introduce the machinery needed for Thm.~\ref{thm:Main}, namely
the conversion $\AFO$ on TSSs and the function {\it dec} into the LTS $\K$. We also establish
requirements~\ref{Completeness} and~\ref{Enc-dec} of
Thm.~\ref{thm:Main}, leaving \ref{Preserved under AFO}--\ref{dec} to the applications of Thm.~\ref{thm:Main} in Sect.~\ref{sec:Apply}
for specific instances of $\sim$ and $\approx$.
Since we are only interested in TSSs with $\tau$-transitions, we here take
the set of actions $Act$ used in Sect.~\ref{sec:Framework} to be $A_\tau$.

\subsubsection[The conversion AFO]{The conversion $\AFO$}

Let again $\Gamma$ denote a predicate that marks the ($\Gamma$-liquid) arguments of function symbols.
In \cite{FvGdW05} we called a standard TSS \emph{abstraction-free} w.r.t.\ $\Gamma$ if only its
$\Gamma$-patience rules carry the label $\tau$ in their conclusion.
Here we use a slightly more liberal definition of abstraction-freeness that also allows rules
that have no premises, and a conclusion of the form \plat{$c \trans\tau d$} for constants $c$ and $d$.

The following conversion turns any $\Gamma$-patient standard TSS $P=(\Sigma,A_\tau,R)$ into a
$\Gamma$-patient and abstraction-free TSS \plat{$\AFO^{O,\zeta}_\Gamma(P)$}.
It is parameterised by the choice of a fresh set of actions $O$, so
$O \cap A_\tau = \emptyset$, and
a partial function $\zeta:\mbox{T}(\Sigma)\rightharpoonup O$, which we call an \emph{oracle}.
The choice of $O$ and $\zeta$ varies for different applications of Thm.~\ref{thm:Main}.
This choice will be made in Sect.~\ref{sec:Apply} in such a way that
requirements~\ref{Enc} and~\ref{Coincide} of Thm.~\ref{thm:Main} are met, for specific instances of $\sim$ and $\approx$.

\begin{definition}
\label{def:AFO}
Given a $\Gamma$-patient standard TSS $P=(\Sigma,A_\tau,R)$.
Let $\hat\Sigma$ be the signature $\Sigma$, enriched with a fresh constant $\surd$ and a fresh constant $\hat p$
for each closed term $p \in\mbox{T}(\Sigma)$.
Pick a fresh action $\iota\notin A_\tau\cup O$.
We define the TSS \plat{$\AFO^{O,\zeta}_\Gamma(P)$} as $(\hat\Sigma,A_\tau\cup O \cup \{\iota\},R')$ where the rules
in $R'$ are obtained from the rules in $R$ as follows:
\begin{enumerate}
\item
$R_1$ is obtained from $R$ by adding for each rule $r$ and each non-empty subset $S$ of positive $\tau$-premises of $r$, a copy of $r$ with as only difference that the labels $\tau$ in the premises in $S$ are replaced by $\iota$;
\item
$R_2$ is obtained from $R_1$ by replacing, in every rule that has a
conclusion with the label $\tau$ and is not a $\Gamma$-patience rule,
the $\tau$-label in the conclusion by $\iota$;
\item
$R_3$ is obtained from $R_2$ by adding the premise $v\ntrans{\iota}$ to each rule
  with a negative $\tau$-premise $v\ntrans{\tau}$;
\item
  $R_4$ is obtained from $R_3$ by the addition of a rule $\hat p \trans\alpha \hat q$ for each transition $p \trans\alpha q$ {\it ws}-provable from $P$;
\item
  $R_5$ is obtained from $R_4$ by the addition of a rule \plat{$\hat p \trans{\zeta(p)} \surd$} for each
  $p\in\mbox{T}(\Sigma)$ with $\zeta(p)$ defined;\vspace{2pt}
\item
  $R'$ adds to $R_5$ a rule \plat{$\displaystyle\frac{x_k \trans{\omega} y}{f(x_1,\dots,x_n)\trans{\omega} y}$} for each
  $\omega\in O$, $f\in\Sigma$ and argument $k$ with $\Gamma(f,k)$.\vspace{1ex}
\end{enumerate}
\end{definition}
Step 2 above makes the resulting TSS abstraction-free by renaming $\tau$-labels in conclusions of
non-patience rules into $\iota$.
To ensure that still the same transitions are derived, modulo the conversion of some $\tau$ into $\iota$-labels,
step 1 above allows positive premises labelled $\iota$ to be used instead of $\tau$ in all rules, and
step 3 achieves the same purpose for negative premises.
These three steps result in a $\Gamma$-patient and abstraction-free
TSS that could be called $\AF_\Gamma(P)$.

Conceptually, steps 1, 2 and 3 on the one hand and steps 4 and 5 on
the other hand are independent. Listing steps 1, 2 and 3 before steps
4 and 5 makes it evident that the transformations 1--3 do not apply to
the rules added by 4 and 5.

For convenience in proofs, and to better explain steps 4 and 5 of Def.~\ref{def:AFO}, we consider an auxiliary LTS $\G=(\mathbb{P}_\G,A_\tau,\rightarrow_\G)$
with $\mathbb{P}_\G=\{\hat p \mid p \in \mbox{T}(\Sigma)\}$ and \plat{$\hat p \trans{\alpha}_\G \hat q$} iff
\plat{$P \vdash_{\it ws} p \trans\alpha q$},\vspace{-2pt} and an auxiliary LTS $\H=(\mathbb{P}_{\,\H},A_\tau \cup O,\rightarrow_{\,\H})$
with $\mathbb{P}_{\,\H}=\mathbb{P}_\G \cup \{\surd\}$ and
${\rightarrow_{\,\H}} := {\rightarrow_\G} \cup \{\hat p \trans{\zeta(p)} \surd \mid \zeta(p)$ defined$\}$.
The LTS $\G$ is simply a disjoint copy of the LTS generated by $P\!$.

\begin{lemma}\label{lem:GH}
If $p \sim_P q$ for some $p,q\in\mbox{T}(\Sigma)$, then $\hat p \sim_\G \hat q$.
\end{lemma}
\begin{proof}
  We have $p \bis{P\uplus\G} \hat p$ for each $p\in\mbox{T}(\Sigma)$, because the relation
  $\{(p,\hat p),(\hat p,p) \mid p\in\mbox{T}(\Sigma)\}$ is a strong bisimulation.
  So $\hat p \bis{P\uplus\G} p \sim_{P\uplus\G} q \bis{P\uplus\G} \hat q$, using the definition
  of a behavioural equivalence, and thus $\hat p \sim_{P\uplus\G} \hat q$, using that ${\bis{}} \subseteq {\sim}$.
  Hence $\hat p \sim_\G \hat q$, again by the definition of a behavioural equivalence.
\end{proof}
The LTS $\H$ adds \emph{oracle transitions} to $\G$.
The idea is that $\zeta(p)$ is particular for the $\sim$-equivalence class of $p\in\mbox{T}(\Sigma)$,
which on the one hand ensures that $\hat p \sim_{\,\H} \hat q$ iff $\hat p \sim_\G \hat q$, and on the
other hand  enforces that $\sim$ and $\approx$ coincide on $\H$. Namely, if $\hat p\approx_{\,\H} \hat q$
then the oracle action of $\hat p$ can be matched by $\hat q$ (and
vice versa), which implies $\hat p\sim_{\,\H} \hat q$.

\begin{example}\label{exa:oracle}
Take $\sim$ to be weakly divergence-preserving branching bisimilarity, and
$\approx$ to be stability-respecting branching bisimilarity.
Let us say that a process $p$ in an LTS is
\emph{divergent} if there exists an infinite sequence of processes
$(p_k)_{k\in\N}$ such that $p=p_0$ and $p_k\trans{\tau}p_{k+1}$ for all
$k\in\N$, i.e.\ if $p \mathbin{\models} \Delta\top\!$.
Take $O\mathbin=\{\Delta\top\}$ and let $\zeta(p)\mathbin={\Delta\top}$ iff $p$ is divergent.
Thus in $\H$ all divergent states of $\G$ have a fresh outgoing transition labelled $\Delta\top$.

With this definition of $\H$, we have $\hat p \sim_{\,\H} \hat q$ iff $\hat p \sim_\G \hat q$:
any weakly divergence-preserving branching bisimulation $\brel$ on $\G$ relates divergent states
with divergent states only, and thus is also a weakly divergence-preserving branching bisimulation
on $\H$ (when adding $\surd\brel\surd$). Furthermore, as we will show
in Prop.~\ref{prop:Coincide}, due to the construction of $\H$, $\hat p \sim_{\,\H} \hat q$ iff $\hat p \approx_\H \hat q$.
\end{example}

\noindent
Steps 4 and 5 of Def.~\ref{def:AFO} incorporate the entire LTS $\H$
into $\AF_\Gamma(P)$: each state appears as a constant and
each transition appears as a rule without premises. The operators from $\Sigma$ can now be
applied to arguments of the form $\hat p$.  Finally, step 6 lets any term $f(x_1,\dots,x_n)$ inherit the
oracle transitions from its $\Gamma$-liquid arguments.
Steps 4, 5 and 6 preserve abstraction-freeness; for step 4 this uses
the relaxed definition of abstraction-freeness that allows to
incorporate $\tau$-transitions between constants as rules without premises.

\begin{example}
  Let $P$ have the rules
  \[
  \frac{x_1\trans\tau y}{g(x_1,x_2,x_3)\trans\tau g(y,x_2,x_3)} \hspace*{1cm}
  \frac{x_1\trans a y_1\quad x_1\trans\tau y_2\quad x_3 \trans\tau y_3}
  {g(x_1,x_2,x_3)\trans\tau x_2} \hspace*{1cm}
  \frac{x_2 \trans\tau y \quad x_3\ntrans\tau}{g(x_1,x_2,x_3)\trans a y}
  \]
  where $\Gamma(g,1)$. Then $\AFO^{O,\zeta}_\Gamma(P)$ has the rules
  \[
  \frac{x_1\trans\tau y}{g(x_1,x_2,x_3)\trans\tau g(y,x_2,x_3)} \hspace*{1cm}
  \frac{x_1\trans a y_1\quad x_1\trans\tau y_2\quad x_3 \trans\tau y_3}
  {g(x_1,x_2,x_3)\trans \iota x_2} \hspace*{1cm}
  \frac{x_2 \trans\tau y \quad x_3\ntrans\tau \quad x_3\ntrans \iota}{g(x_1,x_2,x_3)\trans a y}
  \]
  \[
  \frac{x_1\trans \iota y}{g(x_1,x_2,x_3)\trans \iota g(y,x_2,x_3)} \hspace*{1cm}
  \frac{x_1\trans a y_1\quad x_1\trans \iota y_2\quad x_3 \trans \iota y_3}
  {g(x_1,x_2,x_3)\trans \iota x_2} \hspace*{1cm}
  \frac{x_2 \trans \iota y \quad x_3\ntrans\tau \quad x_3\ntrans \iota}{g(x_1,x_2,x_3)\trans a y}
  \]
  \[
  \frac{x_1\trans a y_1\quad x_1\trans \iota y_2\quad x_3 \trans\tau y_3}
  {g(x_1,x_2,x_3)\trans \iota x_2} \hspace*{.75cm}
  \frac{x_1\trans a y_1\quad x_1\trans\tau y_2\quad x_3 \trans \iota y_3}
  {g(x_1,x_2,x_3)\trans \iota x_2} \hspace*{.75cm}
  \frac{x_1 \trans{\omega} y}{g(x_1,x_2,x_3)\trans{\omega} y}~\mbox{($\omega{\in} O$)}
  \]
  By step 1, the first and third rule spawn one copy and the second rule spawns three copies
  in which part or all of the $\tau$-labels of premises are renamed into $\iota$.
  By step 2, in each rule that has a $\tau$-label in the conclusion, except the patience rule for
  the first argument of $g$, this label is renamed into $\iota$.
  By step 3, in the third rule and its copy, a negative $\iota$-premise is added.
  By step 6, an $\omega$-rule is added for the first argument of $g$.
  Since there are no closed terms in this example, steps 4 and 5 are void.
\end{example}
Clearly, for any $\Gamma$-patient standard TSS $P$, the standard TSS $\AFO^{O,\zeta}_\Gamma(P)$ is
$\Gamma$-patient and abstraction-free w.r.t.\ $\Gamma$.
Henceforth, we drop the superscripts $O$ and $\zeta$, and $\AFO(P)$ denotes $\AFO_\Gamma(P)$ for the
largest predicate $\Gamma$ for which $P$ is $\Gamma$-patient.
The signature of $\AFO(P)$ contains the signature of $P$ enriched by a fresh constant $\hat p$ for
each closed term $p$ in $P$, as required in Thm.~\ref{thm:Main}.

\begin{lemma}\label{lem:HAFO}
Let $P$ be a complete standard TSS in ntyft format.
If $\hat p \sim_{\,\H} \hat q$ for some $p,q\in\mbox{T}(\Sigma)$, then $\hat p \sim_{\AFO(P)} \hat q$.
\end{lemma}
\begin{proof}
Since $\AFO(P)$ has no rules whose source is a variable, the only derivable transitions of the form
$\hat p \trans\alpha q^*$ with $\alpha\in A_\tau\cup O \cup \{\iota\}$ are the ones from $\H$, with
$q^*$ of the form $\hat q$ or $\surd$. For this reason any process $\hat p$ in $\H$ is strongly
bisimilar to the process $\hat p$ in $\AFO(P)$. Using this, the lemma follows just as Lem.~\ref{lem:GH}.
\end{proof}

\noindent
As an immediate consequence of Lem.~\ref{lem:GH} and Lem.~\ref{lem:HAFO} we have the following corollary.
\begin{corollary}\label{cor:req3}
Requirement \ref{Enc} of Thm.~\ref{thm:Main} is met if
$\hat p \sim_\G \hat q$ implies $\hat p \sim_{\,\H} \hat q$.\qed
\end{corollary}
The inference $\hat p \sim_\G \hat q \Rightarrow \hat p \sim_{\,\H} \hat q$ depends on the
choice of $O$ and $\zeta$, and is deferred to Sect.~\ref{sec:Apply}.

The oracle inheritance rules in $\AFO_\Gamma(P)$---introduced in step 6 of Def.~\ref{def:AFO}---ensure that a
closed term $f(p_1,\dots,p_n)$ has an outgoing $\omega$-transition, for $\omega\in O$, iff one of its $\Gamma$-liquid
arguments $p_k$ has such a transition. Ultimately, all such oracle transitions stem from $\H$.
Using that in $\AFO_\Gamma(P)$ any term $p\in\mbox{T}(\hat\Sigma)$ can
uniquely (up to renaming of variables) be
written as $\rho(t)$ with $t\in\mathbb{T}(\Sigma)$ and $\rho:\var(t)\rightarrow\mathbb{P}_{\,\H}$,
this observation can be phrased as follows.

\begin{observation}\label{obs:Liquid oracle}
Let $t\in\mathbb{T}(\Sigma)$ and $\rho:\var(t)\rightarrow\mathbb{P}_{\,\H}$;
then $\AFO_\Gamma(P) \vdash_{\it ws} \rho(t) {\trans{\omega}}$ iff $t$ has a $\Gamma$-liquid
occurrence of a variable $x$ with $\rho(x){\trans{\omega}_{\,\H}}$.
\hfill$\Box$
\end{observation}

\newcommand{\p}{\textit{\v{p}}}
\newcommand{\q}{\textit{\v{q}}}
\newcommand{\vr}{\textit{\v{r}}}
\newcommand{\vN}{\textit{\v{N}}}

\noindent
We proceed to compare provability in a standard TSS $P$ with
provability in its abstraction-free counterpart
$\AFO(P)$, in order to verify requirements 1 and 6 of
Thm.~\ref{thm:Main}.
The most laborious part in this comparison lays in step 4 of
Def.~\ref{def:AFO}. Therefore we deal with this step separately from
the other five. Hence, we define for any standard TSS  $P=(\Sigma,A_\tau,R)$ a TSS $\mathcal{G}(P)=(\hat\Sigma,A_\tau,R'')$, which
constitutes an intermediate between $P$ and $\AFO(P)$. It is built by only applying
step 4 from the construction of $\AFO(P)$, with $R$ in the role of $R_3$.

\subsubsection[Comparison of P and G(P)]{Comparison of $P$ and $\mathcal{G}(P)$}

We restrict our analysis to standard TSSs $P=(\Sigma,A_\tau,R)$ in decent ntyft format.
For $p\in\mbox{T}(\hat\Sigma)$, let $\p\in\mbox{T}(\Sigma)$ be
  obtained from $p$ by replacing every subterm $\hat{q}$ in $p$ by $q$.
Likewise, for $N$ a set of closed negative literals, $\vN := \{\p\ntrans\alpha \mid (p \ntrans\alpha)\in N\}$.

  When we are merely interested in proving, from $P$ or $\mathcal{G}(P)$, literals of the form $p \ntrans\alpha$ or
  $p{\trans\alpha}$, possibly under some hypotheses, where the target of a positive literal is existentially quantified,
  we can equivalently use a version of $P$ resp.\ $\mathcal{G}(P)$ in which all right-hand sides of the premises and
  conclusions of rules have been dropped. Here we use that $P$ and
  $\mathcal{G}(P)$ are in decent ntyft format. This we do below.

\begin{lemma}\label{lem:O irredundant}
  Let $P=(\Sigma,A_\tau,R)$ be a standard TSS in decent ntyft format. Let $\alpha\in A_\tau$ and $p\in \mbox{T}(\hat\Sigma)$.
  \begin{itemize}
  \item If $\displaystyle\mathcal{G}(P) \vdash_{\it irr} \frac{N}{p {\trans \alpha}}$, with $N$
    a set of closed negative literals, then
    $\displaystyle P \vdash_{\it irr} \frac{\vN \cup N'}{\p {\trans \alpha}}$
    for a set $N'$ of closed negative literals $\lambda$ with $P\vdash_{\it ws}\lambda$.\vspace{-4pt}
  \item If $\displaystyle P \vdash_{\it irr} \frac{N'}{\p {\trans \alpha}}$, with $N'$
    a set of closed negative literals, then either $P \not\vdash_{\it ws} \lambda$ for a literal\vspace{-2pt}
    $\lambda\in N'$, or $\displaystyle\mathcal{G}(P) \vdash_{\it irr} \frac{N}{p{\trans \alpha}}$ for some $N$ with $\vN \subseteq N'$.
  \end{itemize}
\end{lemma}
\begin{proof}
For the first statement we apply induction on the irredundant proof $\pi$ of {$\frac{N}{p {\trans \alpha}}$}\vspace{-1pt}
from $\mathcal{G}(P)$.

First assume that $p$ has the form $\hat q$, so that $\p=q$.
Since $\mathcal{G}(P)$ has no rules whose source is a variable, the last step of $\pi$ must be an application of
a rule \plat{$\hat q \trans\alpha \hat r$}, introduced in step 4 of the construction of $\mathcal{G}(P)$, and $N=\emptyset$.
It follows that \plat{$P\vdash_{\it ws} q \trans\alpha r$}. Let $N'$ be the unique set of negative literals
occurring in the well-supported proof of $q \trans\alpha r$ from $P$ that have no negative literals
below them. Then \plat{$P\vdash_{\it irr} \frac{N'}{q \trans\alpha r}$} and each literal $\lambda\in N'$ is
{\it ws}-provable from $P$.

Alternatively, $p=f(p_1,\dots,p_n)$ with $f\in\Sigma$.
Then the last step of $\pi$ must be an application of a decent ntyft
rule $\frac{H}{f(x_1,\dots,x_n){\trans\alpha}}$ in $P$
and the substitution $\sigma:\{x_1,\dots,x_n\}\rightarrow\mbox{T}(\hat\Sigma)$ with $\sigma(x_i)=p_i$ for $i=1,\dots,n$.
Let $\sigma':\{x_1,\dots,x_n\}\rightarrow\mbox{T}(\Sigma)$ be the substitution with $\sigma'(x_i)=\p_i$ for $i=1,\dots,n$.
  With every positive premise $\mu=(t {\trans\gamma})$ in $H$ we may
  associate a set of closed negative literals $N_\mu \subseteq N$ such
  that $\mathcal{G}(P) \vdash_{\it irr} \frac{N_\mu}{\sigma(t) {\trans
      \gamma}}$ and such that $N$ is the union of all $N_\mu$.
  So, by induction, $P \vdash_{\it irr} \frac{\vN_\mu \cup N'_\mu}{\sigma'(t) {\trans \gamma}}$
for a set $N'_\mu$ of closed negative literals $\lambda$ with $P\vdash_{\it ws}\lambda$.
Let $N'$ be the union of all the $N'_\mu$. Note that $\vN$ is the union of all the $\vN_\mu$ and
the literals $\sigma'(u)\ntrans\gamma$ with $(u\ntrans\gamma)\in H$.
It follows that $P \vdash_{\it irr} \frac{\vN \cup N'}{\p {\trans \alpha}}$.

For the second statement assume $P \vdash_{\it irr} \frac{N'}{\p {\trans \alpha}}$ and
$P \vdash_{\it ws} \lambda$ for all $\lambda\in N'$.\vspace{-2pt}
Then  $P \vdash_{\it ws} \p {\trans \alpha}$.
We show that $\mathcal{G}(P) \vdash_{\it irr} \frac{N}{p{\trans \alpha}}$ for some $N$
with $\vN \subseteq N'$, by induction on the irredundant proof $\pi$ of {$\frac{N'}{\p {\trans \alpha}}$} from $P$.
First assume that $p$ has the form $\hat q$, so that $\p=q$.
Then $P \vdash_{\it ws} q {\trans \alpha}$, so
$\mathcal{G}(P) \vdash_{\it irr} \hat q {\trans \alpha}$, meeting the requirement of the lemma.

Alternatively, $p=f(p_1,\dots,p_n)$ with $f\in\Sigma$, so that $\p=f(\p_1,\dots,\p_n)$.
Then the last step of $\pi$ must be an application of a decent ntyft rule $\frac{H}{f(x_1,\dots,x_n){\trans\alpha}}$
and the substitution $\sigma':\{x_1,\dots,x_n\}\rightarrow\mbox{T}(\Sigma)$ with $\sigma'(x_i)=\p_i$ for $i=1,\dots,n$.
Let $\sigma:\{x_1,\dots,x_n\}\rightarrow\mbox{T}(\hat\Sigma)$ be the substitution with $\sigma(x_i)=p_i$ for $i=1,\dots,n$.
For each positive premise $\mu=(t {\trans\gamma})$ in $H$ we have $P \vdash_{\it irr} \frac{N'_\mu}{\sigma'(t) {\trans \gamma}}$
for some $N'_\mu \subseteq N'$ and $P \vdash_{\it ws} \lambda$ for all $\lambda\in N'_\mu$,
and thus, by induction, \plat{$\mathcal{G}(P) \vdash_{\it irr} \frac{N_\mu}{\sigma(t){\trans\gamma}}$}\vspace{2pt}
for some $N_\mu$ with $\vN_\mu \subseteq N'_\mu$. Let $N$ be the union of all the $N_\mu$, and
the literals $\sigma(u)\ntrans\gamma$ with $(u\ntrans\gamma)\in H$. Then $\vN \subseteq N'$.
It follows that $\mathcal{G}(P) \vdash_{\it irr} \frac{N}{p{\trans \alpha}}$.
\end{proof}

\begin{proposition}\label{prop:O complete}
  Let $P=(\Sigma,A_\tau,R)$ be a complete standard TSS in decent ntyft format, $\alpha\in A_\tau$ and $p\in \mbox{T}(\hat\Sigma)$.
  \begin{itemize}
  \item $P \vdash_{\it ws} \p \trans \alpha q^* ~~\Leftrightarrow~~ \exists q\in \mbox{T}(\hat\Sigma).\;
    q^*=\q \wedge \mathcal{G}(P) \vdash_{\it ws} p \trans \alpha q$.
  \item $P \vdash_{\it ws} \p \ntrans \alpha ~~\Leftrightarrow~~ \mathcal{G}(P) \vdash_{\it ws} p \ntrans \alpha$.
  \end{itemize}
\end{proposition}
\begin{proof}
  ``$\Rightarrow$'': We prove both statements by simultaneous induction on the well-founded proof $\pi$
  from $\p \trans \alpha q^*$ or $\p \ntrans \alpha$ from $P$. First consider the case $P \vdash_{\it ws} \p \trans \alpha q^*$.

  Assume that $p$ has the form $\hat r$, so that $\p=r$.
  It follows immediately from step 4 of the construction of $\mathcal{G}(P)$ that
  $\mathcal{G}(P) \vdash_{\it ws} \hat r \trans \alpha \hat q^*$. So take $q:=\hat q^*$.

  Alternatively, $p=f(p_1,\dots,p_n)$ with $f\in\Sigma$, so that $\p=f(\p_1,\ldots,\p_n)$.
  Then the last step of $\pi$ must be an application of a decent ntyft rule
  $\frac{H}{f(x_1,\dots,x_n)\trans\alpha t}$
  and a substitution $\sigma':V\rightharpoonup\mbox{T}(\Sigma)$ with $\sigma'(x_i)=\p_i$ for $i=1,\dots,n$
  and $\sigma'(t)=q^*$.
  Let $\sigma'':\{x_1,\dots,x_n\}\rightarrow\mbox{T}(\hat\Sigma)$ be the substitution with $\sigma''(x_i)=p_i$ for $i=1,\dots,n$.
  For each negative premise $u {\ntrans\gamma}$ in $H$ we have $P \vdash_{\it ws} \sigma'(u) \ntrans \gamma$,
  and thus, by induction, $\mathcal{G}(P) \vdash_{\it ws} \sigma''(u)\ntrans\gamma$. Likewise,
  for each positive premise $\mu=(u \trans\gamma y_\mu)$ in $H$, $P \vdash_{\it ws} \sigma'(u) \trans \gamma \sigma'(y_\mu)$,
  and thus, by induction, $\mathcal{G}(P) \vdash_{\it ws} \sigma''(u)\trans\gamma r_\mu$ for some
  $r_\mu$ with $\vr_\mu = \sigma'(y_\mu)$. Let $\sigma:V\rightharpoonup\mbox{T}(\hat\Sigma)$ be a
  substitution with $\sigma(x_i)=p_i$ for $i=1,\dots,n$ and $\sigma(y_{\mu})=r_\mu$ for each positive premise $\mu$ in $H$.
  Then $\mathcal{G}(P) \vdash_{\it ws} \sigma(\mu)$ for each premise $\mu$ in $H$, and thus
  $\mathcal{G}(P) \vdash_{\it ws} p \trans\alpha \sigma(t)$. Take $q:=\sigma(t)$. Then \mbox{$\q=\sigma'(t)=q^*$}.

  Next consider the case $P \vdash_{\it ws} \p \ntrans \alpha$.
  Suppose that {$\frac{N}{p {\trans\alpha}}$}\vspace{-2pt} is irredundantly provable from $\mathcal{G}(P)$.
  Then, by Lem.~\ref{lem:O irredundant}, $P \vdash_{\it irr} \frac{\vN \cup N'}{\p {\trans \alpha}}$
  for a set $N'$ of closed negative literals $\lambda$ with $P\vdash_{\it ws}\lambda$.\vspace{-4pt}
  By Def.~\ref{def:wsp} $\vN \cup N'$ contains a literal $r' {\ntrans\gamma}$ such that
  $r'{\trans\gamma}$ is {\it ws}-provable from $P$ by means of a strict subproof of $\pi$. 
  By the consistency of $\vdash_{\it ws}$, $(r' {\ntrans\gamma})\notin N'$.
  Hence $N$ contains a literal $r {\ntrans\gamma}$ with $\vr = r'$.
  By induction $\mathcal{G}(P)\vdash_{\it ws} r {\trans \gamma}$,
  and this literal denies a literal in $N$. From this it follows that $\mathcal{G}(P)\vdash_{\it ws}p {\ntrans\alpha}$.

  ``$\Leftarrow$'': We prove both statements by simultaneous induction on the well-founded proof $\pi$
  of $p \trans \alpha q$ or $p \ntrans \alpha$ from $\mathcal{G}(P)$.
  First consider the case $\mathcal{G}(P) \vdash_{\it ws} p \trans \alpha q$.

  Assume that $p$ has the form $\hat r$, so that $\p=r$.
  Since $\mathcal{G}(P)$ has no rules whose source is a variable, the last step of $\pi$ must be an application of
  a rule \plat{$\hat r \trans\alpha \hat s$}, introduced in step 4 of the construction of $\mathcal{G}(P)$.
  It follows that $q=\hat s$, so $\q=s$, and \plat{$P\vdash_{\it ws} r \trans\alpha s$}.
  
  Alternatively, $p=f(p_1,\dots,p_n)$ with $f\in\Sigma$, so that $\p=f(\p_1,\ldots,\p_n)$.
Then the last step of $\pi$ must be an application of a decent ntyft rule $\frac{H}{f(x_1,\dots,x_n)\trans\alpha t}$
and a substitution $\sigma:V\rightharpoonup\mbox{T}(\hat\Sigma)$ with $\sigma(x_i)=p_i$ for
$i=1,\dots,n$ and $\sigma(t)=q$.
Let $\sigma':V\rightharpoonup\mbox{T}(\Sigma)$ be the substitution
with $\sigma'(x)=\makebox[0pt][l]{\,\v{}}\sigma(x)$ for each $x\in V$
such that $\sigma(x)$ is defined, and undefined otherwise.
For each premise $\mu$ in $H$ we have $\mathcal{G}(P) \vdash_{\it ws} \sigma(\mu)$, by means of a subproof of $\pi$,
and thus, by induction, $P \vdash_{\it ws} \sigma'(\mu)$.
It follows that $P \vdash_{\it ws} \p \trans\alpha \q$.

Finally, consider the case $\mathcal{G}(P) \vdash_{\it ws} p \ntrans \alpha$.
  Suppose that {$\frac{N'}{\p {\trans\alpha}}$} is irredundantly provable from $P$.
  Then, by Lem.~\ref{lem:O irredundant}, either $P \not\vdash_{\it ws} \lambda$ for a literal
  $\lambda\in N'$, or $\mathcal{G}(P) \vdash_{\it irr} \frac{N}{p{\trans \alpha}}$ for some $N$ with $\vN \subseteq N'$.\vspace{-2pt}
  In the first case, by the completeness of $P$, $P\vdash_{\it ws}\mu$ for a literal $\mu$
  denying $\lambda$, and we are done. So assume the second case.
  By Def.~\ref{def:wsp} $N$ contains a literal $r {\ntrans\gamma}$ such that
  $r{\trans\gamma}$ is {\it ws}-provable from $\mathcal{G}(P)$ by means of a strict subproof of $\pi$. 
  By induction $P\vdash_{\it ws} \vr {\trans \gamma}$,
  and this literal denies a literal in $N'$. From this it follows that $P \vdash_{\it ws} \p \ntrans \alpha$.
\end{proof}

\begin{corollary}\label{cor:G complete}
  Let $P=(\Sigma,A_\tau,R)$ be a complete standard TSS in decent ntyft format. Then also
  $\mathcal{G}(P)$ is complete.
\end{corollary}
\begin{proof}
  Let $p\in \mbox{T}(\hat\Sigma)$ and $\alpha\in A_\tau$. Since $P$ is complete, either
  $P \vdash_{\it ws} \p {\trans \alpha}$ or $P \vdash_{\it ws} \p {\ntrans \alpha}$.
  Consequently, by Prop.~\ref{prop:O complete}, either 
  $\mathcal{G}(P) \vdash_{\it ws} p {\trans \alpha}$ or $\mathcal{G}(P)\vdash_{\it ws} p {\ntrans \alpha}$.
\end{proof}

\subsubsection[Comparison of G(P) and AFO(P)]{Comparison of $\mathcal{G}(P)$ and $\AFO(P)$}

  $\mathcal{G}(P)$ and $\AFO(P)$ differ on the account of rules with
  premises and conclusions labelled $\iota$ and $\omega\in O$.
  We note that actions $\omega\in O$ play no role in the next lemma
  and proposition.

\begin{lemma}\label{lem:AFO irredundant}
  Let $P=(\Sigma,A_\tau,R)$ be a standard TSS in ntyft format.
  Let $a$ range over $A$ and $p,q$ over $\mbox{T}(\hat\Sigma)$.
  \begin{itemize}
  \item 
  $\AFO(P) \vdash_{\it irr} \frac{N}{p \trans a q}$\vspace{-2pt} for a set of closed negative
  literals $N$ iff there is a set of closed negative literals $N'$ such that
  {$\mathcal{G}(P) \vdash_{\it irr} \frac{N'}{p \trans a q}$}\vspace{-4pt} and
  $N=\{s {\ntrans{a}} \mid (s {\ntrans{a}}) \in N'\} \cup
   \{t {\ntrans{\tau}},~~{t {\ntrans{\iota}}} \mid (t  {\ntrans{\tau}}) \in N'\}$.
  \item 
  $\AFO(P) \vdash_{\it irr} \frac{N}{p \trans \alpha q}$, with $\alpha=\tau$
  or $\alpha=\iota$, for a set of closed negative literals $N$, iff there is an $N'$
  such that {$\mathcal{G}(P) \vdash_{\it irr} \frac{N'}{p \trans \tau q}$} and
  $N=\{s {\ntrans{a}} \mid (s {\ntrans{a}}) \in N'\} \cup
   \{t {\ntrans{\tau}},~~{t {\ntrans{\iota}}} \mid (t  {\ntrans{\tau}}) \in N'\}$.
  \end{itemize}
\end{lemma}
\begin{proof}
  ``If'': Both statements follow by a straightforward simultaneous induction on irredundant
  provability from $\mathcal{G}(P)$.

  ``Only if'': Both statements follow by a straightforward simultaneous induction on irredundant
  provability from $\AFO(P)$.
\end{proof}

\begin{proposition}\label{prop:AFO complete}
  Let $P=(\Sigma,A_\tau,R)$ be a standard TSS in ntyft format.
  Let $a$ range over $A$ and $p,q$ over $\mbox{T}(\hat\Sigma)$.
  \begin{itemize}
  \item $\mathcal{G}(P) \vdash_{\it ws} p \trans a q ~~\Leftrightarrow~~ \AFO(P) \vdash_{\it ws} p \trans a q$.

  \item \plat{$\mathcal{G}(P)\vdash_{\it ws} p \trans \tau q ~~\Leftrightarrow~~ (\AFO(P) \vdash_{\it ws} p \trans \tau q ~\vee~
                                                        \AFO(P) \vdash_{\it ws} p \trans \iota q)$}.

  \item $\mathcal{G}(P) \vdash_{\it ws} p \ntrans a ~~\Leftrightarrow~~ \AFO(P) \vdash_{\it ws} p \ntrans a$.

  \item \plat{$\mathcal{G}(P) \vdash_{\it ws} p \ntrans \tau ~~\Leftrightarrow~~ (\AFO(P) \vdash_{\it ws} p \ntrans \tau ~\wedge~
                                                        \AFO(P) \vdash_{\it ws} p \ntrans \iota)$}.
  \end{itemize}
\end{proposition}
\begin{proof}
Using Lem.~\ref{lem:AFO irredundant}, the lemma follows by two straightforward inductions on well-supported provability,
one for $\Rightarrow$, covering all four claims, and one for $\Leftarrow$.
\end{proof}

\begin{corollary}\label{cor:AFO complete}
  Let $P$ be a standard TSS in decent ntyft format.
  If $P$ is complete, then so is $\AFO(P)$.
\end{corollary}
\begin{proof}
  By Cor.~\ref{cor:G complete}, $\mathcal{G}(P)$ is complete.
  Hence, for each $p \in \mbox{T}(\hat\Sigma)$ and $a\in A$, either $\mathcal{G}(P) \vdash_{\it ws} p \ntrans a$ or
  $\mathcal{G}(P) \vdash_{\it ws} p {\trans a}$.
  So by Prop.~\ref{prop:AFO complete}, either $\AFO(P) \vdash_{\it ws} p \ntrans a$ or
  $\AFO(P) \vdash_{\it ws} p{\trans a}$.

  Each $p\in\mbox{T}(\hat\Sigma)$ can be written as $\rho(t)$ with $t\in\mathbb{T}(\Sigma)$ and
  $\rho:\var(t)\rightarrow\mathbb{P}_{\,\H}$.  Let $\Gamma$ be the largest predicate for which $P$ is $\Gamma$-patient.
  Now $\AFO_\Gamma(P) \vdash_{\it ws} \rho(t){\trans \tau}$ if $t$ has a
  $\Gamma$-liquid argument $x$ for which $\rho(x){\trans{\tau}_{\,\H}}$.
  Otherwise, $\AFO_\Gamma(P) \vdash_{\it ws} \rho(t) \ntrans \tau$.

  Likewise, for any $\omega\mathbin\in O$, $\AFO_\Gamma(P) \vdash_{\it ws}\rho(t) {\trans\omega}$ if $t$ has a
  $\Gamma$-liquid argument $x$ for which $\rho(x) {\trans{\omega}_{\,\H}}$.
  Otherwise, $\AFO_\Gamma(P) \vdash_{\it ws} \rho(t) \ntrans\omega$.

  A closed negative literal $p{\ntrans{\alpha}}$ is called \emph{true} if $\AFO(P) \vdash_{\it ws} p \ntrans a$,
  \emph{false} if $\AFO(P) \vdash_{\it ws} {p \trans a}$, and \emph{ambiguous} if neither applies.
  Above we proved that there are no ambiguous literals labelled $a \in A$ or $\tau$ or $\omega\in O$.
  It remains to show that there are no ambiguous literals labelled $\iota$.
  Towards a contradiction, assume that \plat{$p{\ntrans \iota}$} is ambiguous.

  There must exist a closed term $q$ and set of closed negative literals $N$ such that
  \plat{$\AFO(P)\vdash_{\it irr}\frac{N}{p\trans\iota q}$}\vspace{2pt} and no literal in $N$ is false.
  For if there were no such $q$ and $N$, the literal $p{\ntrans\iota}$ would be true
  by Def.~\ref{def:wsp}. Moreover, $N$ must contain an ambiguous literal \plat{$r{\ntrans\iota}$}, for if all
  literals in $N$ would be true, then the literal \plat{$p{\ntrans\iota}$} would be false by Def.~\ref{def:wsp}.
  By Lem.~\ref{lem:AFO irredundant} there is a set of closed negative literals $N'$ such that
  \plat{$\mathcal{G}(P) \vdash_{\it irr} \frac{N'}{p \trans \tau q}$} and
  $N=\{s {\ntrans{a}} \mid (s {\ntrans{a}}) \mathbin\in N'\} \cup
  \{t {\ntrans{\tau}},~~{t {\ntrans{\iota}}} \mid (t {\ntrans{\tau}}) \mathbin\in N'\}$.\vspace{1pt}
  So $N$ contains both \plat{$r{\ntrans\iota}$} and \plat{$r{\ntrans \tau}$}.
  If $\mathcal{G}(P)\vdash_{\it ws} r{\ntrans\tau}$ then $r{\ntrans{\iota}}$ would be true by Prop.~\ref{prop:AFO complete}.
  Hence $\mathcal{G}(P) \vdash_{\it ws} r{\trans\tau}$, by the completeness of $\mathcal{G}(P)$.
  So by Prop.~\ref{prop:AFO complete}, $\AFO(P) \vdash_{\it ws} r {\trans \tau}$ or
  \plat{$\AFO(P) \vdash_{\it ws} r {\trans\iota}$}. It follows that one of the literals
  \plat{$r{\ntrans\iota}$} or \plat{$r{\ntrans \tau}$} must be false, contradicting the absence of
  false literals in $N$.
\end{proof}

\subsubsection[Comparison of AFO(P) and K]{Comparison of $\AFO(P)$ and $\K$}

The LTS $\K=(\mathbb{P}_{\K},A_\tau\cup O \cup \{\iota\},\rightarrow_{\K})$ has as states
$\mathbb{P}_{\K}=\{{\it dec}(p) \mid p\in \mbox{T}(\hat\Sigma)\}$, and the transitions are the ones
generated by the following two rules:
\[
\frac{x\trans{\alpha}y}{{\it dec}(x)\trans{\alpha}{\it dec}(y)}
\qquad\qquad
\frac{x\trans{\iota}y}{{\it dec}(x)\trans{\tau}{\it dec}(y)}
\]
where $\alpha$ ranges over $A_\tau$. 
The operator ${\it dec}:\mbox{T}(\hat \Sigma) \rightarrow \mathbb{P}_{\K}$ simply sends any process
$p\in\mbox{T}(\hat \Sigma)$ to the state ${\it dec}(p) \in \mathbb{P}_{\K}$.
Thus, {\it dec} erases all transitions with labels from $O$ and
renames labels $\iota$ into $\tau$. All other transitions are preserved.

\begin{proposition}\label{prop:D complete}
  Let $P=(\Sigma,A_\tau,R)$ be a standard TSS\@.
  Let $a$ range over $A$ and $p,q$ over $\mbox{T}(\hat\Sigma)$.
  \begin{itemize}
  \item ${\it dec}(p) \trans a_\K q^* ~~\Leftrightarrow~~ \exists q \in\mbox{T}(\hat\Sigma).\;
        q^* = {\it dec}(q) \wedge \AFO(P) \vdash_{\it ws} p \trans a q$.

  \item \plat{${\it dec}(p) \trans \tau_\K q^* ~~\Leftrightarrow~~  \exists q \in\mbox{T}(\hat\Sigma).\;
        q^* = {\it dec}(q) \wedge \mbox{}$}\raisebox{-6pt}[0pt][0pt]{$\left(\begin{array}{@{}l@{}}\phantom{\vee~} \AFO(P) \vdash_{\it ws} p \trans \tau q\\[-1pt]
                                                                           \vee~ \AFO(P) \vdash_{\it ws} p \trans \iota q)\end{array}\right)$.}
  \end{itemize}
\end{proposition}
\begin{proof}
Straightforward.
\end{proof}

\subsubsection{Verifying requirements \ref{Completeness} and \ref{Enc-dec} of Thm.~\ref{thm:Main}}

That requirement \ref{Completeness} of Thm.~\ref{thm:Main} is met is
immediate by Cor.~\ref{cor:AFO complete}.

We end this section by verifying requirement~\ref{Enc-dec} of Thm.~\ref{thm:Main}.  Intuitively, the
behaviour of a process $f(p_1,\dots,p_n)$ in $P$ is the same as that of the process $f(\hat p_1,\dots,\hat p_n)$
in $\AFO(P)$, except that some $\tau$-transitions of the former are turned into
$\iota$-transitions of the latter process, and that some oracle transitions may have been added in
the latter. Since any rule in $\AFO(P)$ with a conclusion labelled by \mbox{$A_\tau\cup\{\iota\}$} has
positive and negative premises with labels from $A_\tau\cup\{\iota\}$ only, these oracle
transitions have no influence on the derivation of any transitions from $\AFO(P)$ with
labels in $A_\tau\cup\{\iota\}$.
The operator {\it dec} removes all oracle transitions and renames $\iota$ into $\tau$,
thereby returning the behaviour of $f(\hat p_1,\dots,\hat p_n)$ to match that of $f(p_1,\dots,p_n)$ exactly.

\begin{proposition}
Let $P$ be a complete standard TSS in decent ntyft format.\\ Then
 $f(p_1,\dots,p_n) \bis{P\uplus\K} {\it dec}(f(\hat p_1,\dots,\hat p_n))$
    for any $n$-ary $f\in\Sigma$ and $p_1,\dots,p_n\in \mbox{T}(\Sigma)$.
\end{proposition}
\begin{proof}
It suffices to show that the relation
$\{(\p,{\it dec}(p) \mid p\in\mbox{T}(\hat\Sigma)\}$ is a strong bisimulation.

So let $p\in\mbox{T}(\hat\Sigma)$.
Suppose that  ${\it dec}(p) \trans a_\K q^*$, with $a\in A$.
Then, by Prop.~\ref{prop:D complete}, there is a $q \in\mbox{T}(\hat\Sigma)$ with
$q^* = {\it dec}(q)$ and $\AFO(P) \vdash_{\it ws} p \trans a q$.
Hence, by Prop.~\ref{prop:AFO complete}, $\mathcal{G}(P) \vdash_{\it ws} p \trans a q$.
So, by Prop.~\ref{prop:O complete}, $P \vdash_{\it ws} \p \trans a \q$, which had to be shown.

The case that ${\it dec}(p) \trans \tau_\K q^*$ proceeds in the same way.

Now suppose that $P \vdash_{\it ws} \p \trans a q^*$, with $a\in A$.
Then, by Prop.~\ref{prop:O complete}, there is a $q \in\mbox{T}(\hat\Sigma)$ with
$q^* = \q$ and $\mathcal{G}(P) \vdash_{\it ws} p \trans a q$.
Hence, by Prop.~\ref{prop:AFO complete}, $\AFO(P) \vdash_{\it ws} p \trans a q$.
So, by Prop.~\ref{prop:D complete}, ${\it dec}(p) \trans a_\K {\it dec}(q)$, which had to be shown.
Again the case $P \vdash_{\it ws} \p \trans \tau q^*$ proceeds in the same way.
\end{proof}

\subsection{Application of the general framework to divergence-preserving semantics}
\label{sec:Apply}

We now apply Thm.~\ref{thm:Main} to derive that the stability-respecting
branching bisimulation format and its rooted variant are congruence
formats for $\bis[\Delta\top]{b}$ and $\bis[\Delta]{b}$, and
$\bis[\Delta\top]{rb}$ and $\bis[\Delta]{rb}$, respectively.

As congruence format $\mathfrak{F}$ in Thm.~\ref{thm:Main} we
take the (rooted) stability-respecting branching bisimulation format intersected
with the decent ntyft format.
If a TSS $P$ is in this format, there exist predicates $\aleph$ and $\Lambda$ such that 
$P$ is $\aL$-patient and only contains rules that are rooted stability-respecting branching bisimulation safe
w.r.t.\ $\aleph$ and $\Lambda$. Although there may be some freedom in the choice of $\aleph$ and $\Lambda$,
the predicate $\aL$ is completely determined by $P$. Namely, each $\aL$-liquid argument of an
operator symbol $f$ must have a patience rule, and conversely, each argument of $f$ that has a
patience rule must be $\aL$-liquid, by conditions~\ref{rhs} and~\ref{aleph} of
Def.~\ref{def:rooted_bra_bisimulation_safe}.
Hence there can be no ambiguity in the choice of $\Gamma$ in $\AFO_\Gamma(P)$.

It is straightforward to check that the conversion $\AFO$ on standard TSSs defined in
Sect.~\ref{sec:Abstraction-free} preserves the (rooted) stability-respecting branching bisimulation format,
and thus satisfies requirement~\ref{Preserved under AFO} of Thm.~\ref{thm:Main}.
Here it is important that in step 6 of Def.~\ref{def:AFO} a term $f(x_1,\dots,x_n)$
inherits oracle transitions only from its $\Gamma$-liquid arguments---else
condition \ref{aleph} of Def.~\ref{def:rooted_bra_bisimulation_safe} would be violated.
Furthermore, condition \ref{main1} of Def.~\ref{def:rooted_bra_bisimulation_safe} is preserved
because premises \plat{$v\ntrans{\tau}$} are kept in place in step 3 of Def.~\ref{def:AFO}; and
condition \ref{main2} of Def.~\ref{def:rooted_bra_bisimulation_safe} is preserved
because the transformation in Def.~\ref{def:AFO} does not introduce new positive premises with the label $\tau$.

\subsubsection{A congruence format for weakly divergence-preserving branching bisimilarity}

We first apply Thm.~\ref{thm:Main} with $\bis[\Delta\top]{b}$ and $\bis[s]{b}$ in the roles of
$\sim$ and $\approx$.
In the construction of the LTS $\H$ out of the LTS $\G$ (cf.\
Sect.~\ref{sec:Abstraction-free}) we take $O\mathbin=\{\Delta\top\}$ and
let $\zeta(p)\mathbin={\Delta\top}$ iff $p$ is divergent.
As observed in Ex.~\ref{exa:oracle}, $\hat p \sim_{\,\H} \hat
q$ iff $\hat p \sim_\G \hat q$. Hence, with Cor.~\ref{cor:req3}, requirement~\ref{Enc} of Thm.~\ref{thm:Main} is satisfied.
We proceed to show that also requirement~\ref{Coincide} is satisfied.
\begin{lemma}\label{lem:Liquid divergence}
Let $t\in\mathbb{T}(\Sigma)$ and
$\rho:\var(t)\rightarrow\mathbb{P}_{\,\H}$.
In $\AFO_\Gamma(P)$, $\rho(t)$ is divergent
  iff
$t$ has a $\Gamma$-liquid occurrence of a variable $x$
such that $\rho(x)$ is divergent.
\end{lemma}
\begin{proof}
   For ``only if'', suppose $\rho(t)$ is divergent, i.e., there is
     an infinite sequence of $\tau$-transitions from $\rho(t)$. Since
   $\AFO_\Gamma(P)$ is abstraction-free, each of these transitions
   must originate from a $\tau$-transition from a process $\rho(x)$, where $x$ occurs $\Gamma$-liquid in $t$.
   One of the variables $x$ that occurs $\Gamma$-liquid in $t$ must contribute infinitely many of
   these transitions, so that $\rho(x)$ is divergent.

  ``If'' follows immediately from the fact that the TSS $\AFO_\Gamma(P)$ is $\Gamma$-patient.
\end{proof}

\begin{lemma}\label{lem:Oracle divergence}
  Let $p\in\mbox{T}(\hat\Sigma)$.
  Then $\AFO_\Gamma(P) \vdash_{\it ws} p {\trans{\Delta\top}}$ iff
  $p$ is divergent.
\end{lemma}
\begin{proof}
In $\AFO_\Gamma(P)$ any term $p\in\mbox{T}(\hat\Sigma)$ can be
written as $\rho(t)$ with $t\in\mathbb{T}(\Sigma)$ and $\rho:\var(t)\rightarrow\mathbb{P}_{\,\H}$.

Suppose $\rho(t)$ is divergent.
Then $t$ has a $\Gamma$-liquid occurrence of a variable $x$ such
that $\rho(x)$ is divergent, by Lem~\ref{lem:Liquid divergence}. Hence $\rho(x)\trans{\Delta\top}_{\,\H}$
by the construction of $\H$. Thus
$\AFO_\Gamma(P) \vdash_{\it ws} \rho(t) {\trans{\Delta\top}}$,
by Obs.~\ref{obs:Liquid oracle}. The other direction proceeds likewise.
\end{proof}
The following proposition states that requirement~\ref{Coincide} of Thm.~\ref{thm:Main} is satisfied indeed.
\begin{proposition}\label{prop:Coincide}
On $\AFO(P)$ the equivalences $\bis[\Delta\top]{b}$ and $\bis[s]{b}$ coincide.
\end{proposition}
\begin{proof}
It suffices to show that the relation $\bis[s]{b}$ on $\AFO(P)$ is a weakly
divergence-preserving branching bisimulation. By definition it is a (stability-respecting) branching
bisimulation. Hence it remains to show that it is weakly divergence-preserving.
To this end, it suffices to show that if $p$ is divergent and $p
\bis[s]{b} q$ then also $q$ is divergent.

Suppose $p$ is divergent. Then by Lem.~\ref{lem:Oracle divergence} $\AFO(P) \vdash_{\it ws} p {\trans{\Delta\top}}$.
So $\AFO(P) \vdash_{\it ws} q \epsarrow q' {\trans{\Delta\top}}$ for some $q'$, by the
definition of a branching bisimulation. So $q'$ is divergent, by Lem.~\ref{lem:Oracle divergence},
and thus $q$ is divergent.
\end{proof}
The following example shows that Prop.~\ref{prop:Coincide} would not hold if we had
skipped step 6 of Def.~\ref{def:AFO}, inheriting oracle transitions for $\Gamma$-liquid arguments,
or if we had not used the oracle transitions at all.
\begin{example}
Let $p\in \mbox{T}(\Sigma)$ be deadlock (a process without outgoing transitions) and $q\in \mbox{T}(\Sigma)$ a process having only
a $\tau$-transition to itself, and a $\tau$-transition to $p$. Then in the LTS $\G$ we have
$\hat p \bis[s]{b}_{\,\G} \hat q$ but $\hat p \notbis[\Delta\top\!\!]{b}_{\!\!\!\G} \hat q$.\vspace{-3pt}
After translation to $\H$, the processes $\hat p$ and
$\hat q$ are distinguished by means of oracle transitions, so that we have
$\hat p \notbis[s]{b}_{\,\H} \hat q$ and $\hat p\notbis[\Delta\top\!\!]{b}_{\!\!\!\H} \hat q$.
Now let $\Sigma$ feature a unary operator $f$ with
as only rule $\frac{x \trans\tau y}{f(x) \trans\tau f(y)}$. If oracle transitions would not be
inherited in $\AFO(P)$, then we would have $f(\hat p) \bis[s]{b}_{\AFO(P)} f(\hat q)$ but
\plat{$f(\hat p) \notbis[\Delta\top\!]{b}_{\!\!\!\!\AFO(P)} f(\hat q)$}.
\end{example}

\noindent
We now verify requirement~\ref{dec} of Thm.~\ref{thm:Main}.

\begin{proposition}\label{prop:Abstraction congruence}
 $p \mathrel{\bis[\Delta\top\!\!\!]{b}}_{\!\!\!\!\AFO(P)} q  ~~\Rightarrow~~
 {\it dec}(p) \mathrel{\bis[\Delta\top\!\!\!]{b}}_{\!\!\!\!\K} {\it dec}(q)$.
\end{proposition}

\begin{trivlist} \item[\hspace{\labelsep}{\bf Proof:}]
Define the relation $\brel$ on the states of $\K$ by
\[ {\it dec}(p) \mathrel{\brel} {\it dec}(q) ~~\Leftrightarrow~~ p \bis[\Delta\top]{b} q\;. \]
It suffices to show that $\brel$ is a weakly divergence-preserving branching bisimulation.
\begin{itemize}
\item Suppose $p \bis[\Delta\top\!]{b} q$ and ${\it dec}(p)\trans\alpha p^\dagger$.
  By the semantics of ${\it dec}$ there are two possibilities.
  \begin{description}
    \item[{\sc Case 1:}] $p \trans\alpha p'$ for some $p'$ with $p^\dagger={\it dec}(p')$.
    Since $p \bis[\Delta\top\!]{b} q$, either $\alpha = \tau$ and $p'\bis[\Delta\top\!]{b}q$, or
    \plat{$q\epsarrow q' \trans{\alpha} q''$} for some $q'$ and $q''$ with $p\bis[\Delta\top\!]{b} q'$ and $p' \bis[\Delta\top\!]{b} q''$.
    So either $\alpha = \tau$ and ${\it dec}(\hspace{-0.45587pt}p')\mathbin{\brel}{\it dec}(q)$, or
    \plat{${\it dec}(q)\mathbin{\epsarrow} {\it dec}(q') \mathbin{\trans{\alpha}} {\it dec}(q'')$}
    with ${\it dec}(p)\mathbin{\brel}{\it dec}(q')$ and ${\it dec}(p') \mathrel{\brel} {\it dec}(q'')$.
    \item[{\sc Case 2:}] $\alpha=\tau$, and $p \trans\iota p'$ for some $p'$ with $p^\dagger={\it dec}(p')$.
    Since $p \bis[\Delta\top\!]{b} q$,
    \plat{$q\epsarrow q' \trans{\iota} q''$} for some $q'$ and $q''$ with $p\bis[\Delta\top\!]{b} q'$ and $p' \bis[\Delta\top\!]{b} q''$.
    So \plat{${\it dec}(q)\epsarrow {\it dec}(q') \trans{\tau} {\it dec}(q'')$}
    with ${\it dec}(p)\mathrel{\brel} {\it dec}(q')$ and ${\it dec}(p') \mathrel{\brel} {\it dec}(q'')$.
  \end{description}
\item Suppose $p\bis[\Delta\top\!]{b} q$ and there is an infinite sequence $(p_k^\dagger)_{k\in\N}^{\mbox{}}$
    such that ${\it dec}(p)=p_0^\dagger$, $p_k^\dagger\trans{\tau}p_{k+1}^\dagger$ and
    $p_k^\dagger\mathrel{\brel} {\it dec}(q)$ for all $k\in\N$.
    Then there is an infinite sequence $(p_k)_{k\in\N}$
    such that $p_0=p$ and, for all $k\in\N$, $p_k^\dagger={\it dec}(p_k)$, $p_k \bis[\Delta\top\!]{b} q$
    and either \plat{$p_k\trans{\tau}p_{k+1}$} or \plat{$p_k\trans{\iota}p_{k+1}$}.
    We distinguish two cases.
  \begin{description}
    \item[{\sc Case 1:}] 
    For infinitely many of the $k$ we have \plat{$p_k\trans{\iota}p_{k+1}$}.
    Then there is an infinite sequence $(p'_j)_{j\in\N}^{\mbox{}}$
    such that $p'_0=p_0$ and \plat{$p'_j\epsarrow\trans{\iota}p'_{j+1}$} for all $j\in\N$.
    Since $p \bis{b} q$, there must be an infinite sequence $(q'_j)_{j\in\N}^{\mbox{}}$
    such that $q=q'_0$, \plat{$q'_j\epsarrow\trans{\iota}q'_{j+1}$} and $p'_j\bis{b}q'_j$ for all $j\in\N$.
    It follows that ${\it dec}(q)={\it dec}(q'_0)$ and \plat{${\it dec}(p'_j)\epsarrow\trans{\tau}{\it dec}(p'_{j+1})$} for all $j\in\N$.
    In other words, there exists an infinite sequence $(q_\ell^\dagger)_{\ell\in\N}^{\mbox{}}$
    such that ${\it dec}(q)=q_0^\dagger$ and \plat{$q_\ell^\dagger\trans{\tau}q_{\ell+1}^\dagger$} for all $\ell\in\N$.\vspace{2pt}
    \item[{\sc Case 2:}] 
    There is an $n\in \N$ such that \plat{$p_k\trans{\tau}p_{k+1}$} for all $k\geq n$.
    Since $p \bis[\Delta\top\!]{b} q$, there must be a finite sequence $(q_k)_{k=0}^{n}$
    such that $q=q_0$ and, for all $0\leq k<n$, either
    \plat{$q_k\epsarrow q_{k+1}$} or \plat{$q_k\epsarrow\trans{\iota} q_{k+1}$}, and $p_{k+1}\mathbin{\bis[\Delta\top\!\!\!]{b}}q_{k+1}$.
    Moreover, since $p_n \bis[\Delta\top\!\!\!]{b} q_n$, and
    $p_k \mathbin{\bis[\Delta\top\!\!\!]{b}} q \mathbin{\bis[\Delta\top\!\!\!]{b}} p_n \mathbin{\bis[\Delta\top\!\!\!]{b}} q_n$ for each $k\geq n$,
    there must be an infinite sequence $(q_\ell)_{\ell> n}$
    such that \plat{$q_\ell\trans{\tau}q_{\ell+1}$} for all $\ell\geq n$.
    It follows that \plat{${\it dec}(q)={\it dec}(q_0)\epsarrow {\it dec}(q_n)$} and
    \plat{${\it dec}(q_\ell)\trans{\tau}{\it dec}(q_{\ell+1})$} for all $\ell\geq n$.
  \hfill$\Box$
  \end{description}
\end{itemize}
\end{trivlist}

\begin{corollary}
$\mathfrak{F}$ is a congruence format for $\bis[\Delta\top]{b}$.
\end{corollary}
As indicated in Sect.~\ref{sec:ruloids}, each standard TSS $P$ in ready simulation format can be
converted to a TSS $P'$ in decent ntyft format. In \cite{BFvG04} it is
shown that this transformation preserves the set of $ws$-provable
closed literals, hence $\sim$ is a congruence for $P$ iff $P'$ is. It
is not hard to check that if $P$ is in (rooted) stability-respecting
branching bisimulation format then so is $P'$. Thus we obtain the
following congruence theorem.

\begin{theorem}\label{thm:Congruence3}
Let $P$ be a complete standard TSS in stability-respecting branching bisimulation
format. Then $\bis[\Delta\top]{b}$ is a congruence for $P$.
\hfill $\Box$
\end{theorem}

\subsubsection{A congruence format for rooted weakly divergence-preserving branching bisimilarity}

We now apply Thm.~\ref{thm:Main} with $\bis[\Delta\top]{rb}$ and $\bis[s]{rb}$ in the roles of
$\sim$ and $\approx$. The choice of $O$ and $\zeta$ in the construction of $\H$ is the same as above.

Again requirement~\ref{Enc} of Thm.~\ref{thm:Main} is satisfied: since any two processes $g,h\in\mathbb{P}_\G$ are
$\bis[\Delta\top]{b}$-equivalent in $\G$ iff they are $\bis[\Delta\top]{b}$-equivalent in $\H$,
it follows immediately from Def.~\ref{def:rbb} that any two processes $g,h\in\mathbb{P}_\G$ are
$\bis[\Delta\top]{rb}$-equivalent in $\G$ iff they are $\bis[\Delta\top]{rb}$-equivalent in $\H$.

Requirement~\ref{Coincide} of Thm.~\ref{thm:Main} is also satisfied:
since on $\AFO(P)$ the equivalences 
$\bis[\Delta\top]{b}$ and $\bis[s]{b}$ coincide, it follows immediately from Def.~\ref{def:rbb}
that also $\bis[\Delta\top]{rb}$ and $\bis[s]{rb}$ coincide.

To check requirement~\ref{dec}  of Thm.~\ref{thm:Main} define the relation $\rrel$ on the states of $\rm K$ by
\[ {\it dec}(p) \mathrel{\rrel} {\it dec}(q) ~~\Leftrightarrow~~ p \bis[\Delta\top]{rb} q\;. \]
With Def.~\ref{def:rbb} it is straightforward to check that $\rrel$ is a rooted weakly
divergence-preserving branching bisimulation. Hence requirement~\ref{dec} is satisfied.

\begin{corollary}
Rooted $\mathfrak{F}$ is a congruence format for $\bis[\Delta\top]{rb}$.
\end{corollary}
Exactly as above, this yields the following congruence theorem.

\begin{theorem}\label{thm:Congruence4}
Let $P$ be a complete standard TSS in rooted stability-respecting branching bisimulation
format. Then $\bis[\Delta\top]{rb}$ is a congruence for $P$.
\hfill $\Box$
\end{theorem}

\subsubsection{A congruence format for divergence-preserving branching bisimilarity}
\label{sec:Divergence-preserving bb}

Next, we apply Thm.~\ref{thm:Main} with $\bis[\Delta]{b}$ and $\bis[s]{b}$ in the roles of
$\sim$ and $\approx$. In the construction of the LTS $\H$, we let $O$ contain a unique name for
each $\bis[\Delta]{b}$-equivalence class of processes in $\G$, and let $\zeta(p)$ be the name of the 
$\bis[\Delta]{b}$-equivalence class of $\hat p\in\mathbb{P}_\G$, for any $p\in T(\Sigma)$.
Thus in $\H$ all states of $\G$ have a fresh outgoing transition,
labelled with the name of its $\sim$-equivalence class in $\G$.

With this definition of $\H$, requirement~\ref{Enc} of
Thm.~\ref{thm:Main} is satisfied: any divergence-preserving branching bisimulation $\brel$ on $\G$ relates
states in the same $\bis[\Delta]{b}$-equivalence class only, and thus is also a
divergence-preserving branching bisimulation on $\H$ (when adding $\surd\brel\surd$).

We proceed to show that also requirement~\ref{Coincide} is satisfied. A few lemmas are needed.

\begin{lemma}
\label{lem:vGLT09a}
\cite{GLT09a}
Condition (D) in Def.~\ref{def:bb} can be replaced by the following equivalent condition:
\begin{itemize}
\item[(D$'$)]
if $p\brel q$ and there is an infinite sequence of processes $(p_k)_{k\in\N}$ such that
$p=p_0$, $p_k\trans{\tau}p_{k+1}$ and $p_k \brel q$ for all $k\in\N$,
then there is a process $q'$ such that $q\epsarrow\trans{\tau}q'$ and $p_k\brel q'$ for some $k\in\N$.
\end{itemize}
That is, the resulting definition also yields the relation $\bis[\Delta]{b}$.
\end{lemma}
Furthermore, we will employ the following property of $\bis[\Delta]{b}$.

\begin{lemma}
\label{lem:Stuttering}
\cite{GLT09a}
If $p_0\trans{\tau}p_1\trans{\tau}\cdots\trans{\tau}p_n$ and $p_0\bis[\Delta]{b}p_n$, then $p_0\bis[\Delta]{b}p_i$ for all $i=0,\ldots,n$.
\end{lemma}
The following proposition tells that requirement~\ref{Coincide} of Thm.~\ref{thm:Main} is satisfied indeed.

\begin{proposition}
On $\AFO(P)$ the equivalences $\bis[\Delta]{b}$ and $\bis[s]{b}$ coincide.
\end{proposition}
\begin{proof}
It suffices to show that the relation $\bis[s]{b}$ on $\AFO(P)$ is a
divergence-preserving branching bisimulation. By definition it is a (stability-respecting) branching
bisimulation. Hence it remains to show that it is divergence-preserving.
By Lem.~\ref{lem:vGLT09a} it suffices to show that it satisfies condition (D$'$).
So assume that $p\bis[s]{b} q$ and there is an infinite sequence of processes $(p_k)_{k\in\N}$ such that
$p=p_0$, \plat{$p_k\trans{\tau}p_{k+1}$} and $p_k\bis[s]{b} q$ for all
$k\in\N$. We need to find a process
$q'$ such that \plat{$q\epsarrow\trans{\tau}q'$} and $p_k\bis[s]{b} q'$ for some $k\in\N$.
Towards a contradiction, assume that there is no $q'$ with \plat{$q\epsarrow\trans{\tau}q'$} and $p_k\bis[s]{b} q'$ for some $k\in\N$.

Let $p\mathbin=\rho(t)$ and $q\mathbin=\rho(u)$ for univariate $t,u\mathbin\in\mathbb{T}(\Sigma)$ with
$\var(t)\cap\var(u)\mathbin=\emptyset$ and $\rho\mathord{:}\var(t){\cup}\var(u)\mathbin\rightarrow\mathbb{P}_{\,\H}$.
Since $\AFO(P)$ is abstraction-free, all $\tau$-transitions in the infinite sequence
\plat{$(p_k\trans{\tau}p_{k+1})_{k\in\N}$} originate, through patience rules, from $\tau$-transitions of the
processes $\rho(x)$ for variables $x$ occurring $\Gamma$-liquid in $t$.
Let $L$ be the set of variables $x$ that occur $\Gamma$-liquid in $t$.
Then $\rho(x)=h_0^x \trans\tau_{\,\H} h_1^x  \trans\tau_{\,\H} h_2^x \trans\tau_{\,\H} \cdots$ for each $x\in L$, 
and for at least one $x\in L$---call it $z$---this sequence is infinitely long.
For each $k \in\N$ and $x\in L$, let $k_x\leq k$ keep track of how many of the $\tau$-transitions
in the sequence $p_0\trans{\tau}\cdots\trans{\tau}p_k$ originate from $\rho(x)$. Hence
$p_k=\rho_k(t)$ where $\rho_k(x) = h_{k_x}^x$ for each $x\in L$ and $\rho_k(x)=\rho(x)$ for each $x\in\var(t)\setminus L$.
Since, by assumption, the sequence of $\tau$-transitions
originating from $\rho(z)$ is infinite, for each $\ell\in\N$ there exists a $k\in\N$ such that $k_z=\ell$.
\vspace{1ex}

\textbf{Claim:} There is a $j\in \N$ such that $h_n^z \bis[\Delta]{b} h_{j_z}^z$ for all $n\geq j_z$.
\vspace{1ex}

\textbf{Proof of claim:} Suppose there is no such $j$.
Then, using Lem.~\ref{lem:Stuttering}, the processes $h_\ell^z$ for $\ell\in \N$ belong to infinitely many
$\bis[\Delta]{b}$-equivalence classes. So there are infinitely many actions $\omega\in O$ such that
$\exists \ell.\; h^z_\ell \trans{\omega}_{\,\H}$. For each of those actions $\omega$,
\plat{$\AFO(P)\vdash_{\it ws} p_k {\trans{\omega}}$} for some $k\in\N$, by Obs.~\ref{obs:Liquid oracle},
and hence \plat{$\AFO(P)\vdash_{\it ws} q \epsarrow q'' {\trans{\omega}}$} for some $q''$ with $p_k\bis[s]{b} q''$,
since $p_k \bis[s]{b} q$. As we have assumed that there is no $q'$ with \plat{$q\epsarrow\trans{\tau}q'$}
and $p_k\bis[s]{b} q'$ for some $k\in\N$, we have $q''=q$, and thus \plat{$\AFO(P)\vdash_{\it ws} q {\trans{\omega}}$}.
However, according to Obs.~\ref{obs:Liquid oracle} there can only be finitely many actions $\omega\in O$ with
\plat{$\AFO(P)\vdash_{\it ws} q {\trans{\omega}}$}, namely one for each variable occurring $\Gamma$-liquid in $u$.
\hfill \rule{1ex}{7pt}
\vspace{1ex}

Let $\omega$ be the name of the $\bis[\Delta]{b}$-equivalence class of the process \plat{$\rho_j(z) = h_{j_z}^z$}.
Then \plat{$\rho_{j}(z)\trans{\omega}_{\,\H}$}, and consequently $\AFO(P)\vdash_{\it ws}\rho_{j}(t){\trans{\omega}}$
by Obs.~\ref{obs:Liquid oracle}.
Since $\rho_{j}(t)\bis[s]{b}q$, it follows that $\AFO(P)\vdash_{\it ws}q \mathop{\epsarrow} q''\!\trans{\omega}$
for some process $q''$ with $\rho_{j}(t)\bis[s]{b}q''$.
As we assumed that there is no $q'$ with \plat{$q\epsarrow\trans{\tau}q'$}
and $p_k\bis[s]{b} q'$ for some $k\in\N$, we have $q''=q$, and thus \plat{$\AFO(P)\vdash_{\it ws} q=\rho(u) {\trans{\omega}}$}.
So by Obs.~\ref{obs:Liquid oracle} $u$ has a $\Gamma$-liquid
occurrence of a variable $y$ with $\rho(y){\trans{\omega}_{\,\H}}$. It follows that in $\G$ we have
$h_{j_z}^z=\rho_{j}(z) \bis[\Delta]{b} \rho(y)$. Thus, by
Lem.~\ref{lem:vGLT09a}, there exists a $g \in\mathbb{P}_\G$
with \plat{$\rho(y) \epsarrow\trans{\tau} g$} and $h_\ell^z \bis[\Delta]{b} g$ for some $\ell \geq j_z$.
Since the TSS $\AFO(P)$ is $\Gamma$-patient, we have \plat{$q=\rho(u) \epsarrow\trans{\tau} \rho'(u)$}
for a substitution $\rho'$ with $\rho'(y)=g$ and $\rho'(x)=\rho(x)$ for all variables $x\neq y$.
As, using the claim, $\rho(y) \bis[s]{b} \rho_{j}(z) = h_{j_z}^z \bis[s]{b} h_\ell^z  \bis[s]{b} g =\rho'(y)$, and $\bis[s]{b}$ is a
congruence on $\AFO(P)$, it follows that $\rho(u) \bis[s]{b} \rho'(u)$.
Since $p \bis[s]{b} q = \rho(u) \bis[s]{b} \rho'(u)$, it suffices to take $q':=\rho'(u)$, so that \plat{$q \epsarrow\trans{\tau} q'$}
and $p \bis[s]{b} q'$. This contradicts our assumption that no such $q'$ exists.
\end{proof}

\noindent
We now verify requirement~\ref{dec} of Thm.~\ref{thm:Main}.

\begin{proposition}
 $p \mathrel{\bis[\Delta]{b}}_{\AFO(P)} q  ~~\Rightarrow~~ {\it dec}(p) \mathrel{\bis[\Delta]{b}}_{\rm K} {\it dec}(q)$.
\end{proposition}

\begin{trivlist} \item[\hspace{\labelsep}{\bf Proof:}]
Define the relation $\brel$ on the states of $\rm K$ by
\[ {\it dec}(p) \mathrel{\brel} {\it dec}(q) ~~\Leftrightarrow~~ p \bis[\Delta]{b} q\;. \]
By Lem.~\ref{lem:vGLT09a} it suffices to show that $\brel$ is a branching bisimulation satisfying condition (D$'$).
That it is a branching bisimulation follows exactly as in the proof of Prop.~\ref{prop:Abstraction congruence}.
To show that it satisfies (D$'$), suppose $p\bis[\Delta]{b} q$ and there is an infinite
sequence $(p_k^\dagger)_{k\in\N}^{\mbox{}}$ such that ${\it dec}(p)=p_0^\dagger$,
$p_k^\dagger\trans{\tau}p_{k+1}^\dagger$ and $p_k^\dagger\mathrel{\brel} {\it dec}(q)$ for all $k\in\N$.
    Then there is an infinite sequence $(p_k)_{k\in\N}$
    such that $p_0=p$ and, for all $k\in\N$, $p_k^\dagger={\it dec}(p_k)$, $p_k \bis[\Delta]{b} q$
    and either \plat{$p_k\trans{\tau}p_{k+1}$} or \plat{$p_k\trans{\iota}p_{k+1}$}.
    We distinguish two cases.
  \begin{description}
    \item[{\sc Case 1:}] There is a $k\in\N$ with \plat{$p_k\trans{\iota}p_{k+1}$}---consider the first such $k$.
    Then $p\epsarrow p_k \trans{\iota} p_{k+1}$.
    Since $p \bis[\Delta]{b} q$, there is a $q'$ with $q\epsarrow \trans{\iota} q'$ and $p_{k+1}\bis[\Delta]{b} q'$.
    Hence there is a ${\it dec}(q')$ with ${\it dec}(q)\epsarrow \trans{\tau} {\it dec}(q')$ and
    ${\it dec}(p_{k+1})\brel{\it dec}(q')$.
    \item[{\sc Case 2:}] There is no such $k$. Then, since $\bis[\Delta]{b}$ satisfies (D$'$),
    there is a $q'$ such that $q\epsarrow\trans{\tau}q'$ and $p_k \bis[\Delta]{b}q'$ for some $k\in\N$.
    It follows that there is a ${\it dec}(q')$ such that ${\it dec}(q)\epsarrow \trans{\tau} {\it dec}(q')$ and
    ${\it dec}(p_{k+1})\brel{\it dec}(q')$.
    \hfill$\Box$
  \end{description}
\end{trivlist}

\begin{corollary}
$\mathfrak{F}$ is a congruence format for $\bis[\Delta]{b}$.
\end{corollary}
Exactly as above, this yields the following congruence theorem.

\begin{theorem}\label{thm:Congruence5}
Let $P$ be a complete standard TSS in stability-respecting branching bisimulation
format. Then $\bis[\Delta]{b}$ is a congruence for $P$.
\hfill $\Box$
\end{theorem}

\subsubsection{A congruence format for rooted divergence-preserving branching bisimilarity}

We now apply Thm.~\ref{thm:Main} with $\bis[\Delta]{rb}$
and $\bis[s]{rb}$ in the roles of
$\sim$ and $\approx$. The choice of $O$ and $\zeta$ in the construction of $\H$ is the same as in Sect.~\ref{sec:Divergence-preserving bb}.

Again requirement~\ref{Enc} of Thm.~\ref{thm:Main} is satisfied: since any two processes $g,h\in\mathbb{P}_\G$ are
$\bis[\Delta]{b}$-equivalent in $\G$ iff they are $\bis[\Delta]{b}$-equivalent in $\H$,
it follows immediately from Def.~\ref{def:rbb} that any two processes $g,h\in\mathbb{P}_\G$ are
$\bis[\Delta]{rb}$-equivalent in $\G$ iff they are $\bis[\Delta]{rb}$-equivalent in $\H$.

Requirement~\ref{Coincide} of Thm.~\ref{thm:Main} is also satisfied:
since on $\AFO(P)$ the equivalences 
$\bis[\Delta]{b}$ and $\bis[s]{b}$ coincide, it follows immediately from Def.~\ref{def:rbb}
that also $\bis[\Delta]{rb}$ and $\bis[s]{rb}$ coincide.

To check requirement~\ref{dec} of Thm.~\ref{thm:Main} define the relation $\rrel$ on the states of $\rm K$ by
\[ {\it dec}(p) \mathrel{\rrel} {\it dec}(q) ~~\Leftrightarrow~~ p \bis[\Delta]{rb} q\;. \]
With Def.~\ref{def:rbb} it is straightforward to check that $\rrel$ is a rooted weakly
divergence-preserving branching bisimulation. Hence requirement~\ref{dec} is satisfied.

\begin{corollary}
Rooted $\mathfrak{F}$ is a congruence format for $\bis[\Delta]{rb}$.
\end{corollary}
Exactly as above, this yields the following congruence theorem.

\begin{theorem}\label{thm:Congruence6}
Let $P$ be a complete standard TSS in rooted stability-respecting branching bisimulation
format. Then $\bis[\Delta]{rb}$ is a congruence for $P$.
\hfill $\Box$
\end{theorem}

\section{Related work}
\label{sec:related}

Ulidowski \cite{Uli92,UP02,UY00} proposed congruence formats, inside GSOS \cite{BIM95}, for weak
semantics that take into account non-divergence, called \emph{convergence} in \cite{vGl93}.
In \cite{Uli92} he introduces the \emph{ISOS} format, and shows that the weak convergent
\emph{refusal simulation} preorder is a precongruence for all TSSs in the ISOS format.
The GSOS format---in our terminology the \emph{decent nxyft} format---allows only decent ntyft
rules with variables as the left-hand sides of premises. The ISOS format is contained in the
intersection of the GSOS format and our stability-preserving branching bisimulation format.
Its additional restriction is that no variable may occur multiple times as the left-hand side of a
positive premise, or both as the left-hand side of a positive premise and in the conclusion of a rule.

In \cite{UP02,UY00} he employs \emph{Ordered SOS} (OSOS) TSSs \cite{MPRU09}.
An OSOS TSS allows no negative premises, but includes
priorities between rules: $r<r'$ means that $r$ can only be applied if $r'$ cannot.
An OSOS specification can be seen as, or translated into, a GSOS specification with negative premises.
Each rule $r$ with exactly one higher-priority rule $r'>r$
is replaced by a number of rules, one for each (positive) premise of $r'$; in the copy of $r$, this premise is negated.
For a rule $r$ with multiple higher-priority rules $r'$, this replacement is carried out for each such $r'$.

The {\tt ebo} and {\tt bbo} formats from \cite{UP02} target unrooted convergent \emph{delay} and branching bisimulation equivalence, respectively.
The {\tt bbo} format is more liberal than the {\tt ebo} format, which in turn is more liberal than the ISOS format.
The {\tt rebo} and {\tt rbbo} formats from \cite{UY00} target rooted convergent delay and branching bisimulation equivalence, respectively.
These rooted formats are more liberal than their unrooted counterparts, and the {\tt rbbo} format is more liberal than the {\tt rebo} format.

If patience rules are not allowed to have a lower priority than other rules,
then the {\tt (r)bbo} format, upon translation from OSOS to GSOS, can be seen as
a subformat of our (rooted) stability-respecting branching bisimulation format.
The basic idea is that in the {\tt rbbo} format all arguments of so-called $\tau$-preserving function
symbols \cite{UY00}, which are the only ones allowed to occur in targets, are declared $\Lambda$-liquid;
in the {\tt bbo} format, all arguments of function symbols are declared $\Lambda$-liquid.
Moreover, all arguments of function symbols that occur as the left-hand side of a positive premise are declared $\aleph$-liquid.
Patience rules are in the {\tt (r)bbo} format however, under strict conditions, allowed to be dominated by other rules, which in our setting gives rise
to patience rules with negative premises. This is outside the realm of our rooted stability-respecting branching bisimulation format.
On the other hand, the TSSs of the process algebra BPA$_{\epsilon\delta\tau}$, the binary Kleene star and deadlock testing (see \cite{FvG16,FvG17}),
for which rooted convergent branching bisimulation equivalence is a congruence, are outside the {\tt rbbo} format but within
the rooted stability-respecting branching bisimulation format.

\section{Conclusions}

We showed how the method from \cite{FvGdW12} for deriving congruence formats through modal decomposition can be applied to weak semantics that are stability-respecting.
We used (rooted and unrooted) stability-respecting branching bisimulation equivalence as a notable example.
Moreover, we developed a general method for lifting congruence formats from a weak semantics to a finer semantics, and used it to show that congruence formats 
for stability-respecting branching bisimulation equivalence are also congruence formats for their divergence-preserving counterparts.
This research provides a deeper insight into the link between modal logic and congruence formats,
and strengthens the framework from \cite{FvGdW12} for the derivation
of congruence formats for weak semantics.

Almost every weak semantics has stability-respecting and
divergence-preserving variants. Such variants have been studied most
widely in the literature for branching, $\eta$-, delay and weak bisimulation equivalence, but
they have for instance also been considered for decorated trace
semantics, such as \emph{exhibited behaviour equivalence} 
\cite{Pm86}, the \emph{generalised failure preorder} \cite{Lk89},
\emph{refusal equivalence} \cite{Ph87}, and \emph{copy+refusal
  equivalence} \cite{Uli92}.
We expect that the methods developed in this paper, combined with the
methods from \cite{FvGdW12,FvG17},  can serve as a cornerstone in the
generation of congruence formats for such semantics. In particular, we conjecture
that this will straightforwardly yield congruence formats for
stability-respecting and divergence-preserving variants of $\eta$-,
delay and weak bisimulation equivalence.

Admittedly, the whole story is quite technical and intricate. Partly this is because we build on a rich body of earlier work
in the realm of structural operational semantics: the notions of well-supported proofs and complete TSSs from \cite{vGl04};
the ntyft/ntyxt format from \cite{Gro93,BolG96}; the transformation to ruloids, which for the main part goes
back to \cite{FvG96}; and the work on modal decomposition and congruence formats from \cite{BFvG04}.
In spite of these technicalities, we have arrived at a relatively simple framework for the derivation of congruence
formats for weak semantics. Namely, for this one only needs to: (1) provide a modal characterisation of the weak semantics under
consideration; (2) study the class of modal formulas that result from decomposing this modal characterisation, and formulate
syntactic restrictions on TSSs to bring this class of modal formulas within the original modal characterisation; and (3) check
that these syntactic restrictions are preserved under the transformation to ruloids.
Steps (2) and (3) are very similar in structure for different weak semantics, as exemplified by the way we obtained a congruence format
for stability-respecting branching bisimulation equivalence.
And the resulting congruence formats tend to be more liberal and elegant than existing congruence formats in the literature.

Our intention is to carve out congruence formats for all weak semantics in the spectrum from \cite{vGl93} that have reasonable congruence properties.
At first we expected that the current third instalment would allow us to do so. However, it turns out that convergent weak semantics as considered
in for instance \cite{UP02,UY00,Wal90} still need extra work. The modal characterisations of these
semantics are three-valued \cite{vGl93}, which requires an extension of the modal decomposition technique to a three-valued setting.

\bibliographystyle{plain}

\newpage

\appendix

\section{Modal characterisations}

We prove Thm.~\ref{thm:characterisation}, which states that $\IO{b}^s$ is a modal
characterisation of $p\bis[s]{b}q$, and $\IO{rb}^s$ of $p\bis[s]{rb}q$.
So we need to prove, given an LTS $(\mathbb{P},{\it Act},\rightarrow)$, that
$p\bis[s]{b}q\Leftrightarrow p\sim_{\IO{b}^s}q$ and
$p\bis[s]{rb}q\Leftrightarrow p\sim_{\IO{rb}^s}q$ for all $p,q\in\mathbb{P}$.

\begin{proof}
  ($\Rightarrow$) We prove by structural induction on $\phi$, resp.\ $\overline\phi$, that
  $p\bis[s]{b}q \wedge p\models\phi \Rightarrow q\models\phi$ for all $\phi\in\IO{b}^s$,
  and $p\bis[s]{rb}q \wedge p\models\overline\phi \Rightarrow q\models\overline\phi$ for all
  $\overline\phi\in\IO{rb}^s$,
  The converse implications ($q\models\phi \Rightarrow p\models\phi$ and
  $q\models\overline \phi \Rightarrow p\models\overline \phi$) follow by symmetry.
\begin{itemize}
\item
$\phi=\bigwedge_{i\in I}\phi_i$. Then $p\models\phi_i$ for all $i\in I$. By induction
  $q\models\phi_i$ for all $i\in I$, so $q\models\bigwedge_{i\in I}\phi_i$.
\item
$\phi=\neg\phi'$. Then $p\not\models\phi'$. By induction $q\not\models\phi'$, so $q\models\neg\phi'$.
\item
$\phi=\eps\phi_1\diam{\hat\tau}\phi_2$. Then for some $n$ there are $p_0,\ldots,p_n\in\mathbb{P}$ with $p_0=p$, \plat{$p_i\trans{\tau}p_{i+1}$} for $i\in\{0,\ldots,n-1\}$, and $p_n\models\phi_1\diam{\hat\tau}\phi_2$. We apply induction on $n$.
\begin{description}
\item
[$\boldsymbol{n=0}$] Then $p\models\phi_1$, so by induction on formula size, $q\models\phi_1$. Furthermore, either (1) $p\models\phi_2$ or (2) there is a $p'\in\mathbb{P}$ with \plat{$p\trans{\tau}p'$} and $p'\models\phi_2$. In case (1), by induction on formula size, $q\models\phi_2$, so $q\models\eps\phi_1\diam{\hat\tau}\phi_2$. In case (2), since $p\bis[s]{b}q$, by Def.~\ref{def:bb} either (2.1) $p'\bis[s]{b}q$ or (2.2) \plat{$q\epsarrow q'\trans{\tau}q''$} with $p\bis[s]{b}q'$ and $p'\bis[s]{b}q''$. In case (2.1), by induction on formula size, $q\models\phi_2$. In case (2.2), by induction on formula size, $q'\models\phi_1$ and $q''\models\phi_2$. In both cases, $q\models\eps\phi_1\diam{\hat\tau}\phi_2$.
\item
[$\boldsymbol{n>0}$] Since \plat{$p\trans{\tau}p_1$}, and $p\bis[s]{b}q$, according to Def.~\ref{def:bb} there are two possibilities.
\begin{enumerate}
\item
Either $p_1\bis[s]{b}q$. Since $p_1\models\eps\phi_1\diam{\hat\tau}\phi_2$, by induction on $n$, $q\models\eps\phi_1\diam{\hat\tau}\phi_2$.
\item
Or \plat{$q\epsarrow q'\trans{\tau}q''$} with $p_1\bis[s]{b}q''$. Since $p_1\models\eps\phi_1\diam{\hat\tau}\phi_2$, by induction on $n$, $q''\models\eps\phi_1\diam{\hat\tau}\phi_2$. Hence $q\models\eps\phi_1\diam{\hat\tau}\phi_2$.
\end{enumerate}
\end{description}
\item
$\phi=\eps\phi_1\diam{a}\phi_2$. Then for some $n$ there are $p_0,\ldots,p_n\in\mathbb{P}$ with $p_0=p$, \plat{$p_i\trans{\tau}p_{i+1}$} for $i\in\{0,\ldots,n-1\}$, and $p_n\models\phi_1\diam{a}\phi_2$. We apply induction on $n$.
\begin{description}
\item
[$\boldsymbol{n=0}$] Then $p\models\phi_1$, and there is a $p'\in\mathbb{P}$ with \plat{$p\trans{a}p'$} and $p'\models\phi_2$. Since $p\bis[s]{b}q$, by Def.~\ref{def:bb} \plat{$q\epsarrow q'\trans{a}q''$} with $p\bis[s]{b}q'$ and $p'\bis[s]{b}q''$. By induction on formula size, $q'\models\phi_1$ and $q''\models\phi_2$. Hence $q\models\eps\phi_1\diam{a}\phi_2$.
\item
[$\boldsymbol{n>0}$] This case goes exactly as the case $n>0$ above.
\end{description}
\item
$\phi=\eps(\neg\diam{\tau}\top\land\,\overline\phi)$ with $\overline\phi\in\IO{rb}^s$.
Then for some $n$ there are $p_0,\ldots,p_n\in\mathbb{P}$ with $p_0=p$, \plat{$p_i\trans{\tau}p_{i+1}$} for $i\in\{0,\ldots,n-1\}$, and $p_n\models\neg\diam{\tau}\top\land\,\overline\phi$. We apply induction on $n$.
\begin{description}
\item
[$\boldsymbol{n=0}$] Then $p\models\overline\phi$ and \plat{$p\ntrans{\tau}$}.
Since $p$ and $q$ are related by a \emph{stability-respecting} branching bisimulation,
there is a $q'$ with $q \epsarrow q'\ntrans\tau$ and $p\bis[s]{b}q'$.
Because $p$ and $q'$ are both stable, $p\bis[s]{b}q'$ implies $p\bis[s]{rb}q'$.
So by induction  on formula size, $q'\models\overline\phi$. Hence $q \models
\eps(\neg\diam{\tau}\top\land\,\overline\phi)$.
\item
[$\boldsymbol{n>0}$] This case goes exactly as the case $n>0$ above.
\end{description}
\item
$\overline\phi=\bigwedge_{i\in I}\overline\phi_i$. Then $p\models\overline\phi_i$ for all $i\in I$. By induction
  $q\models\overline\phi_i$ for all $i\in I$, so $q\models\bigwedge_{i\in I}\overline\phi_i$.
\item
$\overline\phi=\neg\overline\phi'$. Then $p\not\models\overline\phi'$. By induction
  $q\not\models\overline\phi'$, so $q\models\neg\overline\phi'$.
\item
  $\overline\phi=\diam\alpha\phi$ with $\phi\in\IO{b}^s$.
  Then $p \trans\alpha p'$ for some $p'$ with $p \models\phi$.
  By Def.~\ref{def:rbb}, $q \trans\alpha q'$ for some $q'$ with $p'\bis[s]{b}q'$.
  So by induction $q' \models\phi$, and hence $q\models\diam\alpha\phi$.
\item
  $\overline\phi=\phi \in \IO{b}^s$. Since $p\bis[s]{rb}q$ implies $p\bis[s]{b}q$, we obtain
  $q\models \phi$ by the cases treated above.
\end{itemize}

\vspace{2mm}
\noindent
($\Leftarrow$)
We first prove that $\sim_{\IO{b}^s}$ is a branching bisimulation. The relation is clearly symmetric. Let $p\sim_{\IO{b}^s}q$. Suppose \plat{$p\trans{\alpha}p'$}. If $\alpha=\tau$ and $p'\sim_{\IO{b}^s}q$, then the first condition of Def.~\ref{def:bb} is fulfilled. So we can assume that either (i) $\alpha\neq\tau$ or (ii) $p'\not\sim_{\IO{b}^s}q$. We define two sets:
\[
\begin{array}{lcl}
Q' &=& \{q'\in\mathbb{P}\mid q\epsarrow q'\land p\not\sim_{\IO{b}^s}q'\} \\
Q'' &=& \{q''\in\mathbb{P}\mid\exists q'\in\mathbb{P}:q\epsarrow q'\trans\alpha q''\land p'\not\sim_{\IO{b}^s}q''\}
\end{array}
\]
For each $q'\in Q'$, let $\phi_{q'}$ be a formula in $\IO{b}^s$ such that $p\models\phi_{q'}$ and $q'\not\models\phi_{q'}$. (Such a formula always exists because $\IO{b}^s$ is closed under negation $\neg$.) We define
\[
\phi=\bigwedge_{q'\in Q'}\phi_{q'}
\]
Similarly, for each $q''\in Q''$, let $\psi_{q''}$ be a formula in $\IO{b}^s$ such that $p'\models\psi_{q''}$ and $q''\not\models\psi_{q''}$. We define
\[
\psi=\bigwedge_{q''\in Q''}\psi_{q''}
\]
Clearly, $\phi,\psi\in\IO{b}^s$, $p\models\phi$ and $p'\models\psi$. We distinguish two cases.
\begin{enumerate}
\item
$\alpha\neq\tau$. Since $p\models\eps\phi\diam\alpha\psi\in\!\IO{b}^s$ and $p\!\sim_{\IO{b}^s}\! q$, also $q\models\eps\phi\diam\alpha\psi$. Hence \plat{$q\epsarrow q'\trans\alpha q''$} with $q'\models\phi$ and $q''\models\psi$. By the definition of $\phi$ and $\psi$ it follows that $p\sim_{\IO{b}^s}q'$ and $p'\sim_{\IO{b}^s}q''$.
\item
$\alpha=\tau$ and $p'\not\sim_{\IO{b}^s}q$. Let $\tilde\phi\in\IO{b}^s$ such that $p'\models\tilde\phi$ and $p,q\not\models\tilde\phi$. Since $p\models\eps\phi\diam{\hat\tau}(\tilde\phi\land\psi)\in\IO{b}^s$ and $p\sim_{\IO{b}^s}q$, also $q\models\eps\phi\diam{\hat\tau}(\tilde\phi\land\psi)$. So \plat{$q\epsarrow q'$} with $q'\models\phi\diam{\hat\tau}(\tilde\phi\land\psi)$. By definition of $\phi$ it follows that $p\sim_{\IO{b}^s}q'$. Thus $q'\not\models\tilde\phi$, so $q'\trans\tau q''$ with $q''\models\tilde\phi\land\psi$. By the definition of $\psi$ it follows that $p'\sim_{\IO{b}^s}q''$.
\end{enumerate}
Both cases imply that the first condition of Def.~\ref{def:bb} is fulfilled, i.e.\
that $\sim_{\IO{b}^s}$ is a branching bisimulation.
Next we show that $\sim_{\IO{b}^s}$ is stability-respecting.
Let $p\sim_{\IO{b}^s}q$ and $p\ntrans\tau$. Define $Q'$ and $\phi\in\IO{b}^s\subseteq\IO{rb}^s$ as above.
Then $p\models \eps(\neg\diam{\tau}\top\land\,\phi)$, and thus also 
$q\models \eps(\neg\diam{\tau}\top\land\,\phi)$.
Hence \plat{$q\epsarrow q'$} for some $q'$ with $q'\models\neg\diam{\tau}\top\land\,\phi$.
So $q\ntrans\tau$ and by the definition of $\phi$ it follows that $p\sim_{\IO{b}^s}q'$,
which had to be shown.

Finally we show that $\sim_{\IO{rb}^s}$ is a rooted stability-respecting branching bisimulation.
Let $p\sim_{\IO{rb}^s}q$ and \plat{$p\trans\alpha p'$}.
Let \plat{$Q_\alpha = \{q'\in\mathbb{P}\mid q\trans\alpha q'\land p'\not\sim_{\IO{b}^s}q'\}$}.
For each $q'\in Q'$, let $\chi_{q'}$ be a formula in $\IO{b}^s$ such that $p'\models\chi_{q'}$ and
$q'\not\models\chi_{q'}$. We define $\chi=\bigwedge_{q'\in Q'}\chi_{q'}$.
Since $p \models \diam\alpha\chi$ we have $q \models \diam\alpha\chi$,
so \plat{$q \trans\alpha q'$} for some $q'$ with $q'\models \chi$.
By the definition of $\chi$ it follows that $p'\sim_{\IO{b}^s}q'$,
which had to be shown.
\end{proof}

\end{document}